\documentclass[journal ]{new-aiaa}
\usepackage[utf8]{inputenc}
\usepackage{textcomp}
\usepackage{subcaption}
\usepackage{booktabs}  
\usepackage{multirow}  
\usepackage{hyperref}
\usepackage[section]{placeins} 
\usepackage{array}

\usepackage{amsthm}
\theoremstyle{plain}
\newtheorem{remark}{Remark}
\newtheorem{problem}{Problem}

\newtheorem{lemma}{Lemma}
\newtheorem{proposition}{Proposition}

\newtheorem{theorem}{Theorem}

\usepackage{xltabular}  
\usepackage{booktabs}
\usepackage{makecell}
\usepackage{array}
\usepackage{enumitem}
\usepackage{bm}
\usepackage{algorithm}
\usepackage{algpseudocode}
\usepackage{todonotes}

\newcolumntype{Y}{>{\raggedright\arraybackslash}X}
\newcolumntype{M}{>{\raggedright\arraybackslash}p{2.0cm}}

\newcommand{\rv}[1]{#1}

\usepackage{graphicx}
\usepackage{amsmath}
\usepackage[version=4]{mhchem}
\usepackage{siunitx}
\usepackage{longtable,tabularx}
\makeatletter
\let\oldthebibliography\thebibliography
\let\endoldthebibliography\endthebibliography

\renewenvironment{thebibliography}[1]{%
  \oldthebibliography{#1}%
  \setlength{\itemsep}{0pt}%
  \setlength{\parsep}{0pt}%
  \setlength{\parskip}{0pt}%
}{%
  \endoldthebibliography
}
\makeatother

\title{Physics-informed Learning for Orbital Uncertainty Propagation with Error Bounds}

\author{Chun-Wei~Kong\footnote{Graduate Research Assistant, Department of Aerospace Engineering Sciences. Corresponding Author: chun-wei.kong@colorado.edu.},
Morteza~Lahijanian\footnote{Associate Professor, Department of Aerospace Engineering Sciences: morteza.lahijanian@colorado.edu.},
and Jay~McMahon\footnote{Associate Professor, Department of Aerospace Engineering Sciences: jay.mcmahon@colorado.edu.}%
}

\affil{University of Colorado, Boulder, Colorado 80301}

\begin{document}

\maketitle

\begin{abstract}
The Fokker-Planck partial differential equation (FP-PDE) governs uncertainty evolution in stochastic dynamical systems. In orbital dynamics, solving the FP-PDE is challenging because of nonlinear motion, high-dimensional states, and large space-time domains. We develop a physics-informed neural network (PINN) approach that approximates the FP-PDE solution as a single space-time probability density, while also quantifying its worst-case approximation error. This approach is, in principle, independent of the choice of state coordinates and neural network architecture. Specifically, to enforce probability density function (PDF) properties into the neural network, we design a Physics-informed Gaussian mixture model (PINN-GMM). Then a companion error PINN learns the dynamics of the approximation error and yields time-dependent bounds that define an ambiguity set of PDFs. This ambiguity set enables rigorous computation of upper and lower bounds on event probabilities through tractable linear programs. Numerical studies on illustrative 1D examples and several 4D--6D orbital test cases demonstrate accurate uncertainty propagation, correct and informative error bounds, and improved reliability over common uncertainty-propagation baseline methods (Gaussian approximation, unscented transform, and Gaussian mixture model). Constructing the PINN-GMM requires offline training, making it costlier than the baseline approximations; once trained, however, a single forward pass returns the density at any time in sub-millisecond time $(0.16~\mathrm{ms}$ in our implementation).

\end{abstract}




\section{Introduction}\label{sec:intro}
Uncertainty propagation of physical systems over continuous time can be modeled by a stochastic differential equation (SDE).
This model describes the \emph{exact} time evolution of uncertainty as a probability density function, which is governed by the Fokker-Planck partial differential equation (FP-PDE)~\cite[Ch. 5]{sarkka2019applied}.
Nevertheless, FP-PDE does not admit closed-form solutions for general physical systems.
Recent advances in physics-informed neural networks (PINNs) point toward an efficient and accurate framework for learning solutions of general PDEs~\cite{lu2021deepxde,sirignano2018dgm}. 
\rv{A PINN trains a neural approximation by minimizing a physics-informed loss that combines initial/boundary-condition mismatch with the residual of the governing PDE evaluated at sampled collocation points.}
\rv{Thus, PINNs do not require labeled data from the true PDE solution, which is the probability density function (PDF) $p(x,t)$ in the FP-PDE setting; instead, the governing equation supplies physics-based constraints throughout the computational domain.}
This provides a new perspective on the \emph{exact} uncertainty-propagation problem: instead of propagating uncertainty forward in time, we seek to approximate the full space-time PDF as \emph{a single function} of space and time.
Its main advantage is that PINN can represent the uncertainty at \emph{any time} as an \emph{arbitrary} probability distribution, provided that the neural network has sufficient approximation capacity.

However, casting uncertainty propagation as learning the solution of the FP-PDE yields a substantially harder problem than traditional approaches.
First, training a PINN is inherently a non-convex optimization task, and the loss function induced by the FP-PDE typically creates a highly complex loss landscape~\cite{krishnapriyan2021characterizing}.
Second, existing analyses for PINNs give error estimates that grow exponentially over time, with exponents that themselves increase with the size of the state domain~\cite{mishra2023estimates}.
These effects are especially severe in orbital-dynamics applications, where uncertainty evolves over a large state space over long durations.
Third, although arbitrarily many space-time samples are available for PINN training in principle, one works with a finite number of samples in practice~\cite{sirignano2018dgm}. 
As the state dimension and computational domain grow, this finite \rv{sampling budget suffers from curse-of-dimensionality effects: samples become increasingly sparse over the computational domain, so relatively few points fall in the regions that contain PDF information.}
This \emph{information dilution} is particularly problematic in orbital dynamics, where probability mass often concentrates in small subsets of the state space.

Because directly solving or learning the FP-PDE at orbital scales is difficult, a large body of work instead seeks accurate and efficient uncertainty propagation \emph{without} explicitly solving the FP-PDE.
We group the existing approaches into (i) methods that propagate an initial PDF forward in time using structured approximations and (ii) methods that approximate the FP-PDE more directly through reduced or discretized representations.
For the first group, \rv{the state PDF is represented as a Gaussian mixture, with the mean and covariance of each component propagated using the deterministic drift and linearized dynamics~\cite{terejanu2008uncertainty}.}
To improve accuracy, \rv{the mixture weights are updated} using the FP-PDE, leading to a quadratic optimization problem that requires state-space integrals.
As a result, scalability to higher-dimensional systems is limited, and demonstrations are restricted to 1D and 2D examples.
\rv{Adaptive splitting and merging criteria are introduced to improve the accuracy of Gaussian-mixture uncertainty approximations~\cite{vishwajeet2018adaptive}.}
\rv{Higher-order state transition tensors are used to propagate Gaussian-mixture components and better capture the nonlinearity of orbital dynamics~\cite{park2006nonlinear,fujimoto2012analytical,khatri2023nonlinear}.}
These approaches are developed for dynamical systems without process noise, i.e., for the FP-PDE without the diffusion term. 
Similarly, \rv{higher-order moment propagation based on state-transition tensors is developed for systems without stochastic noise~\cite{acciarini2025nonlinear}.}
\rv{A Gaussian mixture is also propagated through a surrogate dynamical model constructed by polynomial chaos expansions, but this approach is restricted to deterministic discrete-time systems with an initial Gaussian distribution~\cite{vittaldev2016spacecraft}.}
Beyond these modeling restrictions, all of the above methods lack rigorous error bounds to quantify the discrepancy between the propagated PDF and the unknown true PDF.

A second line of work approximates the FP-PDE solution more directly. 
\rv{A Galerkin projection method represents the PDF as a tensor expansion in B-spline bases~\cite{acciarini2024uncertainty}.}
This approach has several limitations.
First, the number of tensor-product B-spline basis functions, and the cost of the associated state-space integrals, grows quickly with state dimension. 
Second, because the B-spline basis does not enforce non-negativity and normalization, PDF properties must be restored by post-processing, which may degrade solution quality. 
Third, due to the decoupled tensor expansion, the formulation is restricted to SDEs whose drift and diffusion do not explicitly depend on time. 
Finally, as with the traditional methods above, no \textit{a priori} error bounds are provided.
In numerical studies, the method is tested on a 2D stochastic damped harmonic oscillator and on 3D orbital cases in modified equinoctial elements.
Accuracy is assessed against high-fidelity Monte Carlo simulations using moment errors and the Hellinger distance, but no comparisons are made to traditional uncertainty-propagation methods.
\rv{A tensor-decomposition method represents the space-time PDF in tensor form discretized with Chebyshev spectral differentiation to solve the FP-PDE~\cite{sun2016uncertainty}.}
In this formulation, the number of tensor grid points grows with both state dimension and desired resolution, so runtime and memory requirements increase rapidly. 
They demonstrate the method on 4D and 6D two-body problems, both unperturbed and perturbed, in rotating spherical coordinates and compare against high-fidelity Monte Carlo simulations.
However, they do not compare against traditional uncertainty-propagation methods, and their validation focuses on errors in low-dimensional marginals rather than the full joint PDF. 
Again, no \emph{a priori} error bound is available to certify the quality of the approximation. 
Overall, existing methods either propagate restricted parametric families without rigorous \emph{a priori} error bounds, or approximate the FP-PDE with representations whose computational cost and constraint handling become limiting in high dimensions \rv{over extended time horizons}, again without providing rigorous error guarantees.
These gaps motivate the approach developed in this paper.

In this work, we develop theory and practical methods to learn orbital uncertainty propagation as a single space-time PDF governed by the FP-PDE, together with worst-case error bounds for the learned approximation.
First, we introduce several coordinate designs that reduce the effective computational domain for PINNs.
These range from special analytical ones (e.g., the normalized rotating spherical coordinates~\citep{sun2016uncertainty}) to more general mappings.
Next, we outline neural networks architectures for representing a space-time probability density function in PINN training.
We focus on a time-conditioned Gaussian-mixture PINN (PINN-GMM), which enforces PDF properties by construction and admits analytical expressions for marginal densities.
We further propose an encoder-decoder structure for the PINN-GMM to learn the PDF in a latent space.
To illustrate how the PINN approximation and its worst-case error bounds can be used in downstream tasks, we consider the problem of bounding event probabilities.
This problem aims to compute upper and lower bounds on the probability measure of an event when the underlying (unknown) probability density is governed by the FP-PDE.
We formulate this problem as a linear program (LP) with interval arithmetic, and show that when a PINN-GMM represents the approximate PDF, the resulting interval computations reduce to convex optimization problems.
Finally, we develop a training procedure for learning approximate space-time PDFs and their worst-case error bounds.
We propose a combined sampling strategy that mixes random uniform samples with an adaptive component to mitigate the \emph{information dilution} characteristic of orbital-dynamics problems.
\rv{We emphasize that these choices are problem-structured mitigation strategies for learning time-varying PDFs; they do not remove the curse of dimensionality for general PINNs or arbitrary high-dimensional PDEs.}

In short, our main contributions are:
\begin{itemize}
    \item a formulation of orbital uncertainty propagation as learning a single space-time probability density function governed by the FP-PDE, together with a \rv{two-stage physics-informed procedure that trains a density approximation $\hat p(x,t)$ and an error approximation $\hat e(x,t)$ to construct} worst-case error bounds,
    \rv{\item a time-conditioned PINN-GMM representation, with an encoder-decoder architecture, that yields a valid PDF at every time in the continuous horizon and improves trainability for hours-scale, large-state-space orbital problems,}
    \item a computationally tractable \rv{LP} that computes rigorous upper and lower bounds on event probabilities via \emph{PINN-GMM} and its error bounds,
    \item a practical training procedure, including a combined random-adaptive sampling strategy to handle PDF concentration over the vast space-time domain \rv{for the systems of up to six dimensions considered here,} and
    \item several numerical studies on 1D examples and \rv{four} 4D--6D orbital cases that quantify accuracy, benchmark against existing methods, validate worst-case error bounds \rv{empirically}, and compute certified event-probability intervals.
\end{itemize}

\section{Problem}\label{sec:problem}
In this work, we study uncertainty propagation in orbital dynamics modeled by a stochastic differential equation (SDE)~\citep{terejanu2008uncertainty,fujimoto2012analytical,sun2016uncertainty,acciarini2024uncertainty}:
\begin{equation}
d\bm{x}(t) = f(\bm{x}(t),t)\, dt + g(\bm{x}(t),t)d\bm{w}(t),
\label{eq:sde_general}
\end{equation}
where $t \in \rv{\mathcal{T}} \subseteq \mathbb{R}_{\geq 0}$ is time, $\bm{x}(t) \in \rv{\mathcal{X}} \subseteq \mathbb{R}^n$ is the state at time $t$, and $\bm{w}(t) \in \mathbb{R}^m$ is a standard Brownian motion. 
The deterministic part of the dynamics is given by the vector field $f: \rv{\mathcal{X} \times \mathcal{T}} \rightarrow \mathbb{R}^n$ and the noise coupling by the matrix-valued function $g: \rv{\mathcal{X} \times \mathcal{T}} \rightarrow \mathbb{R}^{n\times m}$. For clarity, the $i$-th component of $f$ is denoted $f_i$, and the $(j,k)$-th entry of $g$ is denoted $g_{jk}$. We assume that $f$ and $g$ satisfy standard regularity conditions~\citep[Ch.~5]{oksendal2003stochastic}. 
For simplicity, we denote $t=0$ as the initial time unless further specified.
The initial state $\bm{x}(0)$ is a random variable drawn from a given initial probability density function (PDF) $p_0: \rv{\mathcal{X}} \rightarrow \mathbb{R}_{\geq 0}$, which is assumed smooth and bounded.

The evolution of the state uncertainty over time of Eq.~\eqref{eq:sde_general} is described by the time-dependent PDF $p(x,t)$, governed by the associated FP-PDE:
\begin{equation}    \label{eq:fp_pde}
    \frac{\partial p(x,t)}{\partial t} + \sum_{i=1}^n \frac{\partial}{\partial x_i} [f_i p(x,t)] -
    \frac{1}{2} \!\! \sum_{i=1,j=1}^n \!\! \frac{\partial ^2}{\partial x_i \partial x_j} \left[ \sum_{k=1}^m g_{ik}g_{jk}p(x,t) \right] = 0,
\end{equation}
subject to the initial condition
\begin{equation}
    p(x,0) = p_0(x) \qquad \forall x \in \mathcal{X}.
    \label{eq:fp_pde_init_cond}
\end{equation}
To simplify notation, define the differential operator $\mathcal{D}[\cdot]$ associated with the FP-PDE as
\begin{equation}\label{eq:fp_diff_opt}
    \mathcal{D}[\cdot]:= \frac{\partial}{\partial t}[\cdot] + \sum_{i=1}^n \frac{\partial}{\partial x_i} [f_i \cdot] - \frac{1}{2} \! \sum_{i,j=1}^n \frac{\partial ^2}{\partial x_i \partial x_j} \left[ \sum_{k=1}^m g_{ik}g_{jk} \cdot \right].
\end{equation}
Then, Eqs.~\eqref{eq:fp_pde} and \eqref{eq:fp_pde_init_cond} can be rewritten in a compact form  as
\begin{equation}
    \mathcal{D}[p(x,t)] = 0 \quad \text{ subject to } \quad p(x,0) = p_0(x).
    \label{eq:fp_pde_compact}
\end{equation}
Under the regularity assumptions on $f$ and $g$, the FP-PDE in Eq.~\eqref{eq:fp_pde_compact} is well-posed, i.e., it admits a unique smooth solution $p(x,t)$, which is a reasonable assumption for most physical systems.

\begin{figure}[htbp!]
    \centering
    \includegraphics[width=1.0\linewidth]{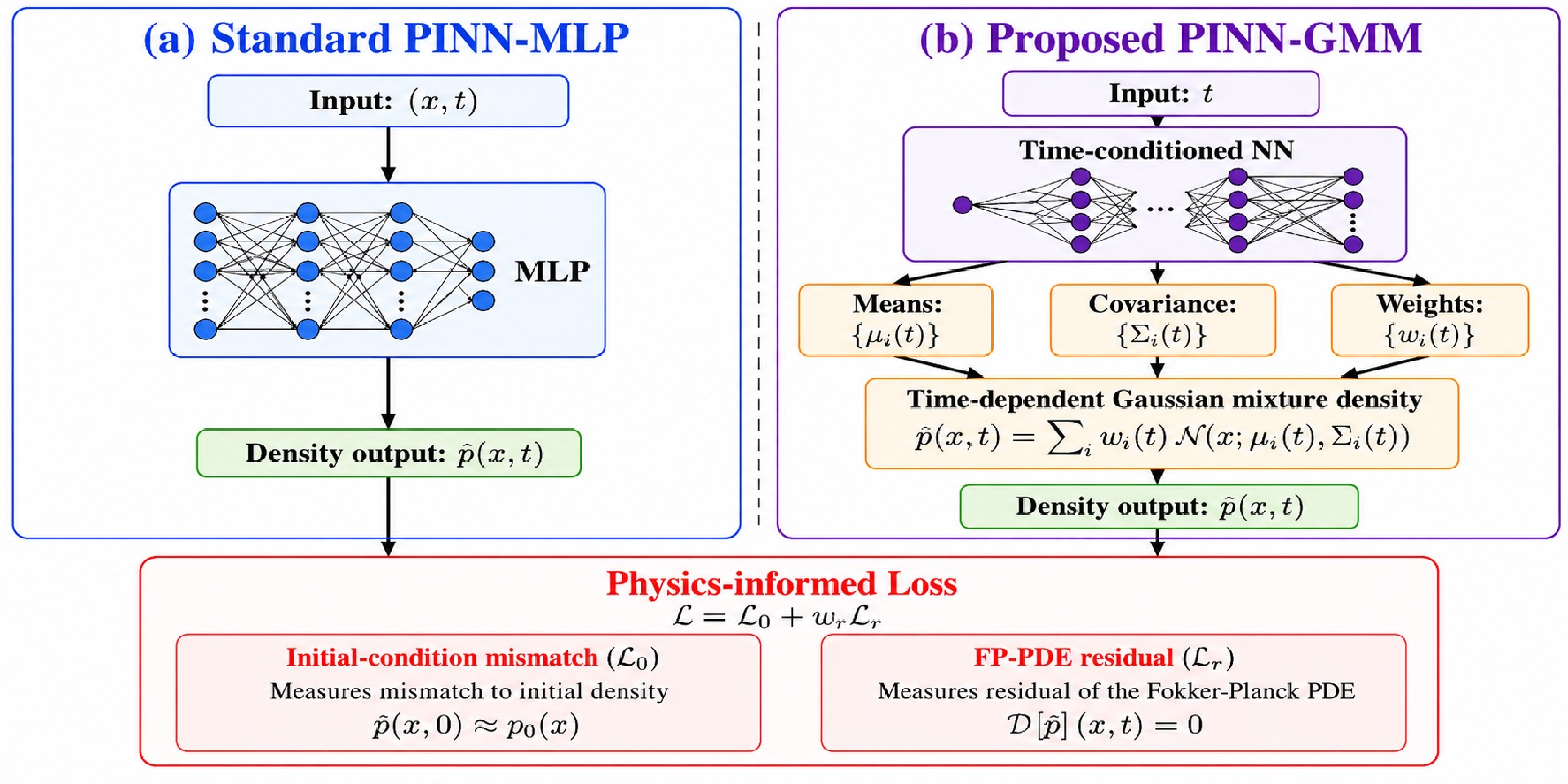}
    \caption{\rv{Physics-informed neural networks considered: a) standard PINN-MLP and b) proposed PINN-GMM, both trained with initial-condition and FP-PDE residual losses.}}
    \label{fig:pinn_concept}
\end{figure}

In general, Eq.~\eqref{eq:fp_pde_compact} cannot be solved in closed form.
We therefore approximate $p(x,t)$ with \rv{a physics-informed neural-network representation $\hat p(x,t)$.
Figure~\ref{fig:pinn_concept} contrasts the standard PINN-MLP representation~\cite{sirignano2018dgm,lu2021deepxde} with the proposed time-conditioned PINN-GMM representation; both are trained by minimizing a physics-informed loss derived from Eq.~\eqref{eq:fp_pde_compact}.}
\rv{Unlike data-driven methods that rely on finite samples of the true solution, this training uses only the FP-PDE operator $\mathcal{D}[\cdot]$ and the initial condition $p_0$, requiring no labeled data for $p(x,t)$ at $t>0$.}
Nevertheless, as with any learning-based solution, quantifying the error remains critical.
We therefore pose the following problem.

\begin{problem}\label{prob:1}
    Given the FP-PDE in Eq.~\eqref{eq:fp_pde_compact}\rv{, a compact state domain $X \subset \mathcal{X}$, and a compact time interval $T=[0,T_{\mathrm{end}}] \subset \mathcal{T}$, let $\Omega := X \times T$.}
    \rv{T}rain a PINN
    $\hat{p}(x,t)$ that approximates the solution $p(x,t)$ \rv{on $\Omega$}, and construct an error bound $B:T \rightarrow \mathbb{R}_{\geq 0}$ such that 
    \begin{equation}\label{eq:prob_1}
        \sup_{x\in X} |p(x,t)-\hat{p}(x,t)| \leq B(t),\; \forall t \in T.
    \end{equation} 
\end{problem}

\rv{Here, $X$ is a computational region of interest, not a finite PDE domain with prescribed artificial boundary conditions.
The FP-PDE is associated with the SDE on the full state space $\mathcal{X}$, where the PDF is assumed to be integrable and to decay at infinity.
The bound in Eq.~\eqref{eq:prob_1} is therefore an in-domain guarantee over $X \times T$; evaluating $\hat p$ outside this region, including time extrapolation beyond $T$, is out of scope.}
Although our prior work~\citep{kongerror} develops \rv{a} theoretical foundation for constructing such error bounds, training a PINN approximation $\hat p$ for orbital uncertainty propagation remains challenging in practice due to the large space-time domain (see the discussion in Sec.~\ref{sec:intro}), let \rv{alone} constructing the corresponding error bound $B$.

Furthermore, the neural-network $\hat p$ introduces an additional practical hurdle. In orbital conjunction analysis, one must integrate the PDF over an event set to compute the event probability. 
Since the true $p$ is unknown, we can only bound such event probability using $\hat p$ together with the error bound $B$, which naturally leads to integrals of $\hat p \pm B$. 
Computing such integrals is itself nontrivial when $\hat p$ is represented by a neural network. This motivates the following problem of computing upper and lower bounds on event probabilities. 
Formally,

\begin{problem}\label{prob:integral_pdf}
    Consider the setting in Problem~\ref{prob:1}, a bounded event set $X_{\text{event}} \subseteq X$.
    Define the event probability (probability mass) at time $t$ by
    $
        \mathbb{P}(x \in X_{\text{event}}, t) := \int_{x \in X_{\text{event}}} p(x,t) dx.
    $
    Given approximate PDF $\hat p(x,t)$ and bound $B(t)$,
    compute lower and upper probability bounds $\mathbb{P}^-(t;X_{\text{event}},\hat p, B)$ and $\mathbb{P}^+(t; X_{\text{event}},\hat p, B)$, respectively, such that
    $
        \mathbb{P}^-(t;X_{\text{event}},\hat p, B) \leq \mathbb{P}(x \in X_{\text{event}}, t) \leq \mathbb{P}^+(t;X_{\text{event}},\hat p, B).
    $
\end{problem}



\section{Methods}
This section presents the theoretical framework for approximating the space–time density 
$p(x,t)$ in orbital dynamics. 
We first describe alternative coordinate choices and how they affect domain size and dynamical structure.
We then introduce a representation of $\hat p(x,t)$ that encodes PDF properties.
Finally, we show how the resulting approximation, together with a time-dependent error bound, enables rigorous downstream analyses; in particular, we compute upper and lower bounds on event probabilities (i.e., integrals of the PDF over event regions) via a linear program.

\subsection{Choice of Coordinates}
The choice of state coordinates $\bm{x}$ in Eq.~\eqref{eq:sde_general} affects both the domain of $p(x,t)$ and the structure of the dynamics.
We first consider problem-specific coordinates and then introduce more general coordinates.

\subsubsection{Rotating Spherical Coordinates}\label{subsubsec:rsc}
We adopt standard spherical coordinates $(r,\theta,\phi)$, related to Cartesian position $(x,y,z)$ by
\begin{equation}\label{eq:def_sphere}
    r = \sqrt{x^2+y^2+z^2},\quad \theta = \cos^{-1}(z/r),\quad \phi=\tan^{-1}(y/\rv{x}),
\end{equation}
where $r,\theta,\phi$ are the radial distance, polar angle (from vertical) and azimuthal angle (from positive Eastward) respectively.
The rotating non-dimensional spherical coordinates $(r',\theta', \phi', t')$ are defined by the mapping between $(r,\theta,\phi,t)$ and $(r',\theta',\phi',t')$:
\begin{equation}\label{eq:def_rotsphere}
    r = Rr',\quad \theta = \Theta \theta',\quad \phi = \Phi \phi' + W(T_{\mathrm{orb}}t'),\quad t = T_{\mathrm{orb}} t',
\end{equation}
where $W$, $R$ and $T_{\mathrm{orb}}$ are the characteristic rotation rate, radius and period of the nominal orbit, respectively, and $\Phi$, $\Theta$ are scalars for scaling and non-dimensionalization. 
\rv{For completeness, Ref.~\cite{sun2016uncertainty} provides the explicit $J_2$-perturbed Keplerian \rv{stochastic differential equation (SDE)} in these coordinates.}

Propagating uncertainty in rotating spherical coordinates can reduce the computation domain~\citep{sun2016uncertainty}.
Besides, transforming back to Cartesian coordinates yields a closed-form expression for the Cartesian-space marginal PDF, i.e., $p_{XYZt}(x,y,z,t) = \frac{1}{R\Theta \Phi r^2 \sin \theta} p_{r'\theta' \phi' t'}(r',\theta',\phi',t')$,
where $\{x,y,z,t\}$ is computed from $\{r',\theta',\phi',t'\}$ through Eqs.~\eqref{eq:def_sphere} and~\eqref{eq:def_rotsphere}.
Note that this expression is singular when $\sin \theta = 0$. 
For instance, if the Cartesian coordinate is defined on the \rv{Earth-Centered Inertial} frame, then singularity occurs when spherical coordinates in question are over one of the poles.

\subsubsection{\rv{Orbital-Elements} Coordinates}\label{subsubsec:oec}
The state $\bm{x}$ may also be expressed in orbital elements (e.g., classical Keplerian or equinoctial elements).
\rv{Here, $a,e,i,\Omega,\omega,\theta$ denote semi-major axis, eccentricity, inclination, right ascension of the ascending node, argument of periapsis, and true anomaly, respectively; $p_{\mathrm{sl}}$ denotes semilatus rectum; $\lambda$ and $L$ denote mean and true longitude, respectively; and $(P_1,P_2)$ and $(Q_1,Q_2)$ are nonsingular equinoctial variables associated with eccentricity and inclination, respectively.}
Propagating uncertainty in the equinoctial orbit elements $\{a,P_1,P_2,Q_1,Q_2,\lambda\}$ can substantially reduce the computation domain for $p(x,t)$. 
In fact, for unperturbed Keplerian dynamics the system becomes decoupled, and nonlinearity appears only in $a$ and $\lambda$ \rv{(see Appendix~\ref{appendix:oe_dyn} or Ref.~\cite{khatri2023nonlinear} for more details)}.
However, a simple closed-form density transformation from these elements to Cartesian coordinates is generally not available, because it requires solving Kepler’s equation to convert mean anomaly to true anomaly.
As a result, the Jacobian and the induced PDF do not admit a straightforward analytic expression.

An alternative is the modified equinoctial set $\{p_{\mathrm{sl}},P_1,P_2,Q_1,Q_2,L\}$, which is related to the classical Keplerian orbit elements $\{a,e,i,\Omega,\omega,\theta \}$ through the true anomaly.
This choice admits compact expressions for common perturbations (e.g., $J_2$, aerodynamic drag) and provides a direct state mapping to Cartesian space, yielding analytical PDF transformation to Cartesian coordinates. 
More generally, one can construct generalized equinoctial element sets with smooth, nonsingular transformations to Cartesian position and velocity, and derive the associated Jacobian matrices explicitly~\citep{bau2021generalization}. 
These Jacobians can then be used to map probability densities between elements and Cartesian coordinates via the standard change-of-variables formula.

\subsubsection{Reference Coordinates}
Let $x_{\mathrm{ref}}(t)$ be a reference trajectory governed by $d(x_{\mathrm{ref}}) = f(x_{\mathrm{ref}})dt$, and define the deviation coordinates by
$
    \delta x(t) := x(t) - x_{\mathrm{ref}}(t).
$
The induced SDE for $\delta x(t)$ 
is
\begin{equation}\label{eq:SDE_ref}
    d(\delta x(t)) = \Big[ f(\delta x(t) + x_{\mathrm{ref}}(t), t) -  f(x_{\mathrm{ref}}(t),t) \Big] dt + g \Big( \delta x(t) + x_{\mathrm{ref}}(t), t \Big) dw.
\end{equation}
To keep the notation light, we drop the explicit time argument $(t)$ and write variables without time dependence unless the context is unclear.
Working in deviation coordinates shifts the learning task from $p(x,t)$ to $p(\delta x,t)$, which is advantageous when $x$ remains concentrated near $x_{\mathrm{ref}}$.
This typically (i) shrinks the computational domain, since the support of $p(\delta x, t)$ is localized near the origin, and (ii) simplifies the local dynamics, i.e.,
$
    f(\delta x + x_{\mathrm{ref}},t) - f(x_{\mathrm{ref}},t) \approx \frac{\partial f}{\partial x} |_{x_{\mathrm{ref}}}
$
for small $\delta x$.
In fact, the rotating spherical coordinates discussed above can be viewed as a specific choice of reference coordinates, combined with normalization scaling factors.

\subsubsection{Generalized Mapped Coordinates}\label{subsec:general_coord}
We introduce a generalized coordinate framework for the class of smooth, invertible transformations.
Let $z = \Phi(x,t) \in \mathbb{R}^{n}$ be a latent state \rv{with the same dimension as $x$}, where for each fixed $t$ the map $\Phi(\cdot,t)$ is smooth and invertible.
Let the original SDE be
$
    dx = f(x,t)\,dt + G(x,t)\,dW,
$
and define $A(x,t) := G(x,t)G(x,t)^{\mathsf T} \in \mathbb{R}^{n \times n}$.
By Itô’s lemma~\citep{oksendal2003stochastic}, the latent process $z(t) = \Phi(x(t),t)$ satisfies an SDE of the form
$
    dz = f_z(z,t)\,dt + G_z(z,t)\,dW,
$
with components
\[
\begin{aligned}
    & dz_i
    = f_{z,i}(z,t)\,dt
       + \sum_{j=1}^m [G_z(z,t)]_{ij}\,dW_j, \quad
    f_{z,i}(z,t)
    = \Bigg[
          \frac{\partial \Phi_i}{\partial t}
          + \sum_{k=1}^{n} \frac{\partial \Phi_i}{\partial x_k} f_k
          + \frac{1}{2} \sum_{k=1}^{n} \sum_{l=1}^{n}
              A_{kl} \frac{\partial^2 \Phi_i}{\partial x_k \partial x_l}
       \Bigg]_{x = \Phi^{-1}(z,t)}, \\
    & [G_z(z,t)]_{ij}
    = \Bigg[
          \sum_{k=1}^{n} \frac{\partial \Phi_i}{\partial x_k} G_{kj}
       \Bigg]_{x = \Phi^{-1}(z,t)}, \quad
    A_{kl}(x,t)=\sum_{j=1}^m G_{kj}(x,t)\,G_{lj}(x,t).
\end{aligned}
\]
The corresponding FP-PDE for the latent-space PDF $p(z,t)$ is
\[
    \frac{\partial p(z,t)}{\partial t}
    + \sum_{i=1}^{n}
        \frac{\partial}{\partial z_i}
        \big[f_{z,i}(z,t)\,p(z,t)\big]
    - \frac{1}{2} \sum_{i=1}^{n} \sum_{j=1}^{n}
        \frac{\partial^2}{\partial z_i \partial z_j}
        \big[A_{z,ij}(z,t)\,p(z,t)\big] = 0,
\]
where the diffusion tensor in $z$-coordinates is
$
    A_{z,ij}(z,t)
    = \Bigg[
        \sum_{k=1}^{n} \sum_{l=1}^{n}
        \frac{\partial \Phi_i}{\partial x_k}\,
        A_{kl}(x,t)\,
        \frac{\partial \Phi_j}{\partial x_l}
      \Bigg]_{x = \Phi^{-1}(z,t)}.$
\rv{As a special case, the affine mapping $z=x-x_{\mathrm{ref}}$ recovers the reference-coordinate SDE of Eq.~\eqref{eq:SDE_ref}.}

The map $\Phi$ can be specified analytically, semi-analytically with learnable parameters or learned entirely by neural networks.
Recent work~\citep{pmlr-v288-peper25a} provides principles for learning physically interpretable mappings, yielding latent variables aligned with the underlying mechanics.
When combined with a space–time density representation (next subsection), this leads to an encoder--latent-dynamics--decoder architecture for the physics-informed neural network $\hat p(x,t)$.

\subsection{Space-time Probability Density Representation}\label{subsec:pdf_representation}
We briefly review the standard multi-layer perceptron (MLP) parameterization of the space-time density $\hat p(x,t)$ commonly used in physics-informed learning.
We then introduce a physics-informed neural-network-based Gaussian mixture model (PINN-GMM) as an alternative representation, outline its desirable properties, and propose a concrete architecture motivated by generalized mapped coordinates.

\subsubsection{Standard Multilayer Perceptron (MLP)}
A fully connected multilayer perceptron (MLP) takes $(x,t)\in\mathbb{R}^{n+1}$ as input and produces a scalar output for the space-time density.
Let the layer widths be $d_0, d_1,\dots, d_L$, with $d_0=n+1$ matching the input dimension $(x,t)$ and $d_L=1$ corresponding to the scalar probability density output.
For each layer $\ell=1,\dots,L$, let $W_{\ell}\in\mathbb{R}^{d_\ell\times d_{\ell-1}}$ and $b_\ell\in\mathbb{R}^{d_\ell}$, and define the affine map
$
    h_\ell(z)=W_\ell z+b_\ell,
$
with element-wise activation $\sigma_\ell:\mathbb{R}\to\mathbb{R}$.
With neural network parameters $\theta=\{(W_\ell,b_\ell)\}_{\ell=1}^L$, the MLP representation of $\hat p$ is
$
    \hat p(x,t;\theta)
    = \big(\sigma_L\circ h_L\circ \sigma_{L-1}\circ h_{L-1}\circ \cdots \circ \sigma_1\circ h_1\big)(x,t).
$
To enforce $\hat p \geq 0$, the output activation $\sigma_L$ is chosen non-negative, e.g., $\sigma_L=\operatorname{Softplus}$.

\subsubsection{Physics-Informed Gaussian Mixture Model (PINN-GMM)}\label{subsec:pinn_gmm_model}
The standard MLP representation above does not enforce the normalization constraint $\int_X \hat p(x,t) dx = 1$.
To enforce this constraint by construction, we introduce a time-conditioned physics-informed Gaussian mixture model
\begin{equation}
\label{eq:phat_GMM}
    \hat{p}(x,t) = \sum_{i=1}^{K_{\mathrm{GMM}}} w_i(t) \mathcal{N}(x\rv{;} \mu_i(t), \Sigma_i(t)).
\end{equation}
\rv{The time-dependent mixture parameters are specified below as neural-network outputs of time, yielding a valid PDF of $x$ at any time $t$.}
\rv{The number of components $K_{\mathrm{GMM}}$ determines the expressiveness of this density representation. Smaller $K_{\mathrm{GMM}}$ yields a more compact, cheaper-to-train model; larger $K_{\mathrm{GMM}}$ captures more complex structure at higher cost, and also increases the training cost of the associated error-PINN $\hat{e}$. We do not claim a systematic rule for selecting $K_{\mathrm{GMM}}$; deriving one is a nontrivial open problem that depends on the complexity of the true FP-PDE solution, the state dimension $n$, and the propagation horizon. We therefore treat $K_{\mathrm{GMM}}$ as a practical accuracy--cost trade-off: start small, inspect the resulting error bound, and increase $K_{\mathrm{GMM}}$ only if the bound is not sufficiently informative for the application.}
\rv{
We apply this guideline in all orbital experiments, yielding $K_{\mathrm{GMM}}$ = 5 (equinoctial orbit-element case) and $K_{\mathrm{GMM}}$ = 11 (rotating-spherical cases); the detailed physical reasoning is explained with the experiment setting in Appendix~\ref{appendix:6d_equin}.
}
\rv{In addition, sampling for PINN-GMM} is straightforward, 
and any marginal density has a closed-form expression---features not available for standard MLP representations.

\paragraph{Time–conditioned Gaussian Mixture Architecture}
We design PINN-GMM \(\hat p(x,t):\mathbb{R}^n \times \mathbb{R} \rightarrow \mathbb{R}_{\geq 0}\) that respects mixture constraints and remains numerically stable under disparate state scales. 
The state is whitened by a fixed encoder 
$z=\Phi(x)=\Sigma_{\mathrm{en}}^{-1}(x-\mu_{\mathrm{en}}),\; 
\Sigma_{\mathrm{en}}=\mathrm{diag}(\sigma_{\mathrm{en}}),\ \sigma_{\mathrm{en}}\in\mathbb{R}^n_{>0}$,
where $\mathrm{diag}(s)$ denotes the diagonal matrix with entries of $s$ on its diagonal.
In this work, the encoder parameter is set using the initial mean and covariance of the state, i.e.,  $\mu_{\mathrm{en}}=\mu_{t_0}$ and $\Sigma_{\mathrm{en}} = \mathrm{diag}(\Sigma_{t_0})$.
The density is learned in latent space, \(\hat p_Z(z,t)\), and decoded by the change of variables
$
\hat p(x,t)=\hat p_Z\!\big(\Sigma^{-1}_{\mathrm{en}}(x-\mu_{\mathrm{en}}),t\big)\prod_{i=1}^n \sigma_{\mathrm{en},i}^{-1}.
$
\rv{The time-dependent mixture parameters are generated by differentiable neural networks, implemented either as separate networks or as output heads sharing a common feature network:}
(i) a mean \rv{output} \(\mu_z(t; \theta_{\mu_z})\in\mathbb{R}^{K_{\mathrm{GMM}}\times n}\);
(ii) a covariance \rv{output} \(\ell_z(t; \theta_{L_z})\in\mathbb{R}^{K_{\mathrm{GMM}}\times \frac{n(n+1)}{2}}\), assembled into a lower-triangular Cholesky factor \(L_z(t)\in\mathbb{R}^{K_{\mathrm{GMM}}\times n\times n}\) with positive diagonals enforced by \(\mathrm{softplus}+\varepsilon_{\min}\)\footnote{$\varepsilon_{\min}$ is set to $10^{-4}$ in our experiments};
and (iii) a \rv{weight-logit output} \(v(t;\theta_{v})\in\mathbb{R}^{K_{\mathrm{GMM}}}\).
\rv{These logits are not the mixture weights directly; they are converted to the weights in Eq.~\eqref{eq:phat_GMM} by}
\[
\rv{
w_i(t)
=(1-\alpha)\frac{\exp(v_i(t))}{\sum_{j=1}^{K_{\mathrm{GMM}}}\exp(v_j(t))}
+\frac{\alpha}{K_{\mathrm{GMM}}},
\quad i=1,\ldots,K_{\mathrm{GMM}} .
}
\]
\rv{Thus, \(w_i(t)\ge \alpha/K_{\mathrm{GMM}}\) and \(\sum_{i=1}^{K_{\mathrm{GMM}}}w_i(t)=1\).}
Here \(\alpha\in(0,1)\) is a small mixing coefficient\footnote{$\alpha$ is set to $0.01$ in our experiments} \rv{that mitigates weight collapse by keeping every mixture component active.}
For any given normalized time $t$, the corresponding \(x\)–space Gaussian mixture parameters can be queried by decoding, i.e., for all $k \in K_{\mathrm{GMM}}$,
$
\mu_{x,k}(t) = \mu_{\mathrm{en}} + \mu_{z,k}(t)\odot\sigma_{\mathrm{en}},\;
\Sigma_{x,k}(t)= L_{x,k}(t)\,L_{x,k}(t)^\top$ and $ 
L_{x,k}(t) = \mathrm{diag}(\sigma_{\mathrm{en}})\,L_{z,k}(t).
$

\paragraph{Key Difference between PINN-GMM and Traditional GMM Propagation}
\textbf{PINN-GMM} (this work) learns a time-conditioned GMM $\hat p(x,t)$ that directly approximate the full space-time density $p(x,t)$ of the \emph{exact} dynamics governed by the FP-PDE.
In contrast, \textbf{traditional GMM propagation} starts from an initial GMM and propagates each mixture through dynamics (typically \emph{approximate}) dynamics, updating the mean and covariances parameters.
In most cases~\citep{park2006nonlinear,fujimoto2012analytical,khatri2023nonlinear}, the weights are held fixed unless additional information is introduced (e.g., observations, splitting criteria).
\subsection{Error Bound Construction and Application}\label{sec:error_bound_application}
Above we introduce different constructions for representing the PDF approximation $\hat p$ with a PINN.
Here, we summarize a key result from our previous work~\citep{kongerror} that bounds the approximation error of PINN $\hat p$, yielding an admissible PDF set.
In the remainder of this section, we specialize this construction to the admissible PDF set induced by PINN-GMM. 
We then focus on Problem~\ref{prob:integral_pdf} to compute upper and lower bounds on event probabilities (i.e., integrals of the PDF) via this admissible set.

Let $\hat p$ be a learned PINN approximation solution to the true PDF \(p(x,t)\).
Define the point-wise approximation error
$
    e(x,t) := p(x,t) - \hat p (x,t)
$
Then $e(x,t)$ is governed by another PDE:
\begin{equation}\label{eq:error_pde}
    \mathcal{D}[e] + \mathcal{D}[\hat{p}] = 0
     \text{ s.t. } 
    e(x,0) = p_0(x) - \hat{p}(x,0).
\end{equation}
\rv{
The error PDE in Eq.~\eqref{eq:error_pde} is not necessarily easier to solve than the original FP-PDE in Eq.~\eqref{eq:fp_pde_compact}.
It has the same differential operator and is forced by the residual $\mathcal{D}[\hat p]$, so a large or oscillatory residual can make error-PINN training difficult.
To address this, we use a smooth PINN-GMM representation for $\hat p$ together with the coordinate/encoder choices introduced above, which aim to keep $\mathcal{D}[\hat p]$ smooth and small, and we train $\hat e$ on informative samples via the combined sampling of Sec.~\ref{sec:train_implementation}; Sec.~\ref{sec:exp} provides empirical evidence that the resulting error-learning step is tractable on the orbital benchmarks considered here.}

Our prior work~\citep{kongerror} showed that, under some sufficient conditions, one can construct a time-conditioned error bound that solves Problem~\ref{prob:1} by training another error-PINN $\hat e(x,t)$ that approximates the solution to the PDE in Eq.~\eqref{eq:error_pde}.
Specifically, the time-dependent error bound takes the first-order form proposed in Ref.~\cite{kongerror}:
\rv{\begin{equation}\label{eq:first_order_error_bound}
  B(t) = 2 \cdot \max_{x\in X} |\hat e(x,t)|,\quad \forall t \in T.
\end{equation}}
\rv{The factor $2$ in Eq.~\eqref{eq:first_order_error_bound} is not a heuristic safety factor; it follows from Ref.~\cite[Cor.~1]{kongerror} under a sufficient accuracy condition on the error-PINN $\hat e$.
A conservative checking condition is given in Ref.~\cite[Prop.~1]{kongerror} but is not evaluated here, as it depends on problem-dependent constants that are difficult to estimate tightly in high-dimensional orbital settings.
Sec.~\ref{sec:exp} empirically validates the constructed bounds against high-fidelity reference solutions.
We note that if Eq.~\eqref{eq:first_order_error_bound} yields a $B(t)$ satisfying Eq.~\eqref{eq:prob_1}---as theoretically grounded by Ref.~\cite[Cor.~1]{kongerror} and empirically supported in Sec.~\ref{sec:exp}---then all subsequent claims, including the admissible PDF set in Eq.~\eqref{eq:set_of_PDFs} and the LP-based event-probability bounds in Proposition~1, follow rigorously.} 

\rv{Formally, let $B(t)$ be constructed by Eq.~\eqref{eq:first_order_error_bound} such that Eq.~\eqref{eq:prob_1} holds, then $\hat p(x,t)$ and $B(t)$} defines a set of admissible PDFs
\begin{equation}\label{eq:set_of_PDFs}
\mathcal{P}(\hat p,B,t)
:=
\Bigl\{\, p(x,t)\ \Big|\ 
\forall x\in X \subset \mathbb{R}^n:\ |p(x,t)-\hat p(x,t)|\le B(t),\ 
p(x,t)\ge 0,\ 
\int_{\mathbb{R}^n} p(x,t)\,dx = 1
\Bigr\}
\end{equation}
such that for every $t \in T$, the true PDF $p(x,t)$ is guaranteed to lie in this set, i.e.,
$
    p(x,t) \in \mathcal{P}(\hat p, B, t).
$

\rv{\begin{remark}
The bound $B(t)$ in Eq.~\eqref{eq:first_order_error_bound} is uniform in space: the first-order construction retains only the spatial maximum of $\hat e(x,t)$, discarding spatial variation and making the admissible set conservative. This form is used because it yields a simple ambiguity set compatible with the LP-based event-probability bounds in Proposition~\ref{proposition:lp}. Tighter spatially varying bounds are possible but require higher-order recursive error learning~\cite{kongerror}, which is not practical for the orbital examples here.
\end{remark}}

\begin{remark}
    The admissible PDF set $\mathcal{P}$ in Eq.~\eqref{eq:set_of_PDFs} is a type of ambiguity set.
    It corresponds to the \emph{band model}~\cite{kassam2003robust}, the \emph{density ratio class}~\cite{berger1990robust,rinderknecht2011eliciting}, and the \emph{intervals of measures}~\cite{deroberts1981bayesian}.
    Moreover, it can be translated into an $\epsilon$-contamination model with bounded outlier distribution~\citep{fauss2016old}.
\end{remark}

Several applications of $\mathcal{P}$ are discussed in Refs.~\cite{berger1990robust,rinderknecht2011eliciting,sriwastava2025sensitivity} and references therein.
Here we focus on Problem~\ref{prob:integral_pdf}, which seeks upper and lower bounds on the probability of an event, such as entering an unsafe region or reaching a target set.
Following the optimization approach used in formal methods~\cite{lahijanian2015formal}, we formulate an LP to solve Problem~\ref{prob:integral_pdf}.

\begin{proposition}[Solving Problem~\ref{prob:integral_pdf} by LP]\label{proposition:lp}
    Given $t \in T$, a bounded set $X \subset \mathbb{R}^n$, and a event set $X_{\mathrm{event}} \subseteq X$.
    Let $C = \{c_1, \ldots, c_N$\}, where $c_i \subseteq X$, be a partition of $X$,
    i.e., $\cup_{i=1}^N c_i = X, c_i \cap c_j = \emptyset$ 
    for all $i\neq j \in \{1,\ldots,N\}$.
    Let $\hat p$ and $B$ be given such that Eq.~\eqref{eq:set_of_PDFs} holds.
    Define computable lower and upper relaxations $\underline p_i,\overline p_i$ satisfying
    \begin{equation}\label{eq:relaxed_cell_bounds}
    \underline p_i \le \min_{x\in c_i}\hat p(x,t)-B(t),\;
    \overline p_i \ge \max_{x\in c_i}\hat p(x,t)+B(t). 
    \end{equation}
    and define
    $
    p_i^- := \max\{0,\underline p_i\}$ and $ p_i^+ := \overline p_i.
    $
    Let $\mathrm{Vol}(c_i)$ be the volume of cell $c_i \in C$.
    Then the upper bound of event probability 
    $\mathbb{P}^+(t; X_{\text{event}}) \geq \mathbb{P}(x \in X_{\text{event}}, t)$ can be obtained by solving:
    \begin{subequations}
        \begin{align}
            & \mathbb{P}^+(t; X_{\text{event}}) = \max_{p_i} \sum_{c_i \in C \, \land \, c_i \cap X_{\text{event}} \neq \emptyset } p_i \mathrm{Vol}(c_i) 
            \\
            & \text{ subject to} \nonumber \\
            & \qquad \ p_i \in [p_i^-, p_i^+] \qquad \forall i\in \{1,\ldots, N\}
            \qquad \sum_{i=1}^N p_i  \mathrm{Vol}(c_i) \leq 1 \label{eq:lp_constraints}.
        \end{align}
    \end{subequations}
    
    Similarly, the lower bound of event probability $\mathbb{P}^-(t; X_{\text{event}}) \leq \mathbb{P}(x \in X_{\text{event}}, t)$ can be obtained by solving $
        \mathbb{P}^-(t; X_{\text{event}})  = \min_{p_i} \sum_{c_i \in C \, \land \, c_i \subseteq X_{\text{event}}} p_i\mathrm{Vol}(c_i)$ subject to Eq.~\eqref{eq:lp_constraints}.
\end{proposition}
\begin{proof}
    Fix $t$ and let $p(\cdot,t)\in\mathcal{P}(\hat p,B,t)$. Define the cell averages
    \begin{equation}\label{eq:avg_pi_def}
        p_i := \frac{1}{\mathrm{Vol}(c_i)}\int_{c_i} p(x,t)\,dx,\quad \forall i\in \{1,\ldots, N\},
    \end{equation}
    Since $p(\cdot,t)\in\mathcal{P}(\hat p,B,t)$, for all $x\in X$ we have
    $
    \hat p(x,t)-B(t) \le p(x,t) \le \hat p(x,t)+B(t).
    $
    By definition of $\underline p_i,\overline p_i$ in Eq.~\eqref{eq:relaxed_cell_bounds},
    $
    \underline p_i \le \min_{x\in c_i}\hat p(x,t)-B(t)
    \le \hat p(x,t)-B(t) \le p(x,t) \le \hat p(x,t)+B(t) \le \max_{x\in c_i}\hat p(x,t)+B(t) \le \overline p_i,
    \quad \forall x\in c_i.
    $
    Therefore, for all $x\in c_i$,
    $
    p_i^- := \max\{0,\underline p_i\} \le p(x,t) \le \overline p_i := p_i^+.
    $ 
    Integrating over $c_i$ and dividing by $\mathrm{Vol}(c_i)$ yields
    $
    p_i^- \le \frac{1}{\mathrm{Vol}(c_i)}\int_{c_i} p(x,t)\,dx = p_i \le p_i^+,
    $
    where the equality holds by definition in Eq.~\eqref{eq:avg_pi_def}.
    Also, by definition of $c_i$ in Proposition~\ref{proposition:lp},
    $
    \sum_{i=1}^N p_i\,\mathrm{Vol}(c_i)
    \le \int_{\mathbb{R}^n} p(x,t)\,dx
    = 1,
    $
    so $\{p_i\}_{i=1}^N$ is feasible for constraints~\eqref{eq:lp_constraints}.
    
    \emph{Upper bound.}
    Let $I:=\{i:\ c_i\cap X_{\mathrm{event}}\neq\emptyset\}$ and $U:=\bigcup_{i\in I} c_i$. Then
    $X_{\mathrm{event}}\subseteq U$, and since $p\ge 0$,
    $
    \mathbb{P}(x\in X_{\mathrm{event}},t)
    = \int_{X_{\mathrm{event}}} p(x,t)\,dx
    \le \int_U p(x,t)\,dx
    = \sum_{i\in I}\int_{c_i} p(x,t)\,dx
    = \sum_{i\in I} p_i\,\mathrm{Vol}(c_i).
    $
    Since the LP defining $\mathbb{P}^+(t;X_{\mathrm{event}})$ maximizes the same expression over all feasible
    decision variables $\{\tilde p_i\}_{i=1}^N$, we have
    $
    \mathbb{P}(x\in X_{\mathrm{event}},t)
    \le \sum_{i\in I} p_i\,\mathrm{Vol}(c_i)
    \le \max_{\tilde p_i} \sum_{i\in I} \tilde p_i\,\mathrm{Vol}(c_i)
    = \mathbb{P}^+(t;X_{\mathrm{event}}).
    $
    
    \emph{Lower bound.}
    Let $J:=\{i:\ c_i\subseteq X_{\mathrm{event}}\}$. Then $\bigcup_{i\in J} c_i \subseteq X_{\mathrm{event}}$, and since $p\ge 0$,
    $
    \mathbb{P}(x\in X_{\mathrm{event}},t)
    = \int_{X_{\mathrm{event}}} p(x,t)\,dx
    \ge \int_{\cup_{i\in J} c_i} p(x,t)\,dx
    = \sum_{i\in J}\int_{c_i} p(x,t)\,dx
    = \sum_{i\in J} p_i\,\mathrm{Vol}(c_i).
    $
    Minimizing the right-hand side over all feasible $\{\tilde p_i\}_{i=1}^N$ gives the LP value
    $\mathbb{P}^-(t;X_{\mathrm{event}})$, hence
    $
    \mathbb{P}(x\in X_{\mathrm{event}},t)
    \ge \min_{\tilde p_i} \sum_{i\in J} \tilde p_i\,\mathrm{Vol}(c_i)
    = \mathbb{P}^-(t;X_{\mathrm{event}}).
    $
\end{proof}

\begin{remark}
The event-probability bounds $\mathbb{P}^+$ and $\mathbb{P}^-$ tighten as (i) the PINN's error bound $B$ decreases or (ii) the partition is refined (larger $|C|$, smaller $\mathrm{Vol}(c_i)$), which tightens $\underline p_i$ and $\overline p_i$.
\end{remark}

To solve Problem~\ref{prob:integral_pdf} by Proposition~\ref{proposition:lp}, one must
compute $\underline p_i,\overline p_i$ for each cell $c_i\in C$, satisfying Eq.~\eqref{eq:relaxed_cell_bounds}, and then set $p_i^-:=\max\{0,\underline p_i\}$ and $p_i^+:=\overline p_i$.
When $\hat p$ is represented by a neural network, $\underline p_i$ and $\overline p_i$ can be obtained via interval arithmetic using tools such as auto-LiRPA~\citep{xu2020automatic}.
However, when $c_i$ is a compact convex polytope and $\hat p$ is a PINN-GMM, these bounds can be computed efficiently as we show below.

\begin{proposition}[Mixture inequality]\label{proposition:mix-bounds}
Consider non-negative weights $w_1, \ldots, w_K$ such that $\sum_{k=1}^K w_k=1$, and 
\rv{let $\mathcal{C}$ be a nonempty compact set, and let $f_k:\mathcal{C}\to\mathbb{R}_{\geq 0}$ be continuous functions on $\mathcal{C}$.}
Then
$
\sum_{k=1}^K \Big( w_k \min_{x\in\mathcal{C}} f_k(x) \Big)
\le
\min_{x\in\mathcal{C}} \Big( \sum_{k=1}^K w_k f_k(x) \Big)
\le
\max_{x\in\mathcal{C}}  \Big( \sum_{k=1}^K w_k f_k(x) \Big)
\le
\sum_{k=1}^K \Big( w_k \max_{x\in\mathcal{C}} f_k(x) \Big).
$
\end{proposition}

\begin{lemma}[PDF bounds of a Gaussian distribution over a compact convex polytope]\label{lemma:quad-polytope}
\rv{Let $\mathcal{C}\subset\mathbb{R}^n$ be a nonempty compact convex polytope, corresponding to a partition cell $c_i$ in Proposition~1,} and let
$\phi:\mathbb{R}^n\to\mathbb{R}_{\ge 0}$ be a Gaussian PDF with mean $\mu\in\mathbb{R}^n$ and
positive definite covariance $\Sigma\in\mathbb{R}^{n\times n}$.
Let $c_{\phi}=(2\pi)^{-n/2}\det(\Sigma)^{-1/2}$ and define
\begin{equation}\label{eq:quad_min_max}
q(x):=\tfrac12 (x-\mu)^\top \Sigma^{-1} (x-\mu),
\qquad
q^-:=\min_{x\in\mathcal{C}} q(x),\quad q^+:=\max_{x\in\mathcal{C}} q(x).
\end{equation}
Then
$
\min_{x\in\mathcal{C}} \phi(x) = c_{\phi}\exp(-q^+)$ and $
\max_{x\in\mathcal{C}} \phi(x) = c_{\phi}\exp(-q^-).
$
\end{lemma}

\begin{proof}
Since $\phi(x)=c_{\phi}\exp(-q(x))$ and $\exp(-y)$ is strictly decreasing in $y$,
minimizing $\phi$ over $\mathcal{C}$ is equivalent to maximizing $q$ over $\mathcal{C}$, and vice versa.
\end{proof}

\begin{remark}\label{remark:eff_bound}
    Lower bound $q^-$ in Eq.~\eqref{eq:quad_min_max} is a standard convex quadratic program. Since $\Sigma^{-1}$ is positive definite, the objective
$q(x)=(x-\mu)^\top\Sigma^{-1}(x-\mu)$ is smooth and strongly convex, so projected gradient descent with an appropriate
step size converges globally to the unique minimizer over
$\mathcal{C}$ with a \emph{linear} rate~\cite{nesterov1998introductory}.
Upper bound $q^+$ in Eq.~\eqref{eq:quad_min_max} is attained at a vertex of $\mathcal{C}$, i.e., $q^+=\max_{v\in\mathrm{Vert}(\mathcal{C})} q(v)$ by Bauer's maximum principle~\citep{bauer1958minimalstellen,kruvzik2000bauer}.
\end{remark}

\begin{theorem}[Bounds for a Gaussian-mixture admissible set over a compact convex polytope]\label{theorem:gmm-cell-bounds}
Consider $\hat p(x,t)=\sum_{k=1}^K w_k\,\phi_k(x,t)$ with $w_k\ge 0$ and $\sum_{k=1}^K w_k=1$,
where each $\phi_k(\cdot,t)$ is a Gaussian PDF as in Lemma~\ref{lemma:quad-polytope}.
Let $B(t)\ge 0$ be given and let $c_i\subset \mathbb{R}^n$ be a nonempty compact convex polytope.
Define $\underline p_i := \sum_{k=1}^K \Bigl(w_k\,\min_{x\in c_i}\phi_k(x,t)\Bigr)-B(t)$ and $\overline p_i
:= \sum_{k=1}^K \Bigl(w_k\,\max_{x\in c_i}\phi_k(x,t)\Bigr)+B(t)$.
Then $\underline p_i,\overline p_i$ satisfy the relaxed constraints in Eq.~\eqref{eq:relaxed_cell_bounds}, i.e.,
$
\underline p_i \le \min_{x\in c_i}\bigl(\hat p(x,t)-B(t)\bigr)$ and $
\overline p_i \ge \max_{x\in c_i}\bigl(\hat p(x,t)+B(t)\bigr).
$
Consequently, setting $p_i^-:=\max\{0,\underline p_i\}$ and $p_i^+:=\overline p_i$ yields valid bounds for
Proposition~\ref{proposition:lp}.
\end{theorem}

\begin{proof}
By Proposition~\ref{proposition:mix-bounds} with $\mathcal{C}=c_i$ and $f_k(\cdot)=\phi_k(\cdot,t)$, we have
$
\sum_{k=1}^K \Bigl(w_k\,\min_{x\in c_i}\phi_k(x,t)\Bigr)
\le
\min_{x\in c_i}\sum_{k=1}^K w_k\phi_k(x,t)
=
\min_{x\in c_i}\hat p(x,t),
$
and
$
\max_{x\in c_i}\hat p(x,t)
=
\max_{x\in c_i}\sum_{k=1}^K w_k\phi_k(x,t)
\le
\sum_{k=1}^K \Bigl(w_k\,\max_{x\in c_i}\phi_k(x,t)\Bigr).
$
Since $B(t)$ is constant in $x$, subtracting/adding $B(t)$ yields
$
\sum_{k=1}^K \Bigl(w_k\,\min_{x\in c_i}\phi_k(x,t)\Bigr) - B(t)
\le
\min_{x\in c_i}\hat p(x,t)-B(t),
$
and
$
\max_{x\in c_i}\hat p(x,t)+B(t)
\le
\sum_{k=1}^K \Bigl(w_k\,\max_{x\in c_i}\phi_k(x,t)\Bigr) + B(t).
$
Let
$
\underline p_i = \sum_{k=1}^K \Bigl(w_k\,\min_{x\in c_i}\phi_k(x,t)\Bigr) - B(t)\le
\min_{x\in c_i}\hat p(x,t)-B(t),
$
and
$
\overline p_i = \sum_{k=1}^K \Bigl(w_k\,\max_{x\in c_i}\phi_k(x,t)\Bigr) + B(t)\geq \max_{x\in c_i}\hat p(x,t)+B(t).
$
By Lemma~\ref{lemma:quad-polytope} and Remark~\ref{remark:eff_bound}, $\underline p_i$ and $\overline p_i$ are computable, and they satisfy Eq.~\eqref{eq:relaxed_cell_bounds}.
\end{proof}

Building on Theorem~\ref{theorem:gmm-cell-bounds}, we present a general Algorithm~\ref{alg:solvelp} for Proposition~\ref{proposition:lp} via PINN-GMM $\hat p$ and its error bound $B$.

\begin{algorithm}[htbp!]
\caption{Compute Upper/Lower Bounds for $\mathbb{P}(x \in X_{\mathrm{event}}, t)$}
\label{alg:solvelp}
\begin{algorithmic}[1]
\Require $X, X_{\mathrm{event}}$; time $t$; PINN-GMM $\hat p(x,t)$ and error bound $B(t)$
\State Partition $X$ into polytope cells $C = \{c_i\}_{i=1}^N$
\State Compute $p_i^-,p_i^+$ for each cell $c_i \in C$ via Theorem~\ref{theorem:gmm-cell-bounds}
\State Solve the LP in Proposition~\ref{proposition:lp}
\State \Return $\mathbb{P}^+(t),\,\mathbb{P}^-(t)$
\end{algorithmic}
\end{algorithm}
\section{Training Implementation}\label{sec:train_implementation}
We describe our physics-informed neural network (PINN) training pipeline in three parts:
(i) loss functions derived from the Fokker-Planck operator and the initial condition;
(ii) sampling strategies for the initial and differential losses; and
(iii) the optimizer.
We first give pseudocode for training the space-time probability density function $\hat p(x,t)$ (Algorithm~\ref{alg:phat})
and the error approximation $\hat e(x,t)$ used to construct the time-dependent first-order bound $B(t)$ (Algorithm~\ref{alg:ehat}),
then detail each component.
\rv{All hyperparameter values are summarized in Table~\ref{tab:hyperparams}.}

\begin{algorithm}[htbp!]
\caption{Training $\hat p$}
\label{alg:phat}
\begin{algorithmic}[1]
\Require $\hat p_\theta$; FP-PDE (Eq.~\eqref{eq:fp_pde_compact});
sampling strategy $\textsc{Sampling}$; optimizer $\textsc{Opt}$
\State Initialize losses $\mathcal{L}_0,\mathcal{L}_r$; set $\text{best\_loss}\gets\infty$
\For{$k=1$ \textbf{to} $k_{\max}$}
    \State $(\mathcal{B}^{(0)}_k,\mathcal{B}^{(r)}_k)\gets \textsc{Sampling}$ \Comment{IC and residual batches}
    \State $\mathcal{L}\gets \mathcal{L}_0 + w_r\,\mathcal{L}_r,\quad \mathcal{L}_0 = \mathcal{L}_0^{\hat p}(\hat p;\mathcal{B}^{(0)}_k, p_0),\; \mathcal{L}_r = \mathcal{L}_r^{(\hat p)}(\hat p;\mathcal{B}^{(r)}_k, \mathcal{D})$ \Comment{Physics-informed loss}
    \State \textsc{Backprop}$(\mathcal{L})$;\quad \textsc{Opt.step}()
    \State \textbf{if} $\mathcal{L}<\text{best\_loss}$ \textbf{then} $\text{best\_loss}\leftarrow \mathcal{L}$;\ $\theta^\star\leftarrow\theta$
\EndFor
\State \Return $\hat p_{\theta^\star}$
\end{algorithmic}
\end{algorithm}

\begin{algorithm}[htbp!]
\caption{Training $\hat e$ and constructing $B(t)$}
\label{alg:ehat}
\begin{algorithmic}[1]
\Require $\hat e_\theta$; error PDE for $e$ (Eq.~\eqref{eq:error_pde} including $\hat p$); $\textsc{Sampling}$; $\textsc{Opt}$
\State Initialize $\mathcal{L}_0,\mathcal{L}_r$; set $\text{best\_loss}\gets\infty$
\For{$k=1$ \textbf{to} $k_{\max}$}
    \State $(\mathcal{B}^{(0)}_k,\mathcal{B}^{(r)}_k)\gets \textsc{Sampling}$ \Comment{IC and residual batches}
    \State $\mathcal{L}\gets \mathcal{L}_0 + w_r\,\mathcal{L}_r, \quad \mathcal{L}_0 = \mathcal{L}_0^{(e)}(\hat e;\mathcal{B}^{(0)}_k, p_0, \hat p),\; \mathcal{L}_r = \mathcal{L}_r^{(e)}(\hat e;\mathcal{B}^{(r)}_k, \mathcal{D}, \hat p)$ \Comment{Physics-informed loss}
    \State \textsc{Backprop}$(\mathcal{L})$;\quad \textsc{Opt.step}()
    \State \textbf{if} $\mathcal{L}<\text{best\_loss}$ \textbf{then} $\text{best\_loss}\leftarrow \mathcal{L}$;\ $\theta^\star\leftarrow\theta$
\EndFor
\State \Return $\hat e_{\theta^\star}$ and $B(t)=2\,\max_{x\in X}|\hat e_{\theta^\star}(x,t)|$ by Eq.~\eqref{eq:first_order_error_bound}
\end{algorithmic}
\end{algorithm}

\subsection{Loss}\label{subsec:loss}
At a generic training iteration, the sampling strategy returns an initial-condition batch
$\mathcal{B}^{(0)}=\{x_i^{(0)}\}_{i=1}^{N^{(0)}}$ (evaluated at $t=0$) and a residual batch
$\mathcal{B}^{(r)}=\{(x_i^{(r)},t_i^{(r)})\}_{i=1}^{N^{(r)}}$ in $X\times T$.
\rv{As noted in Sec.~\ref{sec:problem}, $X$ is a computational region of interest, so no boundary-condition loss is imposed on $\partial X$; the losses enforce only the initial condition and PDE residuals.}
We train PINN-GMM $\hat p$ with a Hellinger-type initial-condition (IC) loss.
The HT loss is a surrogate inspired by the Hellinger distance; its appeal stems from the inequality that, for any two densities $p$ and $q$, the exact Hellinger distance upper bounds the total variation (TV)\rv{, i.e.,
$H^2(p,q)=\frac{1}{2}\int_X(\sqrt{p(x)}-\sqrt{q(x)})^2dx$,
$\mathrm{TV}(p,q)=\frac{1}{2}\int_X |p(x)-q(x)|dx$, and
$\mathrm{TV}(p,q)\leq \sqrt{2}H(p,q)$.}
This perspective is not new: several studies (e.g., Ref.~\cite{izenman1991review} and references therein) have argued that $\ell_1$-type metrics (including the Hellinger distance) can be more suitable than $\ell_2$ metrics for density estimation.
Therefore, given per-iteration batches $(\mathcal{B}^{(0)},\mathcal{B}^{(r)})$, the loss is $\mathcal{L}^{(p)} = \mathcal{L}_{0,\mathrm{HT}}^{(p)}(\mathcal{B}^{(0)}) + w_r\,\mathcal{L}_r^{(p)}(\mathcal{B}^{(r)})$, where $\mathcal{L}_{0,\mathrm{HT}}^{(p)}(\mathcal{B}^{(0)}) = \frac{1}{2N^{(0)}}\sum_{x\in\mathcal{B}^{(0)}}\Big(\sqrt{p(x,0)}-\sqrt{\hat p(x,0)}\Big)^2$ and $\mathcal{L}_r^{(p)}(\mathcal{B}^{(r)}) = \frac{1}{N^{(r)}}\sum_{(x,t)\in\mathcal{B}^{(r)}}\Big(\mathcal{D}[\hat p](x,t)\Big)^2$.
For the error-PINN $\hat e$, where the \emph{true} error $e:=p-\hat p$ satisfies Eq.~\eqref{eq:error_pde}, we adopt a standard mean-square-error PINN loss, i.e., $\mathcal{L}^{(e)} = \mathcal{L}_0^{(e)} + w_r\,\mathcal{L}_r^{(e)}$, where $\mathcal{L}_0^{(e)} = \frac{1}{N^{(0)}}\sum_{i=1}^{N^{(0)}}\big(p(x_i^{(0)},t_0)-\hat p(x_i^{(0)},t_0)-\hat e(x_i^{(0)},t_0)\big)^2$ and $\mathcal{L}_r^{(e)} = \frac{1}{N^{(r)}}\sum_{i=1}^{N^{(r)}}\big(\mathcal{D}[\hat e](x_i^{(r)},t_i^{(r)})+\mathcal{D}[\hat p](x_i^{(r)},t_i^{(r)})\big)^2$.

\subsection{Sampling}

In high dimensions, uniformly drawn points often lie in regions where the density and/or the residual magnitude $|\mathcal{D}[\cdot]|$ are small, contributing little to the physics-informed loss; see the empty-space phenomenon in multivariate settings~\citep{scott1983probability}.
\rv{In Algorithms~\ref{alg:phat} and~\ref{alg:ehat}, \textsc{Sampling} denotes the batch-generation process that returns an initial-condition batch $\mathcal{B}^{(0)}$ and a residual batch $\mathcal{B}^{(r)}$ at each training iteration.}
To \rv{implement} this \rv{process}, we first define two basic strategies: (i) a \emph{uniform baseline}, and (ii) an \emph{adaptive pool} strategy that maintains a persistent set of informative samples.
Our proposed approach is the (iii) \emph{combined} strategy, which builds on these two by augmenting the adaptive pool with fresh random mini-batches.

\textbf{Uniform baseline:}
At iteration $k$, the strategy produces
$
\mathcal{B}^{(0)}_k=\{x_i^{(0)}\}_{i=1}^{N^{(0)}}$ and $
\mathcal{B}^{(r)}_k=\{(x_i^{(r)},t_i^{(r)})\}_{i=1}^{N^{(r)}},
$
with independent draws
$
x_i^{(0)}\stackrel{\text{i.i.d.}}{\sim}U(X),
x_i^{(r)}\stackrel{\text{i.i.d.}}{\sim}U(X),$
and $
t_i^{(r)}\stackrel{\text{i.i.d.}}{\sim}U(T).
$
This uniform sampling coincides with the standard sampling strategy used in most PINN implementations~\cite{sirignano2018dgm,lu2021deepxde}.
While it offers unbiased coverage, it may still place many samples in regions that provide weak gradient signals.

\textbf{Adaptive pool:}
Adaptive sampling grows a persistent pool of training points over time~\citep{lu2021deepxde}. 
Instead of drawing a fresh random mini-batch at every iteration, we
(i) initialize a pool set $\mathcal{S}$ with samples from a proposal distribution $\pi_{\mathrm{cand}}$;
(ii) use the \emph{entire} pool as the training batch for the next $k_{\mathrm{adap}}$ iterations; and
(iii) every $k_{\mathrm{adap}}$ iterations, draw $N_{\mathrm{cand}}$ candidates from $\pi_{\mathrm{cand}}$, score them (e.g., by initial-condition discrepancy or differential residual), and \rv{enlarge $\mathcal{S}$ by adding the $N_{\mathrm{add}}$ highest-scoring candidates, with $N_{\mathrm{add}}\le N_{\mathrm{cand}}$.}
This repeated enrichment helps avoid early stagnation and improves stability.
The proposal $\pi_{\mathrm{cand}}$ can be uniform, biased toward informative regions, or a mixture.
While the adaptive pool targets informative samples, the combined strategy below restores per-iteration domain coverage by adding fresh random samples.

\textbf{Combination:}
At iteration $k$, we form each batch by mixing the current pool with new random samples:
$
\mathcal{B}^{(0)}_k
= \mathcal{S}^{(0)}_k \cup  \mathcal{B}^{(0)}_{\mathrm{rand},k}$ and $
\mathcal{B}^{(r)}_k
= \mathcal{S}^{(r)}_k \cup  \mathcal{B}^{(r)}_{\mathrm{rand},k},
$
where $\mathcal{S}^{(0)}_k$ and $\mathcal{S}^{(r)}_k$ are the current pools of the IC and residual samples, respectively, and
$\mathcal{B}^{(\cdot)}_{\mathrm{rand},k}$ \rv{is a fresh random mini-batch of size $N_{\mathrm{rand}}$ sampled i.i.d.\ from $\pi_{\mathrm{cand}}$.}
We use a first-in-first out (FIFO) rule on the adaptive pool, i.e.,
$
|\mathcal{S}^{(0)}_k|\le N_{\mathrm{rand}}$ and $ |\mathcal{S}^{(r)}_k|\le N_{\mathrm{rand}},
$
so that the pool does not dominate the fresh random draws as the training progresses.

\subsection{Optimization}
We use the Adam optimizer with Pytorch default setting $(\beta_1,\beta_2,\varepsilon)=(0.9,0.999,10^{-8})$ to update the neural network parameters.
For a standard MLP representation, all parameters share the initial learning rate
$
\eta_0^{\mathrm{MLP}} = 10^{-3}.
$
For the PINN-GMM, we optimize three parameter groups corresponding to the network outputs:
(i) the mean network \(\mu_z(t)\),
(ii) the covariance network \(L_z(t)\), and
(iii) the weight network.
\rv{As defined in the time-conditioned Gaussian-mixture architecture in Sec.~\ref{subsec:pinn_gmm_model}, the encoder parameters $(\mu_{\mathrm{en}},\Sigma_{\mathrm{en}})$ are fixed from the initial state distribution; these quantities are therefore not learned and are excluded from the optimizer.} 
A time-dependent encoder, as in Sec.~\ref{subsec:general_coord}, is a possible extension but is not used here.
Because the weight map \(v\!\mapsto\!w\) and the covariance map \(\ell\!\mapsto\!L_z\!\mapsto\!\Sigma_z\) are constrained/reparameterized, their effective gradients tend to be sharper and less well scaled than those of the means. We therefore assign a smaller learning rates for weights and covariances, i.e.,
$
\eta_0^{(\mu_z)} = 10^{-3},
\eta_0^{(L_z)}   = 10^{-4}$ and $
\eta_0^{(w_z)}   = 10^{-4}.
$
This choice stabilizes training of weights and covariances, while allowing the mean to adapt more quickly.
We apply a global step-decay schedule to all parameters, i.e.,
$
\eta_k = \eta_0 \,\gamma^{\left\lfloor k/k_{\mathrm{decay}}\right\rfloor},
$
where $\gamma\in(0,1)$ and $k$ is the iteration number and $k_{\mathrm{decay}}$ the decay period.
We use gradient-norm clipping with threshold $c_{\mathrm{grad}}=1$ to improve training stability.
Training runs for at most $k_{\max}$ iterations, and we save the iterate that has the minimum total loss over the run.
\begin{table}[htb!]
\centering
\caption{\rv{Hyperparameter values for training the density PINN $\hat p$ and error PINN $\hat e$.}}
\label{tab:hyperparams}
\small
\resizebox{\linewidth}{!}{
\begin{tabular}{lll}
\toprule
Loss & Sampling & Optimization \\
\midrule
$w_r=1.0$ &
$k_{\mathrm{adap}}=100$, $N_{\mathrm{cand}}=30000$, $N_{\mathrm{add}}=32$, $N_{\mathrm{rand}}=4000$ &
$\eta_0\in\{10^{-3},10^{-4}\}$, $\gamma=0.95$, $k_{\mathrm{decay}}=1000$, $k_{\max}$ case-dependent \\
\bottomrule
\end{tabular}}
\end{table}




\section{Experiments}\label{sec:exp}
We evaluate our proposed PINN-GMM for PDF approximation, worst-case error bounding, and event-probability bounding on stochastic dynamical systems of increasing complexity.
\rv{These experiments also assess the learnability of the error PINN, since the reported bounds depend on learning the forced error PDE in Eq.~\eqref{eq:error_pde}.}
We compare \rv{PINN-GMM} against three baseline uncertainty-propagation methods: Gaussian approximation (GA), unscented transform (UT), and Gaussian mixture models (GMM); see Appendix~\ref{appendix:baselines} for implementation details.
We also compare against standard \rv{and ablated PINN variants to isolate the effects of the PINN-GMM architecture, encoder, and sampling strategy.}

We begin with two 1D examples: a linear SDE used to validate the baseline implementations\rv{, and a nonlinear SDE exhibiting non-Gaussian density evolution.}
\rv{The nonlinear example includes a sensitivity study over the number of PINN-GMM mixture components, \(K_{\mathrm{GMM}}\), to illustrate the accuracy--cost trade-off of the mixture representation.}
\rv{The orbital experiments are then organized into two groups.}
\rv{First, we study a 6D unperturbed Keplerian orbit.
The equinoctial-coordinate case provides an analytical reference for evaluation and demonstrates the computation of event-probability bounds via Algorithm~\ref{alg:solvelp}. Separately, the 6D examples also compare coordinate choices for training PINN-GMM by contrasting rotating spherical and Cartesian coordinates.}
\rv{Second, we study 4D perturbed Keplerian orbits with diffusion, including a \(J_2\)+diffusion case, one event-probability-bound demonstration, and an additional \(J_2\)+diffusion+low-thrust case.}
\rv{The longest horizon is approximately \(7.2\) hours, demonstrating hours-scale orbital uncertainty propagation with quantified error bounds.}

For all orbital cases, we use the same architectures for PINN-GMM $\hat p$ and the error network $\hat e$ (Tables~\ref{tab:nn_pinngmm_summary} and~\ref{tab:nn_enetxl_summary} in Appendix~\ref{appendix:nn}) and the same training hyperparameters (Table~\ref{tab:hyperparams}).
\rv{
The computational domain $X$ is selected by propagating a small batch of initial samples over the prescribed horizon and choosing $X$ to cover the reached region, following Ref.~\cite{sun2016uncertainty}; this small-batch propagation is used only for domain selection, not as a replacement for the high-fidelity MC reference. 
An overly large $X$ worsens information dilution during training, while an overly small $X$ causes probability mass to approach $\partial X$, making the approximation unreliable; if this occurs, the domain should be enlarged or redefined.}

Some recent orbit-specific methods are not included because their assumptions or available implementations do not cover all stochastic settings considered here. \rv{The state-transition-tensor methods~\citep{park2006nonlinear,fujimoto2012analytical,khatri2023nonlinear} and the GMM-PCE method~\citep{vittaldev2016spacecraft} are most relevant to the noise-free Keplerian case; however, a fair comparison would require independent implementation choices, including expansion order, mixture splitting, and polynomial-chaos construction. Moreover, for the 6D unperturbed Keplerian case in equinoctial elements, the exact push-forward density is available, so we evaluate all methods directly against this analytical reference rather than relying on another numerical propagation method as a surrogate benchmark.} We also do not include the tensor-decomposition method~\cite{sun2016uncertainty} because we could not access an open-source implementation. Similarly, the public implementation of the Galerkin projection method~\cite{acciarini2024uncertainty}, used as-is, is limited to lower-dimensional systems. \rv{Direct comparisons with verified implementations of these methods are left for future work.}
All code is available on GitHub~\citep{PINN_Error:github}.

\subsection{Reference ``True'' Probability Density}\label{subsec:ref_true_pdf}
To evaluate the uncertainty propagation approximation, we require an accurate ``true'' solution of the FP-PDE in Eq.~\eqref{eq:fp_pde_compact} as a reference.
A standard choice is a high-fidelity Monte Carlo (MC) simulation:
draw \(N\) independent initial states from \(p_{0}(x)\) and propagate each sample via the underlying SDE in Eq.~\eqref{eq:sde_general}.
\rv{We propagate each sample with the Euler--Maruyama scheme~\citep[Ch.~8]{sarkka2019applied}, choosing the integration time step to balance accuracy and cost.}

Given the propagated samples $\{x_i(t) \}_{i=1}^N$, the remaining task is to estimate a reference PDF.
For low-dimensional systems (up to 3D), simple histograms on a discretized grid are adequate~\citep[Ch.~2]{silverman2018density}.
For orbital dynamics, however, high dimension and large domains make direct discretization intractable and the resulting histogram-PDF inaccurate.
Accordingly, we construct the reference ``true'' PDF by fitting a Gaussian mixture model (GMM) to the MC samples at each time $t$.
Instead of using a fixed number of GMM components as in Ref.~\cite{acciarini2024uncertainty}, we sweep over different numbers of GMM components to obtain the best fitting; see Appendix~\ref{appendix:fitgmm} for details.

\rv{
We emphasize that this GMM-fitted reference density is used only as a post-processing estimator of the MC-sampled density for evaluation. When an analytical reference density is available, we use that reference directly instead of fitting a GMM to MC samples. 
In addition, evaluation measures that can be evaluated directly from samples, such as the relative divergence described in the next subsection, do not require a GMM-fitted reference density.}
We also tested kernel density \rv{(KDE)} for constructing the MC-based reference density. 
\rv{However, the resulting KDE estimates were sensitive to bandwidth selection and exhibited finite-sample numerical irregularities, a known challenge in density estimation~\citep[Ch.~2.9]{silverman2018density}.} 
\rv{We therefore use a GMM fit as a smoother, more tractable MC-based density estimator, while recognizing that a finite-component GMM introduces parametric fitting error and may not fully capture strongly non-Gaussian features.}

\subsection{Evaluation Measures}

We assess an approximate density $q$ (possibly unnormalized) against a reference density $p$ using three empirical measures:
(i) \rv{total variation ($\widetilde{\mathrm{TV}}$), which measures global distributional discrepancy;}
(ii) \rv{worst-case normalized error ($\widetilde{\mathrm{WNE}}$), which measures the largest normalized pointwise deviation; and}
(iii) \rv{relative divergence ($\widetilde{\mathrm{RD}}$), which measures likelihood-based agreement with reference samples.}

\subsubsection{Total Variation}
\rv{We use the total variation (TV) defined in Sec.~\ref{subsec:loss} to measure the maximal difference between the probability of an event under the distributions with densities $p$ and $q$.}
Because the integral is seldom available in closed form, we estimate TV from samples:
\begin{equation}\label{eq:tv_emp}
\tilde{\mathrm{TV}}(p,q;X)
\;=\; \tfrac{1}{2}\,\frac{\mathrm{vol}(X)}{N}\sum_{i=1}^N |p(x_i)-q(x_i)|,
\qquad x_i \stackrel{\text{i.i.d.}}{\sim} U(X),
\end{equation}
where $\mathrm{vol}(X)$ is the volume of $X$. 
As $N\!\to\!\infty$, $\tilde{\mathrm{TV}}(p,q;X)$ converges to $TV(p,q;X)$.
We note that the uniform sampling $U(X)$ and the estimator in Eq.~\eqref{eq:tv_emp} are well-defined since we assume $X$ to be bounded in Problem~\ref{prob:1}.

\subsubsection{Worst-Case Normalized Error}
$
\|p-q\|_{\infty,X} \;=\; \sup_{x\in X}|p(x)-q(x)|.
$
measures the largest pointwise deviation over $X$.
We estimate this empirically and normalize by the sampled peak of $p$ to obtain a scale-free measure:
\begin{equation}\label{eq:worst_norm_error_emp}
\widetilde{\mathrm{WNE}}(p,q;X)
\;=\; \frac{\max_{1\le i\le N}|q(x_i)-p(x_i)|}{\max_{1\le i\le N} p(x_i)},
\qquad x_i \stackrel{\text{i.i.d.}}{\sim} U(X).
\end{equation}
This reports the worst-case normalized error over $X$.

\subsubsection{Relative Divergence}
The generalized Kullback-Leibler divergence (gKLD)~\citep{ferdosi2011comparison}, also known as Csiszár’s $I$-divergence~\citep{csiszar1991least}, between nonnegative $p$ and $q$ is
$
D(p\|q) = \int_{\mathbb{R}^n} \Big( p(x)\log\tfrac{p(x)}{q(x)} - p(x) + q(x) \Big)dx .
$
Assume $p$ is a probability density and $q$ is possibly unnormalized density with bounded total mass $Z_q:=\int_{\mathbb{R}^n} q(x)\,dx\in(0,\infty)$. 
Then
$
D(p\|q)
= \mathbb{E}_{x\sim p}\big[\log p(x)-\log q(x)\big] - 1 + Z_q
= C_p - \mathbb{E}_{x\sim p}[\log q(x)] - 1 + Z_q,
$
where $C_p:=\mathbb{E}_{x\sim p}[\log p(x)]$ depends only on $p$. Approximating the \rv{expectation} by a sample average for i.i.d.\ $x_i\sim p$ gives
$
D(p\|q) \approx (C_p-1) + \frac{1}{N}\sum_{i=1}^N \big(-\log q(x_i)\big) + Z_q,
$
which motivates the \emph{empirical relative divergence}, defined without evaluating the reference PDF $p$:
\begin{equation}
\widetilde{\mathrm{RD}}(p\|q) \;:=\; \frac{1}{N}\sum_{i=1}^N \big(-\log q(x_i)\big) \;+\; Z_q, \qquad \{x_i\}_{i=1}^N\sim p.
\end{equation}
If $q$ is normalized ($Z_q=1$), then $\widetilde{\mathrm{RD}}(p\|q)$ equals the negative log-likelihood up to an additive constant. 
Since the omitted constant $(C_p-1)$ is independent of $q$, minimizing $\widetilde{\mathrm{RD}}(p\|q)$ is an empirical surrogate for minimizing $D(p\|q)$, whose unique minimum is attained at $q=p$.
\rv{Because this constant is omitted, $\widetilde{\mathrm{RD}}(p\|q)$ is not the full nonnegative generalized KL divergence and may take negative values; it should therefore be interpreted comparatively within the same experiment, with smaller values indicating better likelihood-based agreement with the reference samples.}


Taken together, for sufficiently large $N$ and a reference density $p$,
$\widetilde{\mathrm{TV}}(p,\cdot; X)$ measures global discrepancy over $X$,
$\widetilde{\mathrm{WNE}}(p,\cdot; X)$ captures the worst normalized error over $X$, and
$\widetilde{\mathrm{RD}}(p\| \cdot)$ provides a likelihood-based measure over $N$ samples drawn directly from $p$.
\rv{
To complement the MC-based density comparisons, Appendix~\ref{appendix:6d_equin} reports the final PINN-GMM physics-informed training loss, including the FP-PDE residual term, for each orbital case.
}


\subsection{1D Illustrative Example}
\paragraph{Linear Example}
We first consider a 1D linear SDE with an analytical solution~\cite[Sec. 4.1]{pavliotis2014stochastic} to validate the baseline uncertainty-propagation methods: Gaussian
approximation (GA) and unscented transform (UT).
The third baseline: Gaussian Mixture Model (GMM) is built on GA; we therefore do not validate it separately here.
Figure~\ref{fig:1d_ou} shows the uncertainty propagation and the evaluation measures over time. 
Because the system is linear, both GA and UT (pink circle and cyan square lines, respectively) achieve zero total variation and zero worst-case normalized errors, and they also yield smaller relative divergence than a standard PINN-MLP $\hat p$ (purple triangle line).
Note that, for the baseline methods considered here, the relative divergence corresponds to the expected negative log-likelihood; consequently, it varies over time as the true covariance changes.
In Fig.~\ref{fig:1d_ou_worst_error}, we show the PINN's worst-case error bound $B(t)$ (purple shaded region) against the worst-case error and confirm that $B(t)$ upper-bounds it.
\begin{figure}[hbt!]
\centering
\begin{minipage}{0.32\textwidth}\centering
  \begin{subfigure}[t]{\linewidth}
    \includegraphics[width=\linewidth]{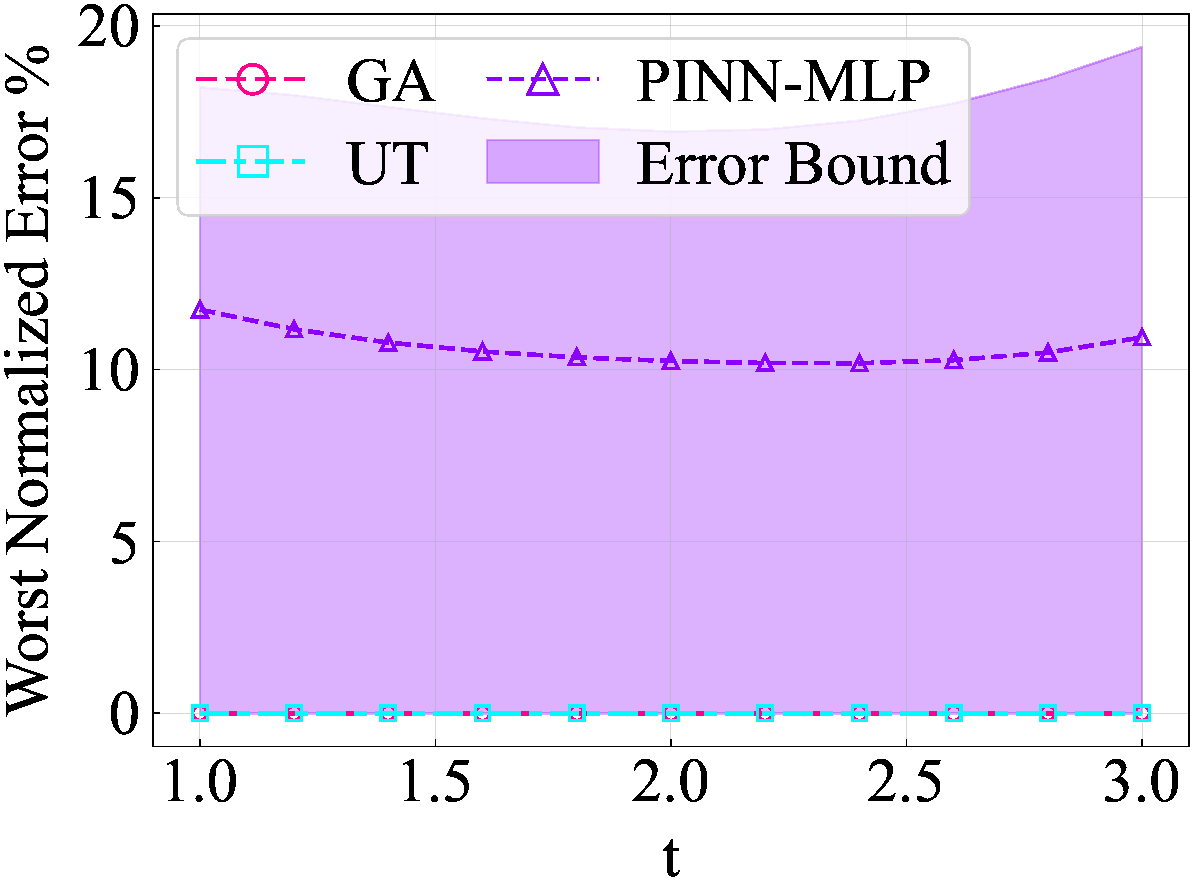}
    \caption{Worst normalized error \%}\label{fig:1d_ou_worst_error}
  \end{subfigure}
\end{minipage}
\begin{minipage}{0.32\textwidth}\centering
  \begin{subfigure}[t]{\linewidth}
    \includegraphics[width=\linewidth]{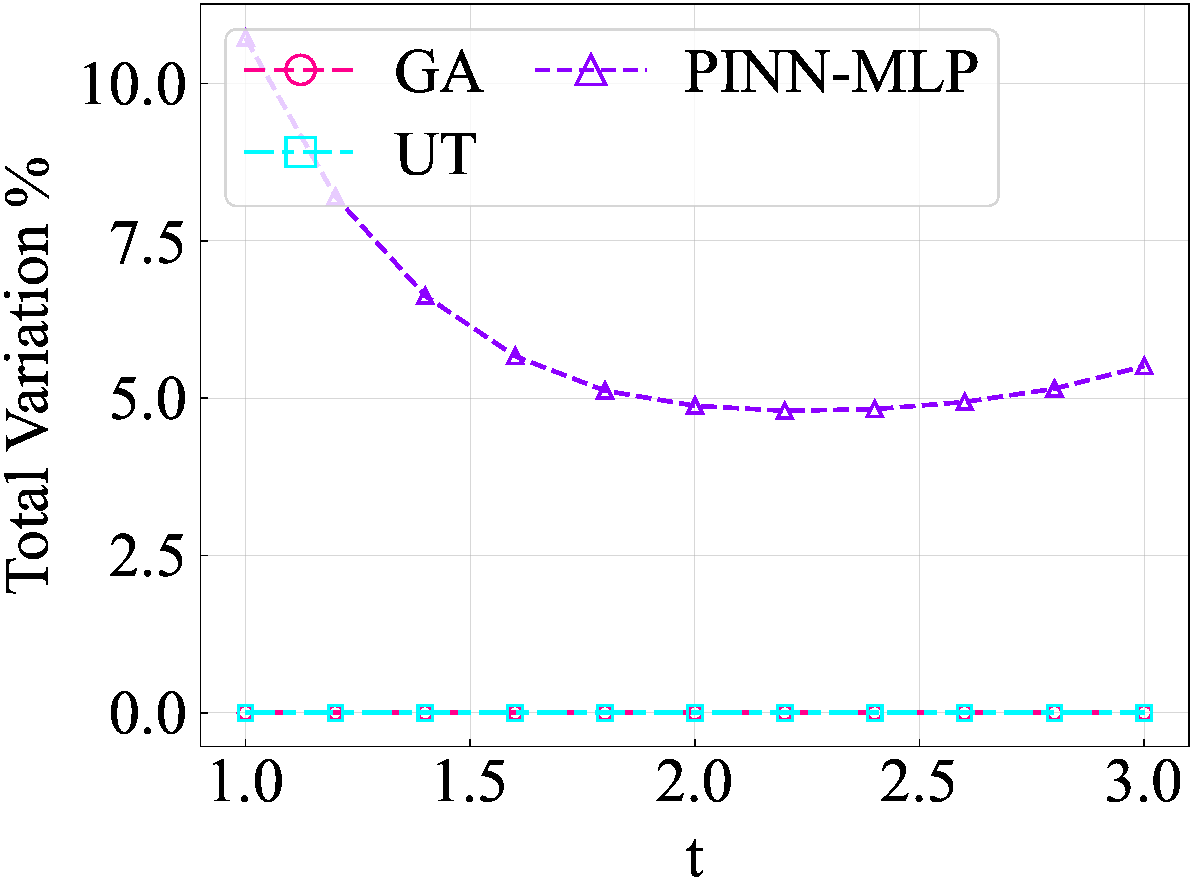}
    \caption{Total variation \%}\label{fig:1d_ou_tv}
  \end{subfigure}
\end{minipage}
\begin{minipage}{0.32\textwidth}\centering
  \begin{subfigure}[t]{\linewidth}
    \includegraphics[width=\linewidth]{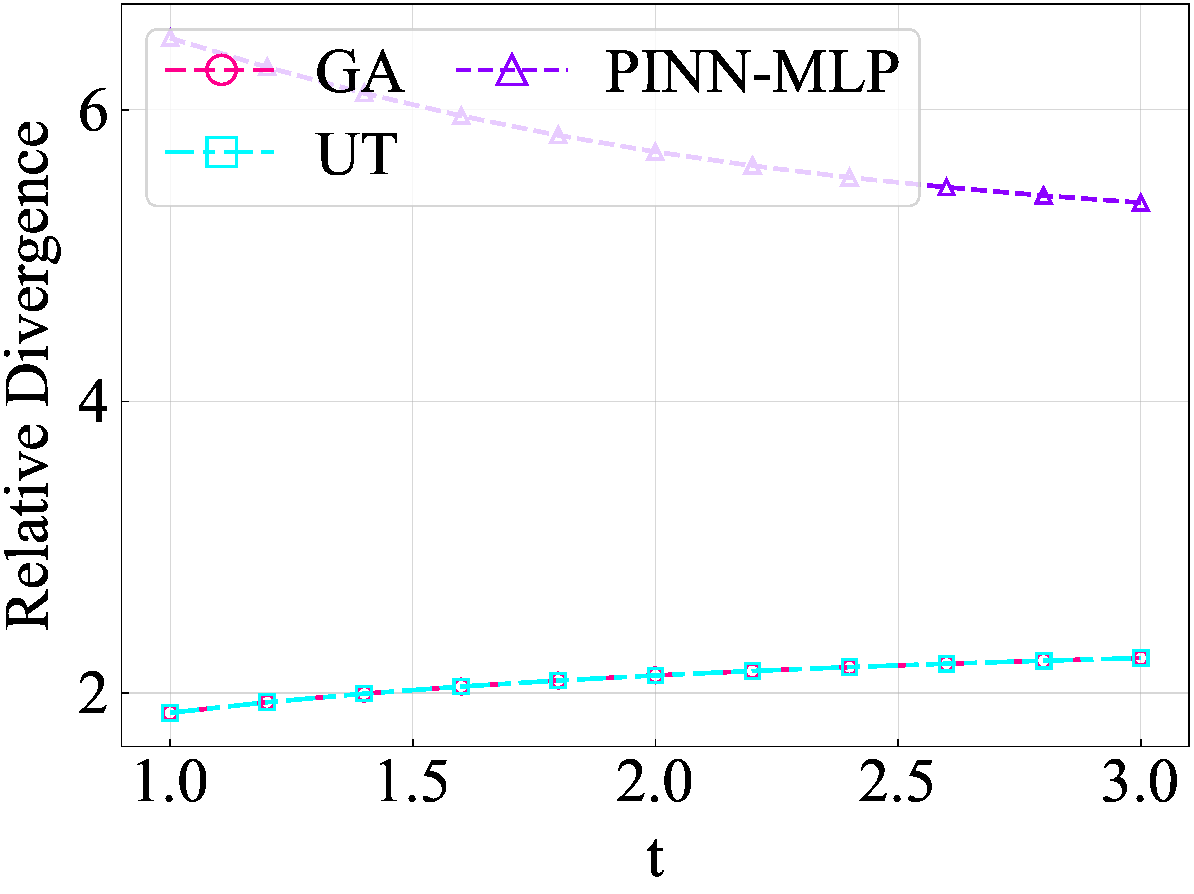}
    \caption{Relative divergence}\label{fig:1d_ou_rd}
  \end{subfigure}
\end{minipage}

\caption{\rv{
1D linear SDE error validation: a) $\widetilde{\mathrm{WNE}}$ with learned bound, b) $\widetilde{\mathrm{TV}}$, and c) $\widetilde{\mathrm{RD}}$ for GA, UT, and PINN-MLP.
}}
\label{fig:1d_ou}
\end{figure}

\paragraph{Nonlinear Example}
Next, we consider a nonlinear 1D SDE 
with Gaussian initial distribution $p_0(x)\sim \mathcal{N}(x;\mu=-2,\sigma^2=0.25)$.
Because this system has no analytical solution, the reference PDF $p(x,t)$ is obtained by extensive Monte-Carlo simulation using the Euler-Maruyama scheme\rv{; see Ref.~\cite{kongerror} for the system details and the visualization of the non-Gaussian density evolution.}
\rv{For this sensitivity study, we compare GA, UT, GMM, and PINN-GMM with \(K_{\mathrm{GMM}}\in\{1,5,11,15\}\) using $\widetilde{\mathrm{WNE}}$, the error measure directly associated with the learned worst-case error bound $B(t)$ by Algorithm~\ref{alg:ehat}.}
\rv{Figure~\ref{fig:1d_nl_wne_all} compares $\widetilde{\mathrm{WNE}}$ of all methods, while Figs.~\ref{fig:1d_nl_wne_bound_lowK} and~\ref{fig:1d_nl_wne_bound_higK} focus on the PINN-GMM models and overlay their \(B(t)\) as shaded regions; for visual clarity, we split to two groups.}
\rv{For all tested \(K_{\mathrm{GMM}}\), \(B(t)\) contains the associated $\widetilde{\mathrm{WNE}}$ over the time horizon.}
\rv{The corresponding PINN-GMM training times are \(9.7\), \(9.9\), \(12.5\), and \(13.2\) seconds for \(K_{\mathrm{GMM}}=1,5,11,15\), respectively.}
\rv{These results illustrate the accuracy--cost trade-off of \(K_{\mathrm{GMM}}\) discussed in Sec.~\ref{subsec:pinn_gmm_model}. They should not be interpreted as evidence of monotone convergence as \(K_{\mathrm{GMM}}\) increases. If the true PDF \(p(x,t)\) is already well approximated by a low-component Gaussian mixture, adding more components may provide little benefit and can also make optimization more difficult. Consequently, neither the approximation error nor the learned bound \(B_1(t)\) is expected to decrease monotonically with \(K_{\mathrm{GMM}}\).}
\begin{figure}[hbt!]
\centering
\begin{minipage}{0.32\textwidth}\centering
  \begin{subfigure}[t]{\linewidth}
    \includegraphics[width=\linewidth]{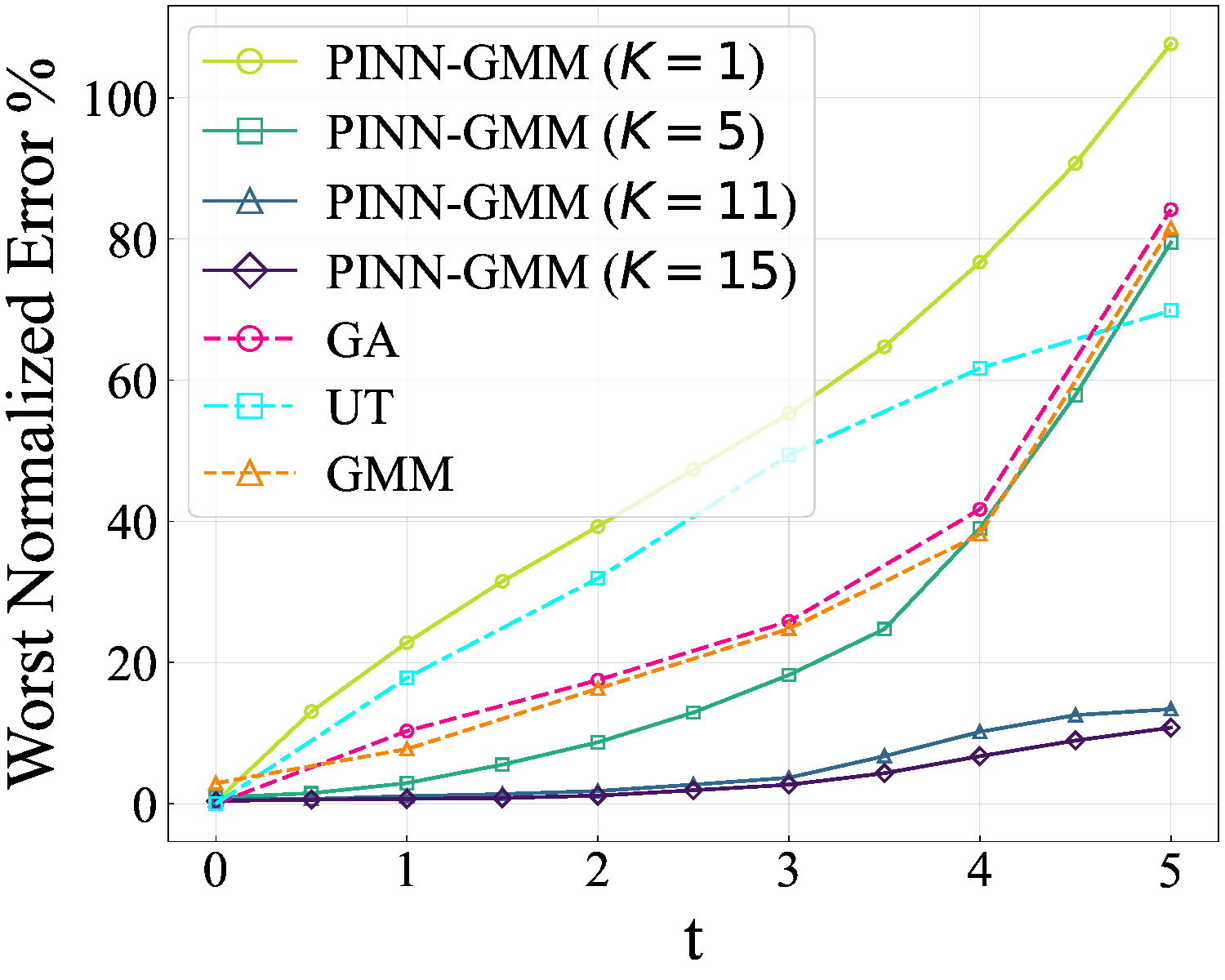}
    \caption{\rv{$\widetilde{\mathrm{WNE}}$ over time}}
    \label{fig:1d_nl_wne_all}
  \end{subfigure}
\end{minipage}
\begin{minipage}{0.32\textwidth}\centering
  \begin{subfigure}[t]{\linewidth}
    \includegraphics[width=\linewidth]{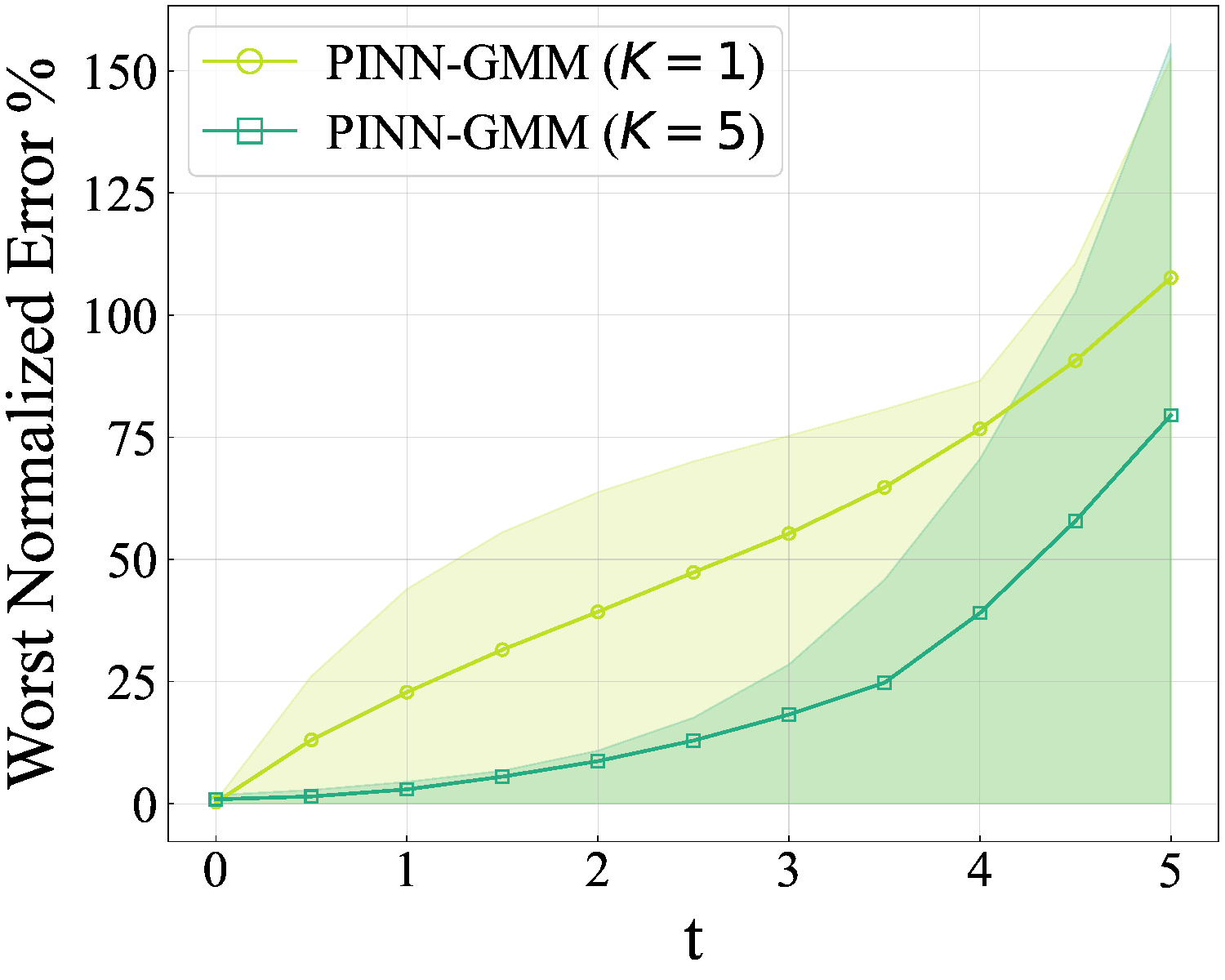}
    \caption{\rv{\(\widetilde{\mathrm{WNE}}\) and \(B_1\) for PINN-GMM.}}
    \label{fig:1d_nl_wne_bound_lowK}
  \end{subfigure}
\end{minipage}
\begin{minipage}{0.32\textwidth}\centering
  \begin{subfigure}[t]{\linewidth}
    \includegraphics[width=\linewidth]{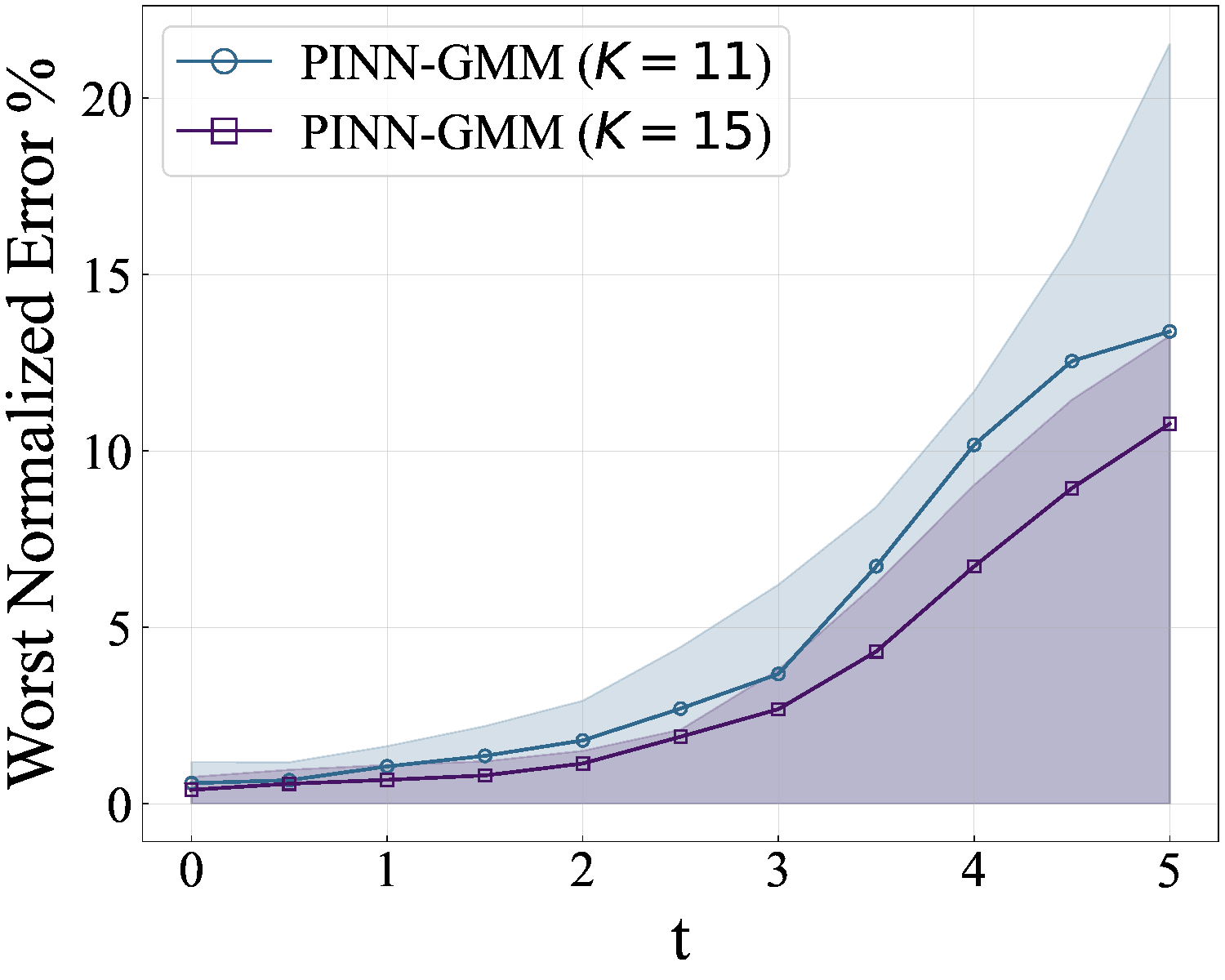}
    \caption{\rv{\(\widetilde{\mathrm{WNE}}\) and \(B_1\) for PINN-GMM}}
    \label{fig:1d_nl_wne_bound_higK}
  \end{subfigure}
\end{minipage}
\caption{\rv{$K_{\mathrm{GMM}}$ sensitivity in the 1D nonlinear case: a) $\widetilde{\mathrm{WNE}}$ comparison and b,c) PINN-GMM error bounds for varying $K_{\mathrm{GMM}}$.}}
\label{fig:1d_nl_wne}
\end{figure}
\begin{figure}[hbt!]
\centering
\begin{minipage}{0.45\textwidth}\centering
  \begin{subfigure}[t]{\linewidth}
    \includegraphics[width=\linewidth]{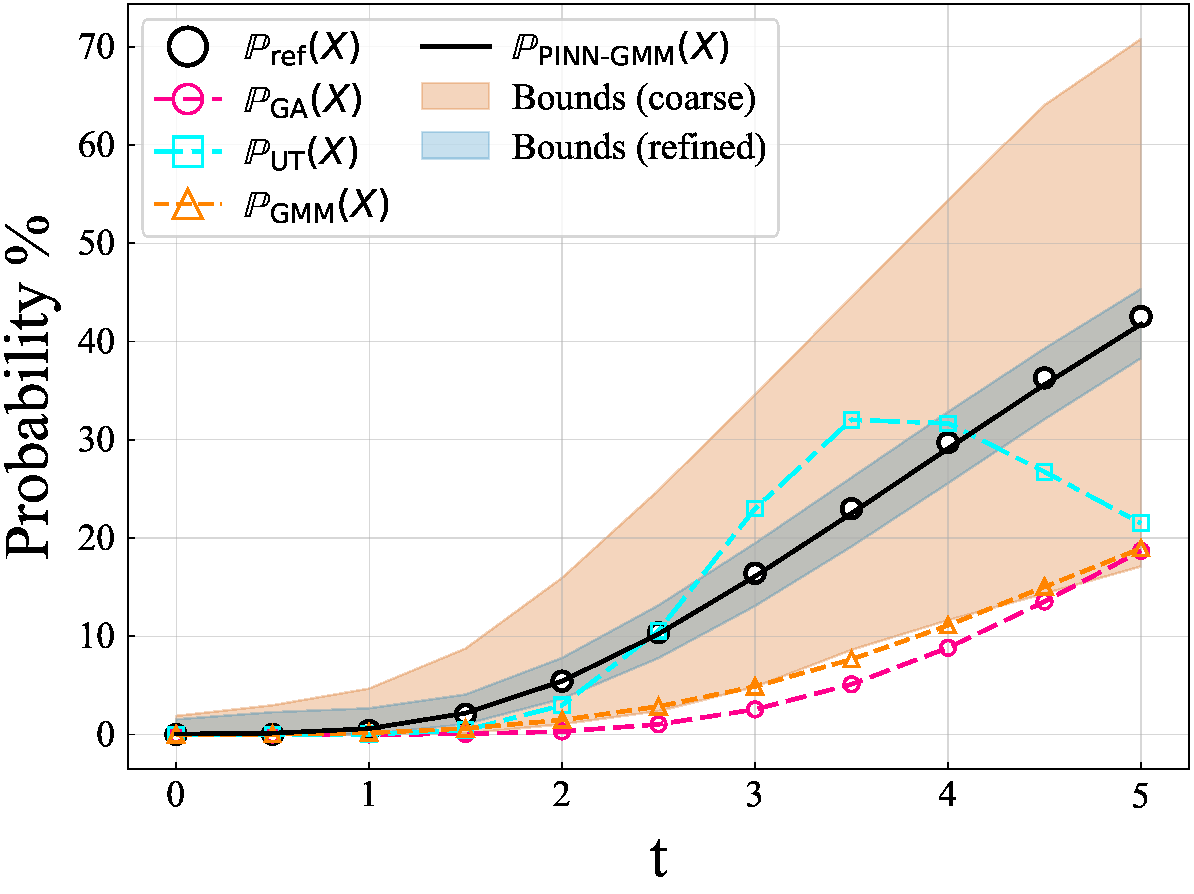}
    \caption{\rv{Event probability and bounds}}\label{fig:1d_nl_solver_N500_all_prob}
  \end{subfigure}
\end{minipage}
\begin{minipage}{0.45\textwidth}\centering
  \begin{subfigure}[t]{\linewidth}
    \includegraphics[width=\linewidth]{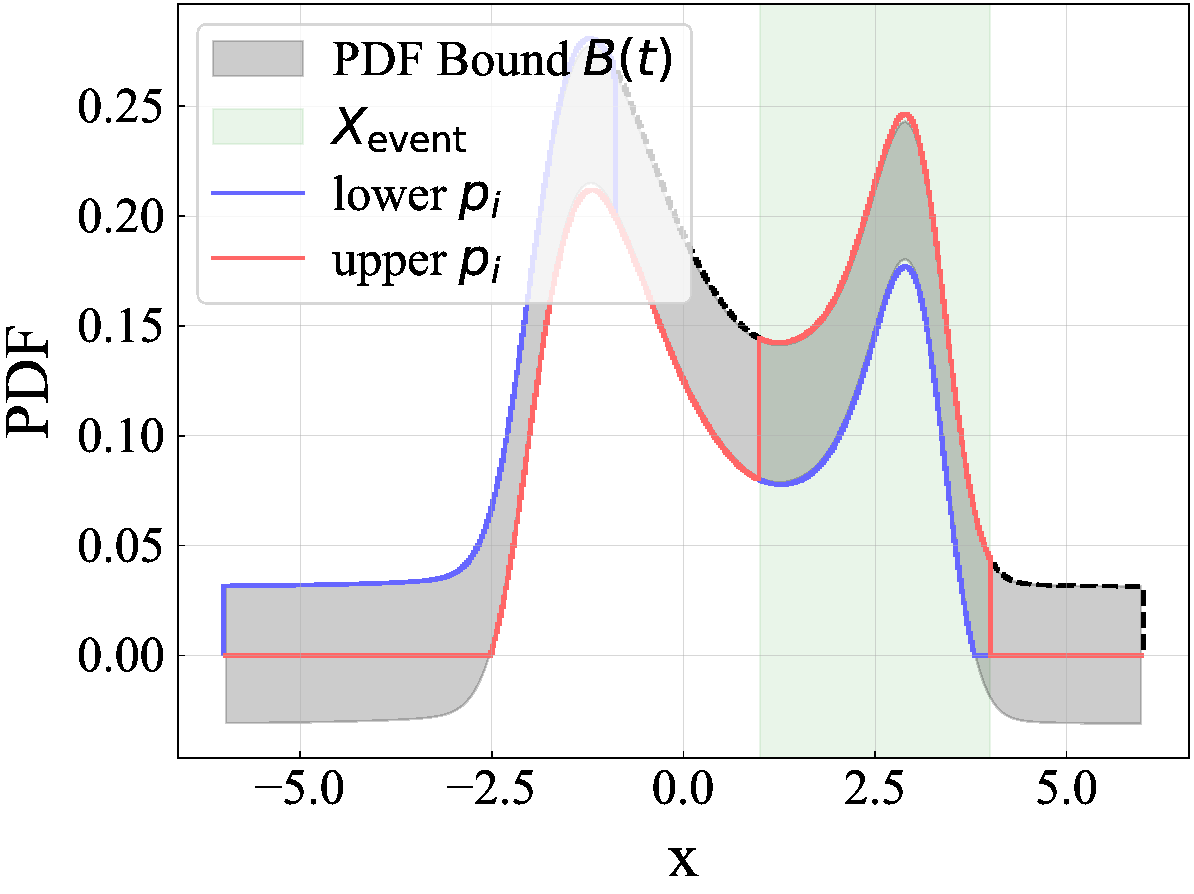}
    \caption{\rv{Refined LP solution at \(t=T_{\mathrm{end}}=5\) s}}\label{fig:1d_nl_solver_N500_pdf_t5.00}
  \end{subfigure}
\end{minipage}
\caption{\rv{1D nonlinear event-probability bounds: a) coarse and refined bounds over time and b) final-time LP-optimized densities under cellwise constraints.}}
\label{fig:1d_nl_prob}
\end{figure}

\rv{For the event-probability demonstration, we use the \(K_{\mathrm{GMM}}=15\) PINN-GMM and pass the resulting \(\hat p\pm B\) to Algorithm~\ref{alg:solvelp} to solve Problem~\ref{prob:integral_pdf} for \(X_{\mathrm{event}}=[1,4]\).}
\rv{Figure~\ref{fig:1d_nl_solver_N500_all_prob} overlays the coarse and refined probability bounds over time. Both discretizations contain the reference probability, and refinement tightens the bounds.}
\rv{The point estimates from GA and GMM deviate noticeably after \(t\approx2\) s, while UT remains closer but provides no analogous error-bound interval.}
\rv{Figure~\ref{fig:1d_nl_solver_N500_pdf_t5.00} shows the refined LP solution at the final time; dashed black lines indicate the per-cell constraints $[p_i^-,p_i^+]$, and solid blue/red lines denote the LP-optimized densities attaining the lower/upper event-probability bounds.
The refined computation takes 0.18 seconds per evaluated time point on average.}

\FloatBarrier
\subsection{6D Unperturbed Keplerian Orbit}
\rv{We next consider a six-dimensional unperturbed Keplerian orbit, adapted from Ref.~\cite{sun2016uncertainty}, to evaluate PINN-GMM in a high-dimensional nonlinear orbital setting.}
\rv{This group contains two studies. First, the equinoctial-coordinate case admits an analytical push-forward reference solution and is used to evaluate uncertainty propagation, error bounds, and event-probability bounds via Algorithm~\ref{alg:solvelp}. Second, the same orbit is used to study the effect of coordinate choice by 
comparing PINN-GMMs trained in rotating spherical and Cartesian coordinates.}
\rv{To estimate the empirical total variation \(\widetilde{\mathrm{TV}}\) and worst normalized error \(\widetilde{\mathrm{WNE}}\), we draw \(10^9\) i.i.d. uniform samples over the state domain at each evaluation time. The empirical relative divergence \(\widetilde{\mathrm{RD}}\) is computed directly using the available reference samples.}

\subsubsection{Equinoctial Orbit Elements Coordinates}
In equinoctial elements, the unperturbed Keplerian \rv{flow admits an analytical push-forward solution of the FP-PDE; see Appendix~\ref{appendix:oe_dyn}.} 
Because the flow is nonlinear, \rv{an initially Gaussian density generally becomes non-Gaussian for \(t>0\)}, making this a useful high-dimensional nonlinear testbed.
\rv{The initial PDF is Gaussian, $p_0=\mathcal{N}(x;\mu_0,\Sigma_0)$, with $\mu_0$ and $\Sigma_0$ reported in Appendix~\ref{appendix:6d_equin}.}
We consider \rv{$T_{\mathrm{end}} = 25849$ seconds (about $7.2$ hours) and compare PINN-GMM with GA, UT, and GMM.}
Further setup details appear in Appendix~\ref{appendix:6d_equin}.

\begin{figure}[hbtp!]
\centering
\begin{minipage}{0.32\textwidth}\centering
  \begin{subfigure}[t]{\linewidth}
    \includegraphics[width=\linewidth]{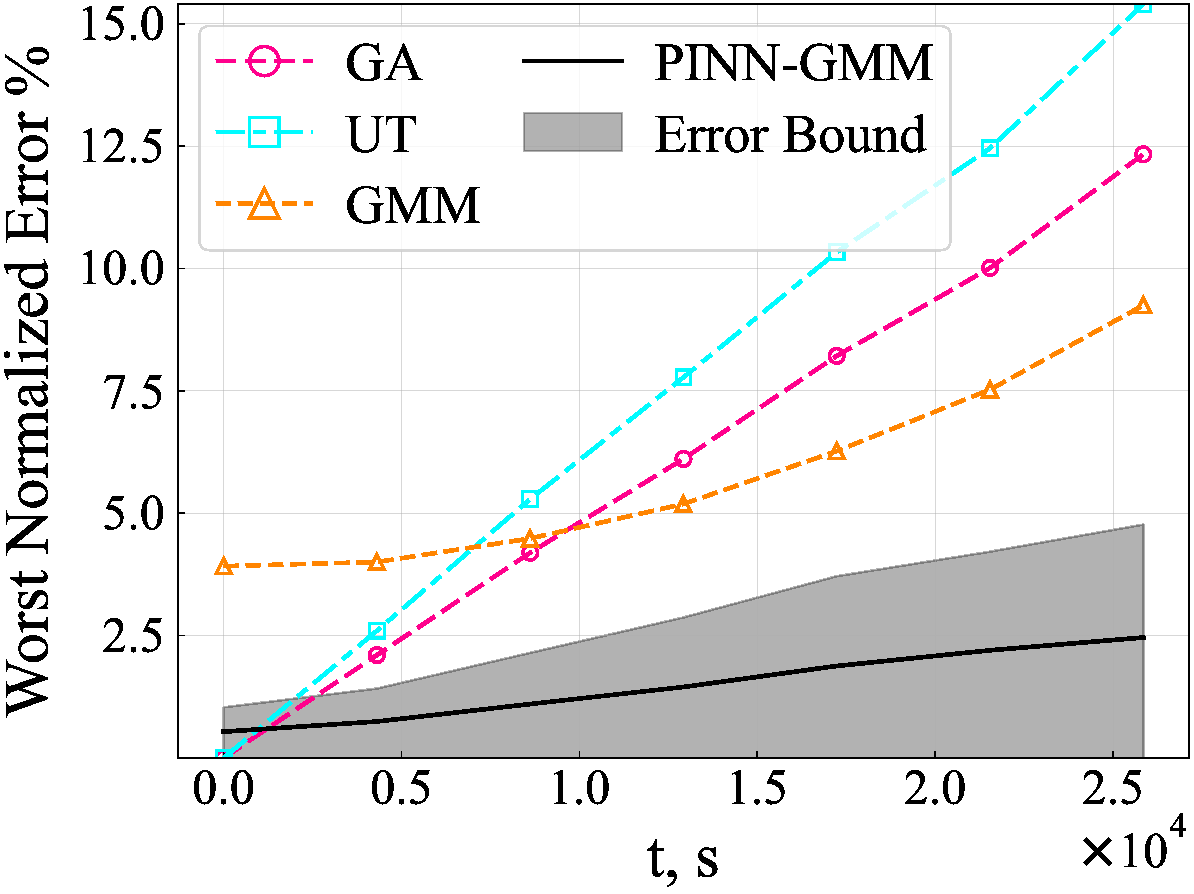}
    \caption{Worst normalized error \%}\label{fig:6d_geo_norm_error}
  \end{subfigure}
\end{minipage}
\begin{minipage}{0.32\textwidth}\centering
  \begin{subfigure}[t]{\linewidth}
    \includegraphics[width=\linewidth]{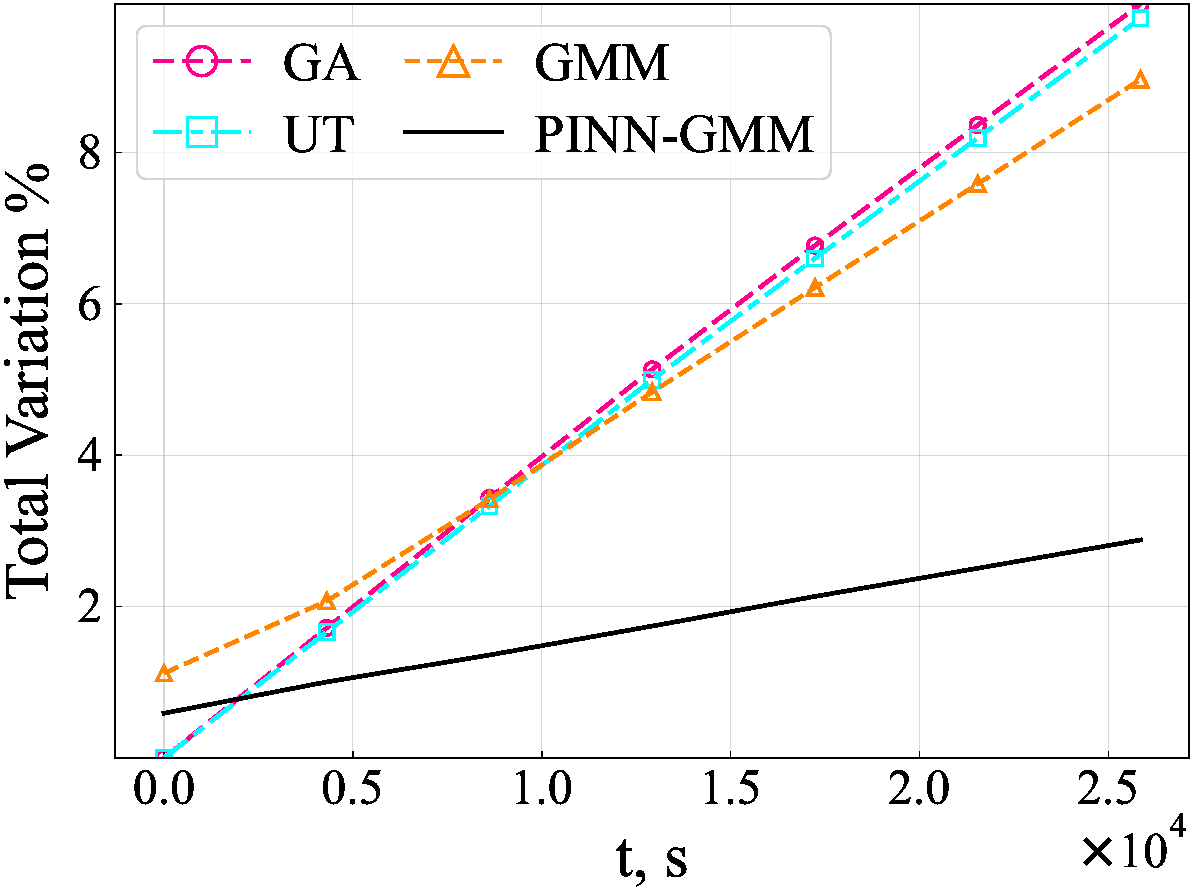}
    \caption{Total variation \%}\label{fig:6d_geo_metric-TV}
  \end{subfigure}
\end{minipage}
\begin{minipage}{0.32\textwidth}\centering
  \begin{subfigure}[t]{\linewidth}
    \includegraphics[width=\linewidth]{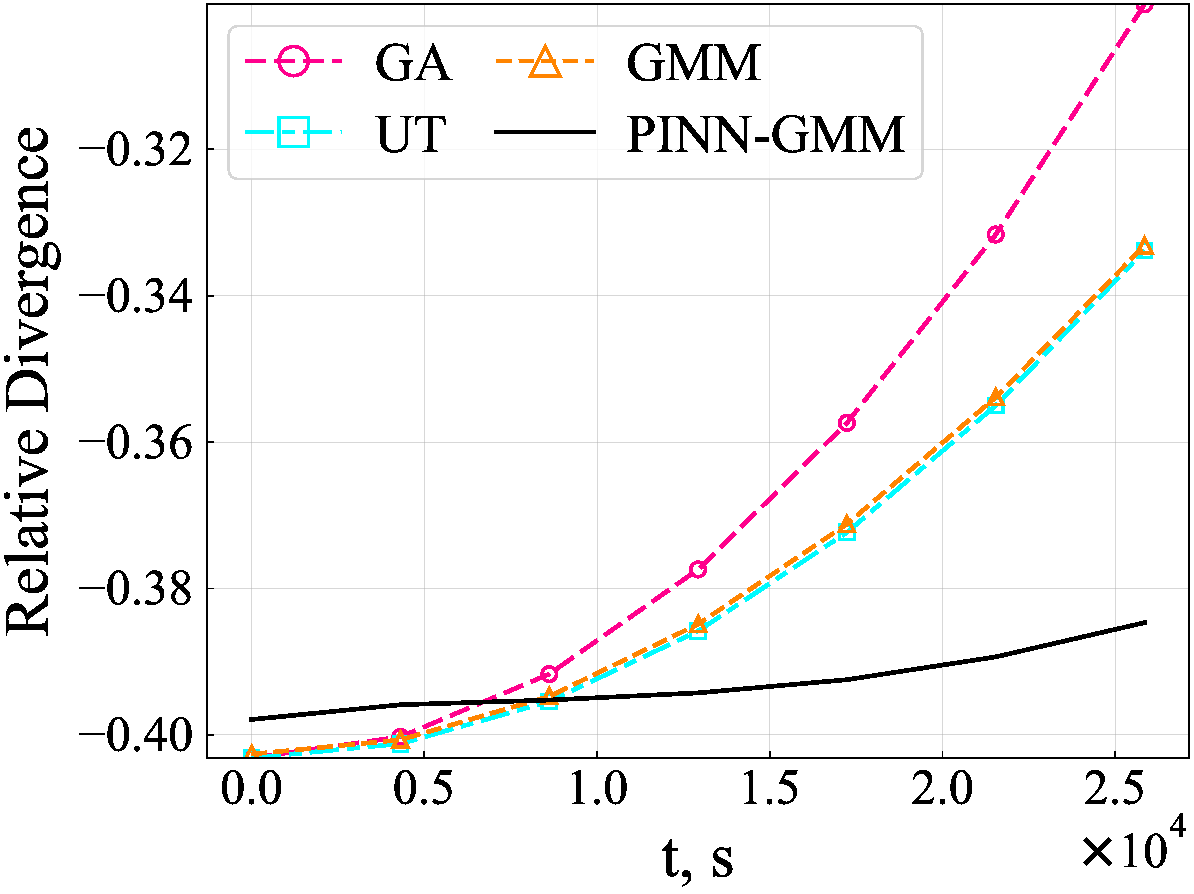}
    \caption{Relative divergence}\label{fig:6d_geo_rd}
  \end{subfigure}
\end{minipage}
\caption{\rv{
6D equinoctial Keplerian error evaluation for PINN-GMM and baselines: a) $\widetilde{\mathrm{WNE}}$ with the PINN-GMM bound $B(t)$, b) $\widetilde{\mathrm{TV}}$, and c) $\widetilde{\mathrm{RD}}$ against the analytical reference.
}}
\label{fig:6d_geo}
\end{figure}
Figures~\ref{fig:6d_geo_metric-TV},
\ref{fig:6d_geo_norm_error}, and~\ref{fig:6d_geo_rd} report the evaluation measures over time, where PINN-GMM consistently attains \rv{smaller errors than the baselines.}
In Fig.~\ref{fig:6d_geo_norm_error}, \rv{the PINN-GMM error bound contains the worst normalized error and remains reasonably close; the error bound stays below approximately \(5.0\%\) over the horizon.}
\rv{The bound \(B(t)\) is a worst-case pointwise PDF error bound, so it does not directly bound total variation or relative divergence.}

\begin{figure}[hbtp!]
\centering
\begin{minipage}{0.41\textwidth}\centering
  \begin{subfigure}[t]{\linewidth}
    \includegraphics[width=\linewidth]{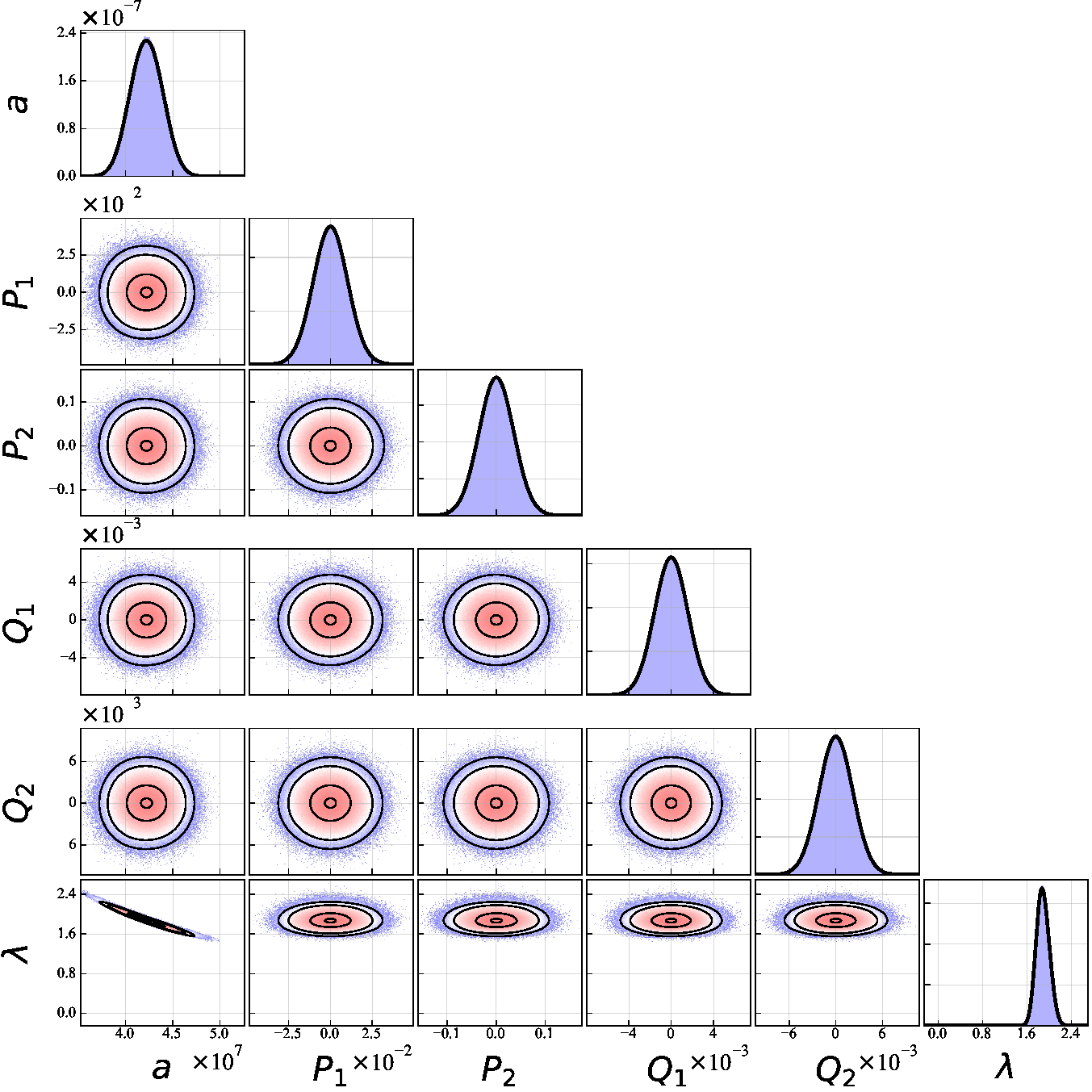}
    \caption{Full corner plot at $t=T_{\mathrm{end}}$}\label{fig:6d_geo_fullcorner_t0}
  \end{subfigure}
\end{minipage}
\begin{minipage}{0.54\textwidth}\centering
  \begin{subfigure}[t]{\linewidth}
    \includegraphics[width=\linewidth]{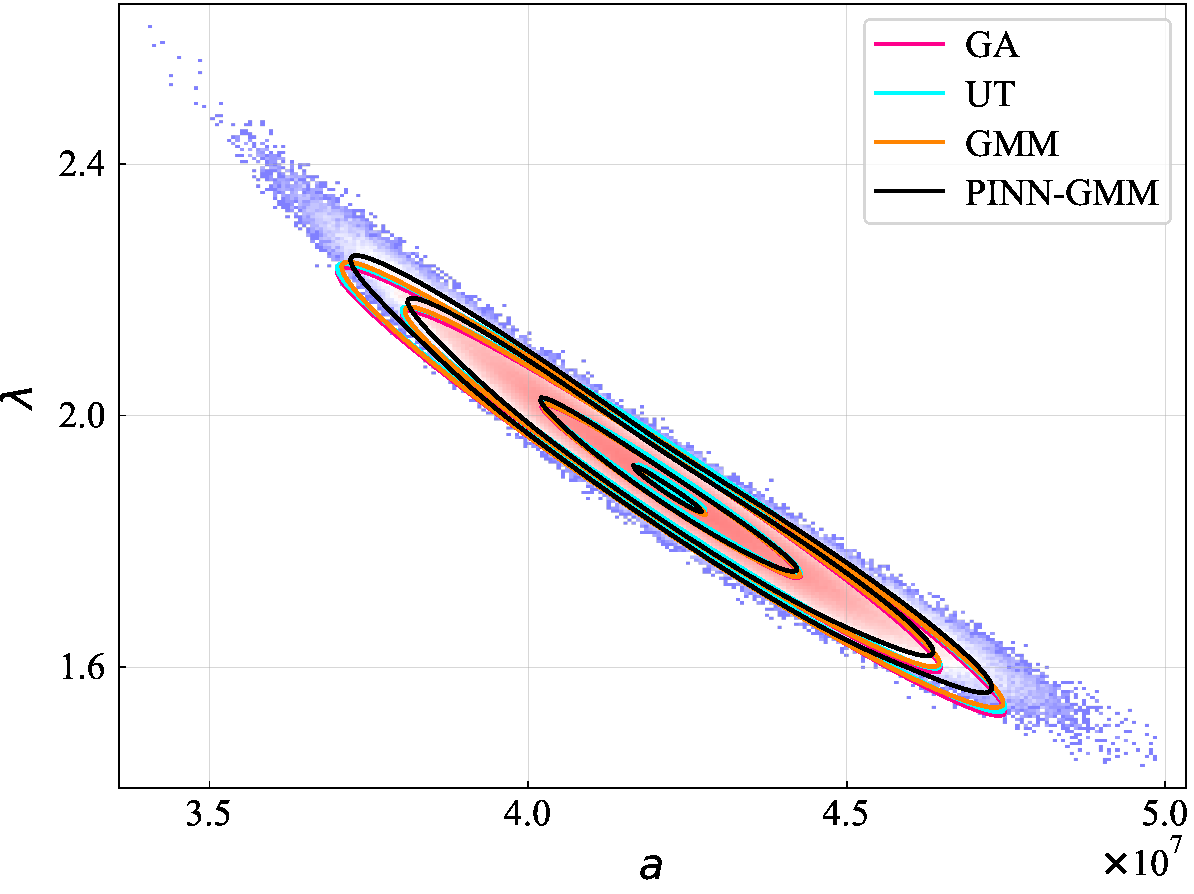}
    \caption{Enlarged $(a,\lambda)$ marginal at $t=T_{\mathrm{end}}$}
    \label{fig:6d_geo_element}
  \end{subfigure}
\end{minipage}
\caption{\rv{Final-time 6D Keplerian PDF marginals: a) PINN-GMM contours against MC samples and b) enlarged $(a,\lambda)$ comparison with baselines.}}
\label{fig:6d_geo_fullcorner}
\end{figure}
Figure~\ref{fig:6d_geo_fullcorner} show 1D and 2D marginals of the PINN-GMM at \rv{the final time $T_{\mathrm{end}}$.}
The black contours of PINN-GMM align well with the reference density (blue--red heat map).
\rv{Since $\lambda$ is the only time-varying state and its evolution depends on $a$, Fig.~\ref{fig:6d_geo_element} shows the \((a,\lambda)\) marginal at the final time.
This enlarged comparison shows that PINN-GMM remains close to the reference, while the baseline contours exhibit relatively larger deviations.}

\rv{Figure~\ref{fig:case1_equin_prob} reports the event-probability computation for this 6D case. Figure~\ref{fig:case1_equin_prob_est} shows point estimates from the reference solution and the tested propagation methods, while Fig.~\ref{fig:case1_equin_solver_Nx1_Ndeg2000_GMMguide1_Cheap0(Bounds)} shows the reference probability together with the PINN-GMM upper and lower bounds from Algorithm~\ref{alg:solvelp}.}
We instantiate Line~1 in Algorithm~\ref{alg:solvelp} \rv{with a guided partitioning scheme, placing finer cells near the event region and in high-density regions of the PINN-GMM approximation.
This bound computation required \(37.2\) seconds per evaluated time point on average.}
For \rv{the} chosen $X_{\mathrm{event}}$, \rv{the point estimates are comparable, but only PINN-GMM provides an \textit{a priori} interval that quantifies the possible deviation from the true event probability.}
\rv{Moreover, T}he posterior empirical total variation in Fig.~\ref{fig:6d_geo_metric-TV} implies that there exist events $X_{\mathrm{event}} \subseteq X$ for which baseline estimates may deviate from the truth by roughly $2$--$9\%$.
\begin{figure}[hbtp!]
\centering
\begin{minipage}{0.45\textwidth}\centering
  \begin{subfigure}[t]{\linewidth}
    \includegraphics[width=\linewidth]{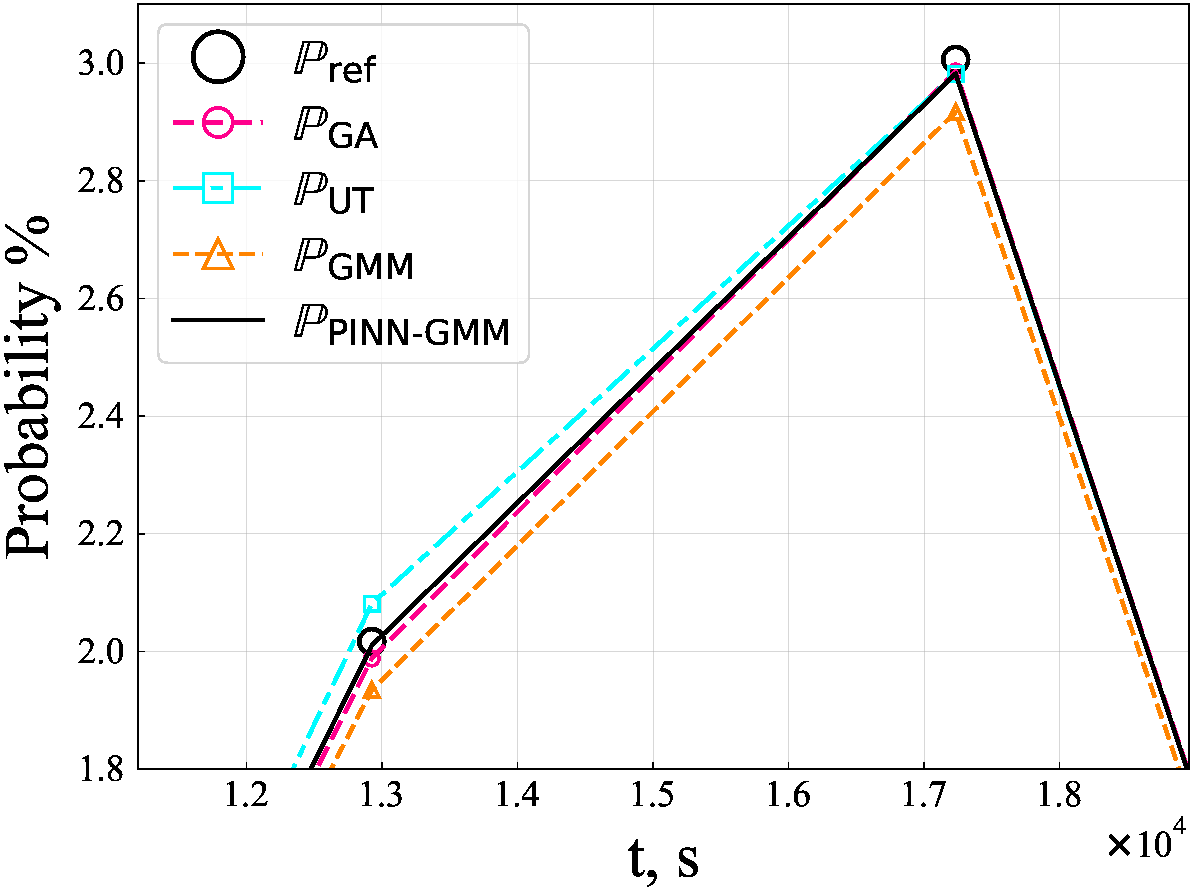}
    \caption{\rv{Point estimates over a zoom-in time window}}\label{fig:case1_equin_prob_est}
  \end{subfigure}
\end{minipage}
\begin{minipage}{0.45\textwidth}\centering
  \begin{subfigure}[t]{\linewidth}
    \includegraphics[width=\linewidth]{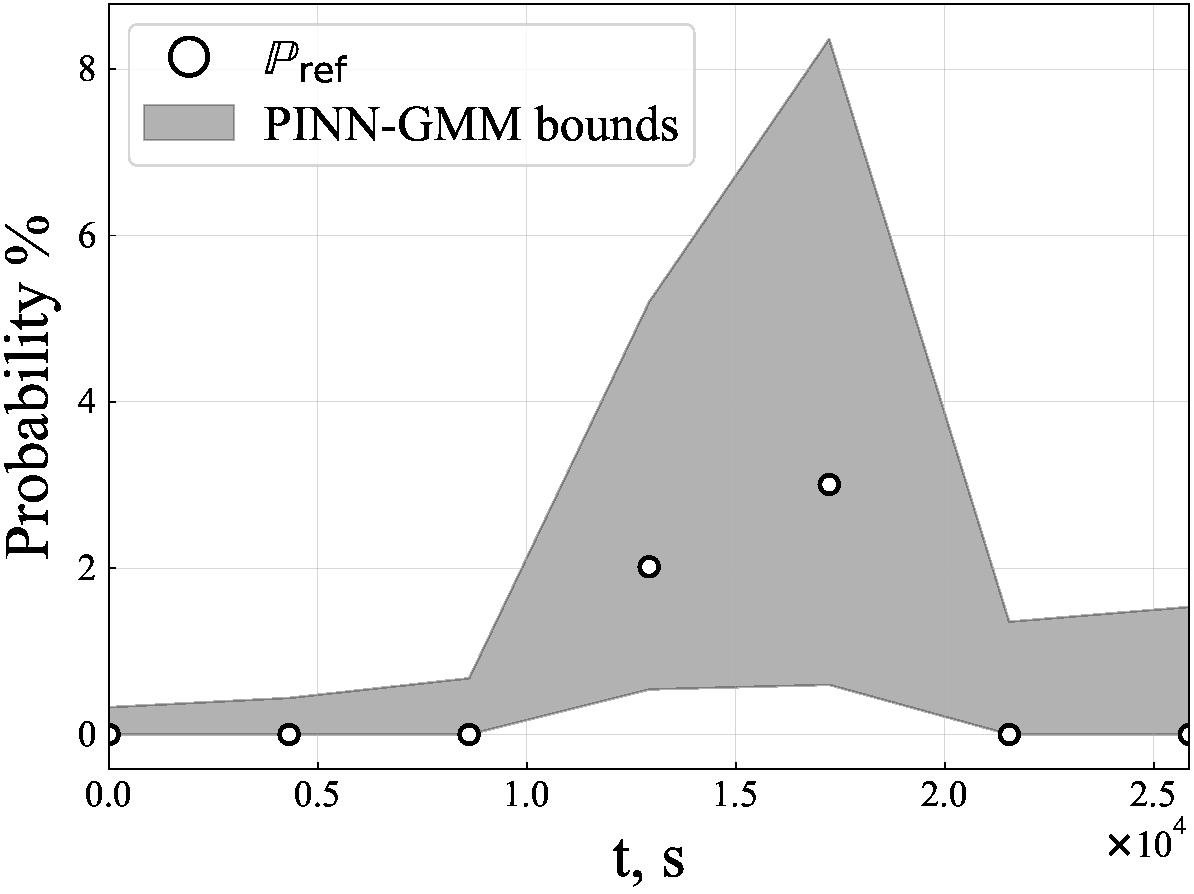}
    \caption{\rv{Reference and PINN-GMM bounds over $t$}}
    \label{fig:case1_equin_solver_Nx1_Ndeg2000_GMMguide1_Cheap0(Bounds)}
  \end{subfigure}
\end{minipage}
\caption{
\rv{6D equinoctial event-probability results: a) method point estimates and b) reference probability with PINN-GMM bounds from Algorithm~\ref{alg:solvelp}.}}
\label{fig:case1_equin_prob}
\end{figure}

\subsubsection{Effect of Coordinate Choice}\label{subsec:6D_coord}
\rv{We next use the same 6D unperturbed Keplerian orbit to examine the effect of coordinate choice.}
In rotating spherical coordinates, the FP-PDE has no closed-form solution. 
We therefore construct a reference density by propagating \(10^7\) samples to selected times and fitting a Gaussian mixture to the resulting samples (Sec.~\ref{subsec:ref_true_pdf}).
For efficiency, we map the states to orbital elements, evolve them analytically, and map back---a semi-analytical Monte Carlo specific to this unperturbed Keplerian case.
We \rv{use the same time horizon as in the equinoctial-coordinate case;} see Appendix~\ref{appendix:6d_sph} for additional details.

\rv{Figure~\ref{fig:6d_geosph} reports the rotating-spherical-coordinate error measures.}
\begin{figure}[hbt!]
\centering
\begin{minipage}{0.32\textwidth}\centering
  \begin{subfigure}[t]{\linewidth}
    \includegraphics[width=\linewidth]{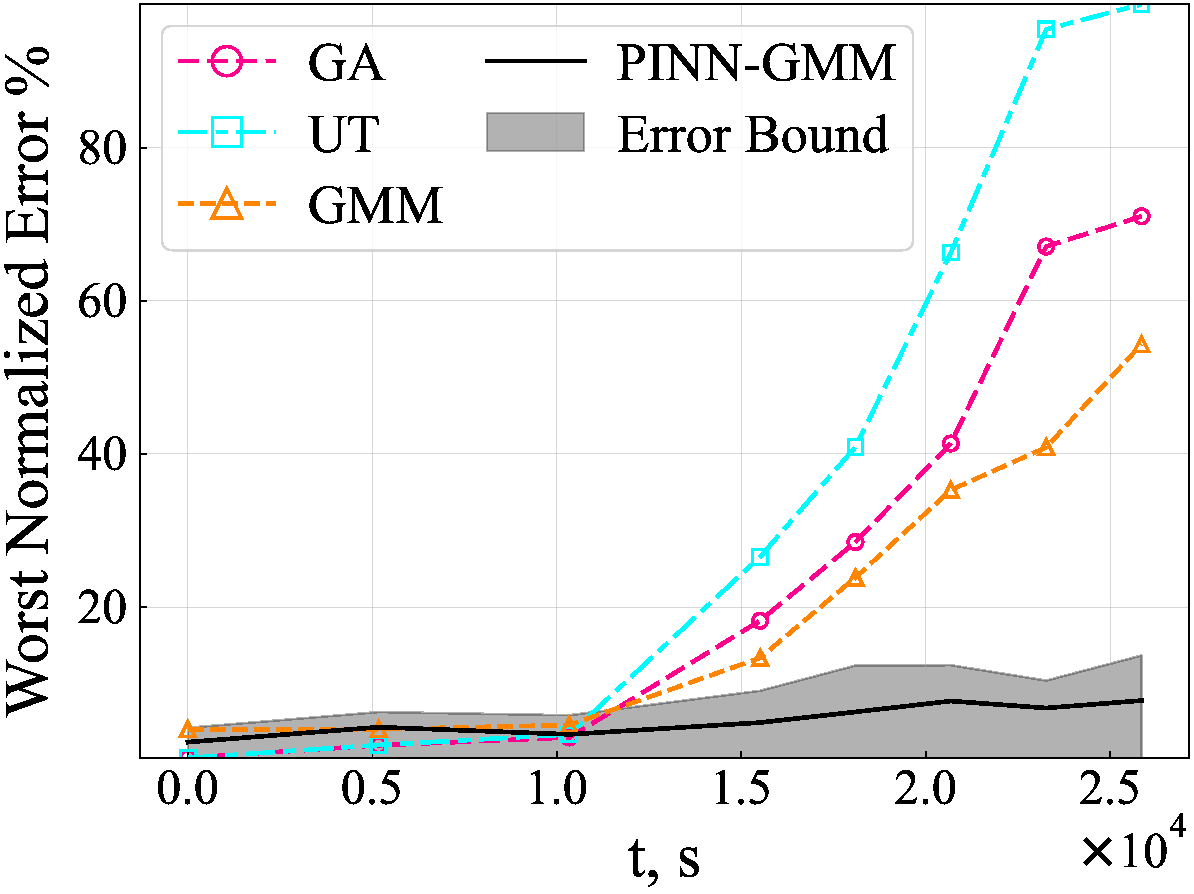}
    \caption{Worst normalized error \%}\label{fig:6d_geosph_norm_error}
  \end{subfigure}
\end{minipage}
\begin{minipage}{0.32\textwidth}\centering
  \begin{subfigure}[t]{\linewidth}
    \includegraphics[width=\linewidth]{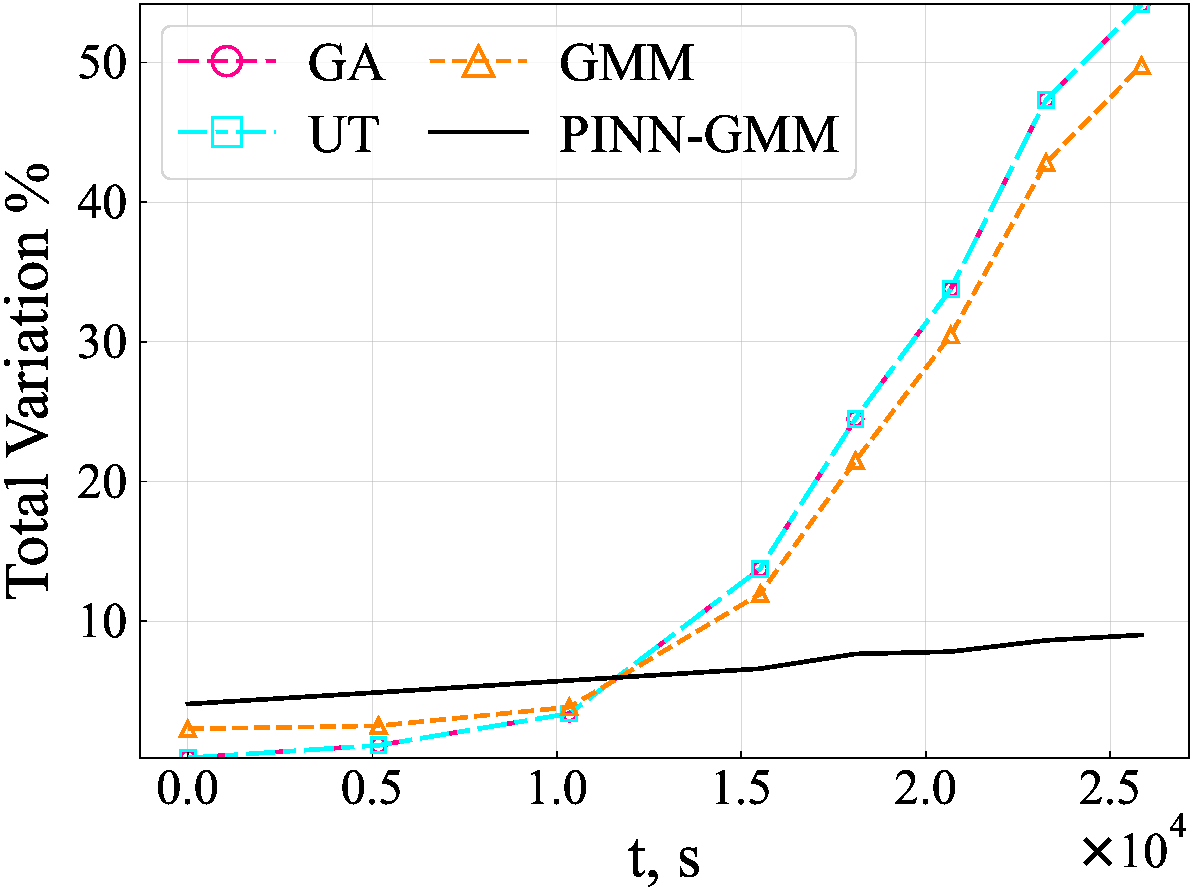}
    \caption{Total variation \%}\label{fig:6d_geosph_metric-TV}
  \end{subfigure}
\end{minipage}
\begin{minipage}{0.32\textwidth}\centering
  \begin{subfigure}[t]{\linewidth}
    \includegraphics[width=\linewidth]{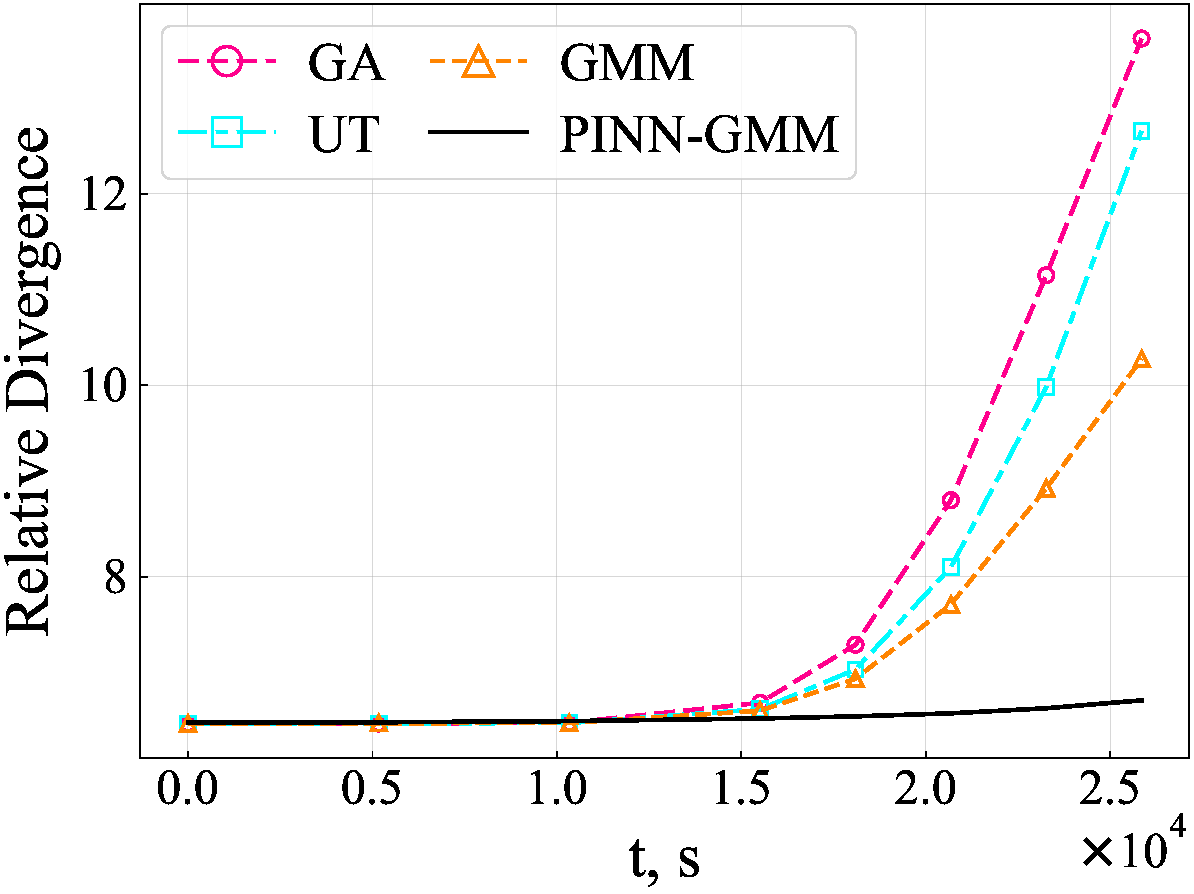}
    \caption{Relative divergence}\label{fig:6d_geosph_rd}
  \end{subfigure}
\end{minipage}
\caption{\rv{
6D rotating-spherical Keplerian error evaluation: a) $\widetilde{\mathrm{WNE}}$ with the PINN-GMM error bound $B(t)$, b) $\widetilde{\mathrm{TV}}$, and c) $\widetilde{\mathrm{RD}}$ against the MC reference.
}}
\label{fig:6d_geosph}
\end{figure}
\begin{figure}[hbtp!]
\centering
\begin{minipage}{0.45\textwidth}\centering
  \begin{subfigure}[t]{\linewidth}
    \includegraphics[width=\linewidth]{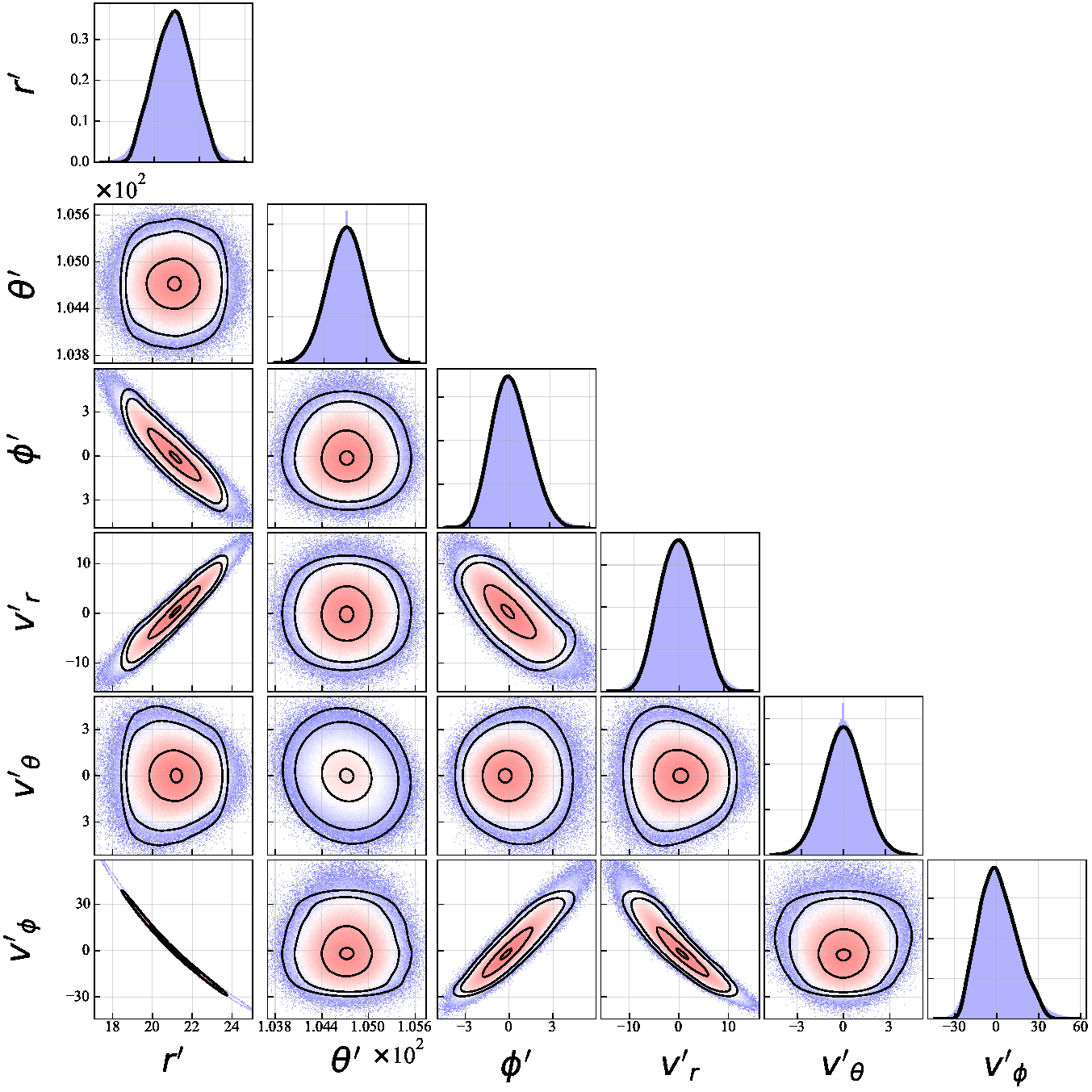}
    \caption{}\label{fig:6d_geosph_fullcorner_T0.30(pinn-gmm)}
  \end{subfigure}
\end{minipage}
\begin{minipage}{0.45\textwidth}\centering
  \begin{subfigure}[t]{\linewidth}
    \includegraphics[width=\linewidth]{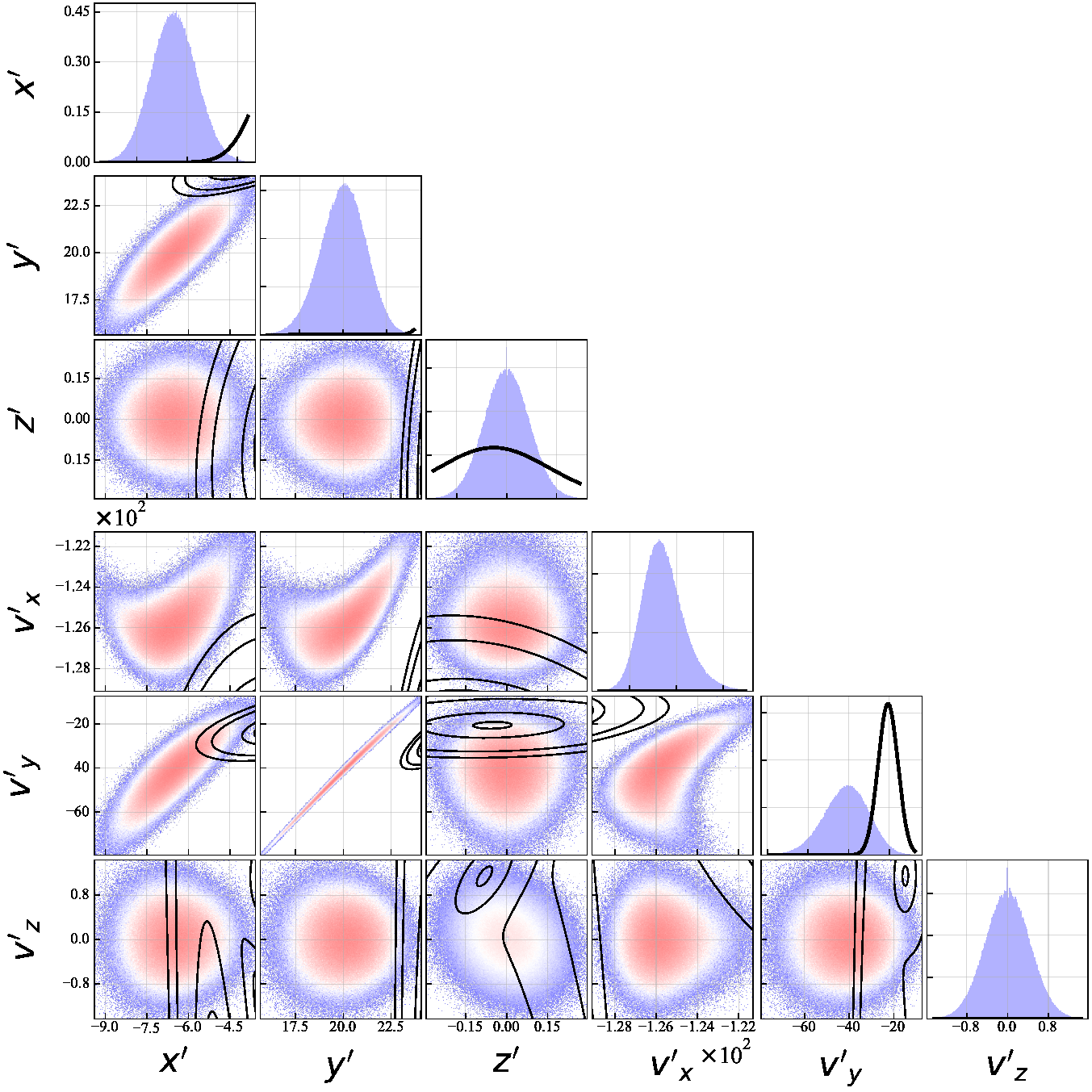}
    \caption{}
    \label{fig:6d_cart_fullcorner_T0.30(pinn-gmm)}
  \end{subfigure}
\end{minipage}
\caption{\rv{Final-time 6D Keplerian PDF marginals for PINN-GMM trained in: a) rotating-spherical and b) Cartesian coordinates, both compared with MC samples.}}
\label{fig:6d_geosph_fullcorner}
\end{figure}
PINN-GMM yields consistently smaller errors than the baselines: its total variation \rv{remains} between 4\%--9\%, whereas \rv{the GA, UT, and GMM errors grow to} 50\%--54\%, by the final time.
\rv{The worst normalized error and relative divergence show the same trend, and Fig.~\ref{fig:6d_geosph_norm_error} shows that the PINN-GMM error bound contains the worst normalized error.}
\rv{Figure~\ref{fig:6d_geosph_fullcorner_T0.30(pinn-gmm)} shows the corresponding marginal PDF at the final time. The PINN-GMM contours closely follow the MC reference samples.}

\rv{To compare the effect of coordinate systems, we also train PINN-GMM in normalized Cartesian coordinates with a matched initial condition; see Appendix~\ref{appendix:6d_cart}.}
\rv{Under the same training budget, the Cartesian-coordinate model does not converge to comparable loss levels: its IC and FP-PDE residual losses are about an order of magnitude larger than those obtained in rotating spherical coordinates; see Appendix~\ref{appendix:6d_cart} for the loss values.}
\rv{This leads to a poor Cartesian-coordinate PDF approximation, as shown in Fig.~\ref{fig:6d_cart_fullcorner_T0.30(pinn-gmm)}, where the final-time PINN-GMM contours deviate visibly from the reference MC samples.}
\rv{Thus, although PINN-GMM is not inherently tied to a particular coordinate system, the coordinate choice substantially affects trainability. Rotating spherical coordinates reduce the effective state-domain size and concentrate residual samples in regions containing PDF information, whereas in Cartesian coordinates the probability mass moves through a much larger domain, making training more difficult for a fixed sampling budget.}

\subsection{4D Perturbed Keplerian Orbit with Diffusion}
\rv{We next consider four-dimensional planar perturbed Keplerian orbits with stochastic diffusion in rotating spherical coordinates.}
\rv{Unlike the 6D unperturbed Keplerian cases, these systems do not admit a closed-form FP-PDE solution, so the reference densities are constructed from fully numerical Monte Carlo simulations as described in Sec.~\ref{subsec:ref_true_pdf}.}
\rv{This group contains two experiments. The first uses the \(J_2\)+diffusion case from Ref.~\cite{sun2016uncertainty} to evaluate uncertainty propagation and demonstrate event-probability bounds via Algorithm~\ref{alg:solvelp}. 
The second adds a non-conservative low-thrust acceleration over a longer time horizon, testing the method in a more general perturbed-dynamics setting.}
\subsubsection{\rv{\(J_2\) Perturbation with Diffusion}}\label{exp:4d_j2_diffusion}
This experiment, adapted from Ref.~\cite{sun2016uncertainty}, considers a 4D planar Keplerian orbit with \(J_2\) perturbation and diffusion driven by Brownian motion.
The dynamics is expressed in rotating spherical coordinates, with the same setup as in Ref.~\cite{sun2016uncertainty}: time horizon \rv{$T_{\mathrm{end}} = 17232.9$} seconds and Gaussian initial distribution.
\rv{The scaling constants in Eq.~\eqref{eq:def_rotsphere} are \(R=2 \times 10^6\) m, \(T_{\mathrm{orb}}=86164.695\) s, $\Theta=0.015$ rad, \(\Phi=0.0387\) rad, and \(W=7.2920644\times 10^{-5}\) rad/s. The diffusion matrix is} \[ \rv{ g(\mathbf{x}(t)) = \begin{bmatrix} \mathbf{0}_{2\times 2}\\ \mathrm{diag}(\sigma_r,\sigma_\phi) \end{bmatrix} \in \mathbb{R}^{4 \times 2}, \qquad \sigma_r=\frac{10^{-4}}{\sqrt{R^2/T_{\mathrm{orb}}^3}},\quad \sigma_\phi=\frac{10^{-11}}{\sqrt{\Phi^2/T_{\mathrm{orb}}^3}}, } \] \rv{which is constant in \(\mathbf{x}(t)\) and enters only through the velocity-state equations.}
Since the corresponding FP-PDE has no closed-form solution, we construct the reference PDF as in Sec.~\ref{subsec:ref_true_pdf} using a fully numerical Monte Carlo (MC) simulation.
\begin{figure}[hbtp!]
\centering
\begin{minipage}{0.32\textwidth}\centering
  \begin{subfigure}[t]{\linewidth}
    \includegraphics[width=\linewidth]{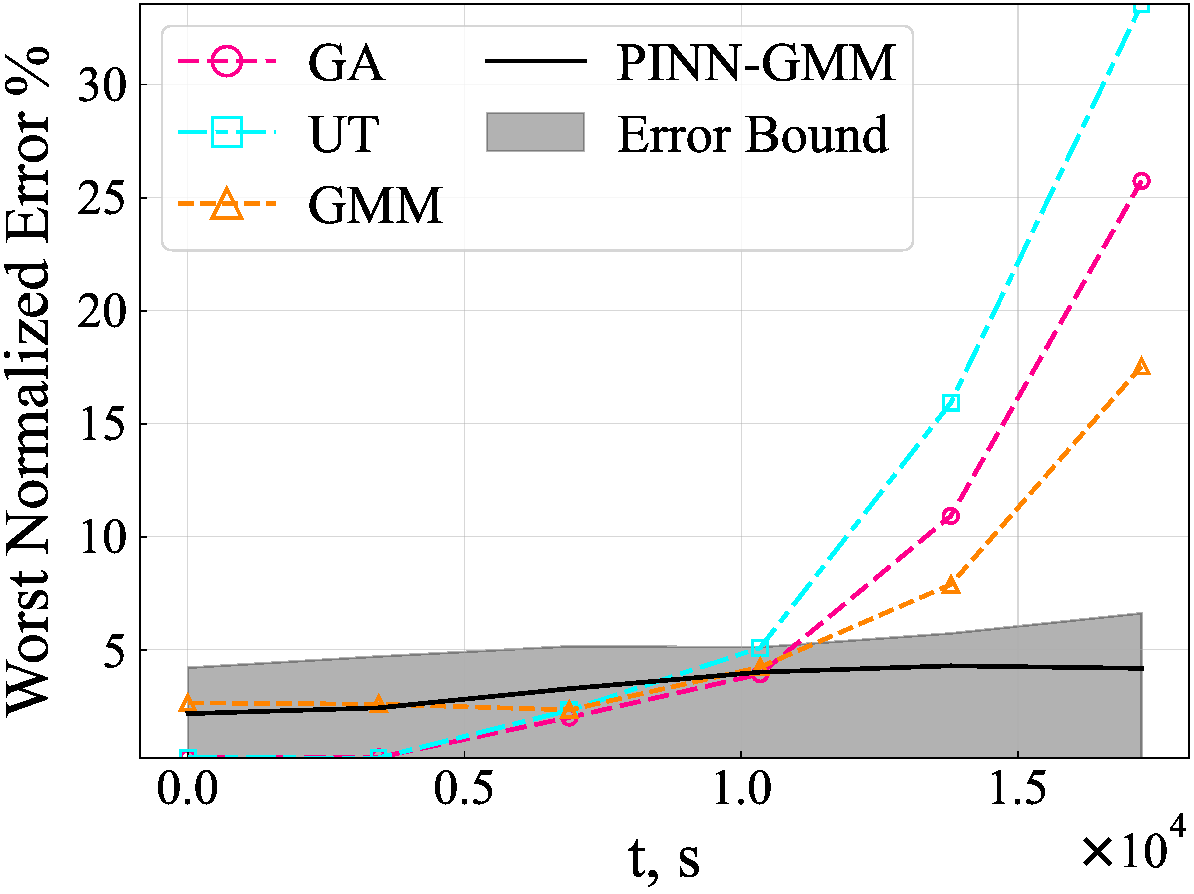}
    \caption{Worst normalized error \%}\label{fig:4d_j2_norm_error}
  \end{subfigure}
\end{minipage}
\begin{minipage}{0.32\textwidth}\centering
  \begin{subfigure}[t]{\linewidth}
    \includegraphics[width=\linewidth]{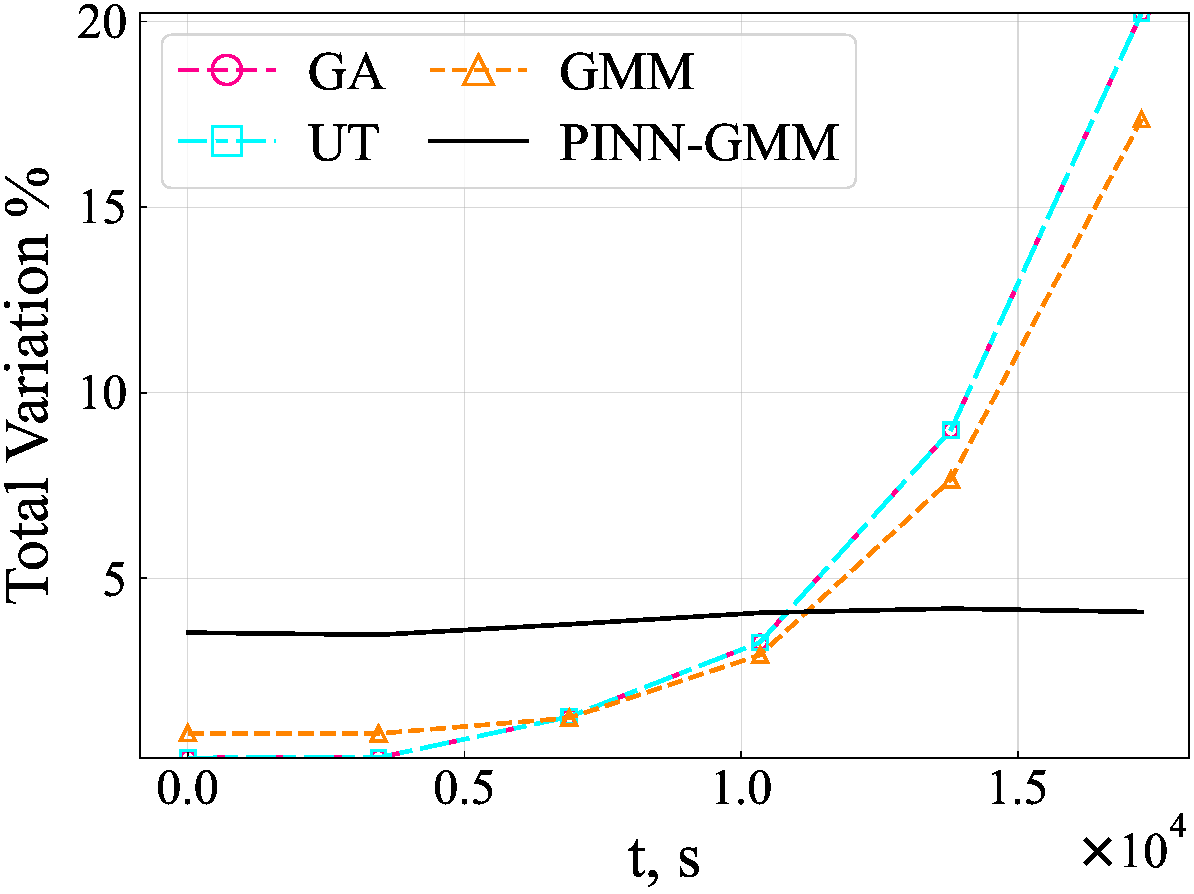}
    \caption{Total variation \%}\label{fig:4d_j2_metric-TV}
  \end{subfigure}
\end{minipage}
\begin{minipage}{0.32\textwidth}\centering
  \begin{subfigure}[t]{\linewidth}
    \includegraphics[width=\linewidth]{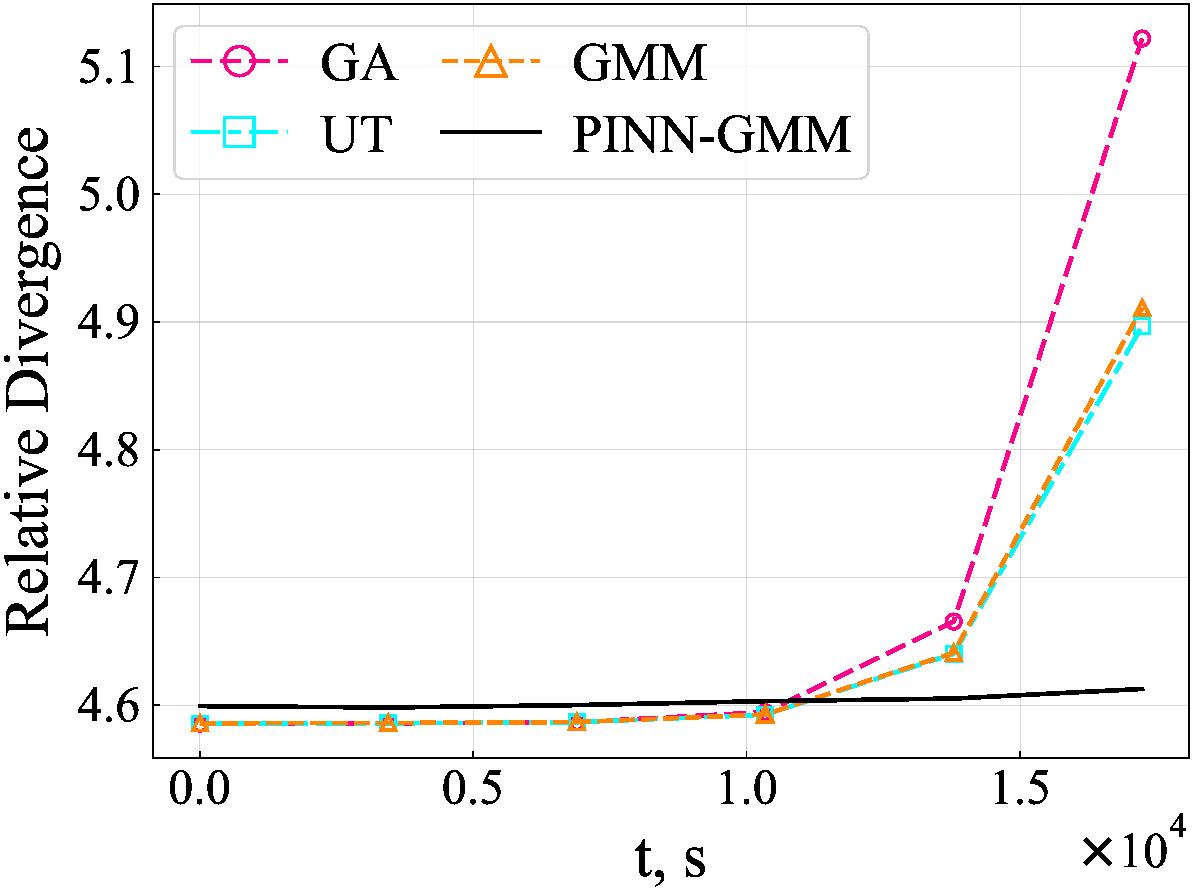}
    \caption{Relative divergence}\label{fig:4d_j2_rd}
  \end{subfigure}
\end{minipage}
\caption{\rv{
4D $J_2$+diffusion error evaluation: a) $\widetilde{\mathrm{WNE}}$ with the PINN-GMM error bound $B(t)$, b) $\widetilde{\mathrm{TV}}$, and c) $\widetilde{\mathrm{RD}}$ against the MC reference.
}}
\label{fig:4d_j2}
\end{figure}
For this MC, propagating $10^6$ samples took about $36375$ seconds, whereas training PINN-GMM $\hat p$ over 40000 iterations took $7246$ seconds and training its error-PINN $\hat e$ over 20000 iterations took an additional $10566$ seconds.
Further details are given in Appendix~\ref{appendix:4d_sph}.

Figure~\ref{fig:4d_j2} \rv{reports the errors of} PINN-GMM \rv{and} the baseline methods.
\rv{A}lthough PINN-GMM has larger errors at early times, its errors remain relatively small over the full horizon \rv{under all three evaluation measures}.
Figure~\ref{fig:4d_j2_norm_error} also \rv{shows} that the PINN-GMM \rv{error} bounds \rv{contains the observed worst normalized error; the baselines provide no analogous error quantification}.
The large\rv{r} initial error of PINN-GMM stems from \rv{the fact that the baselines start from the exact initial distribution, whereas PINN-GMM learns the initial density through the IC loss.}

\rv{Figure~\ref{fig:4d_j2_cart_pdf} maps the PINN-GMM density learned in rotating spherical coordinate to the Cartesian $X$-$Y$ plane (see Sec.~\ref{subsubsec:rsc}) and compares the resulting marginal PDF with the reference MC density at uniform time points. The PINN-GMM surface closely follow the MC density (wireframes) throughout the horizon, capturing both the drift and spreading of the probability mass. This agreement provides an additional coordinate-transformed validation of the learned density, complementing the quantitative error measures in Fig.~\ref{fig:4d_j2}.}
\begin{figure}[htbp!]
\centering
\includegraphics[width=0.72\linewidth]{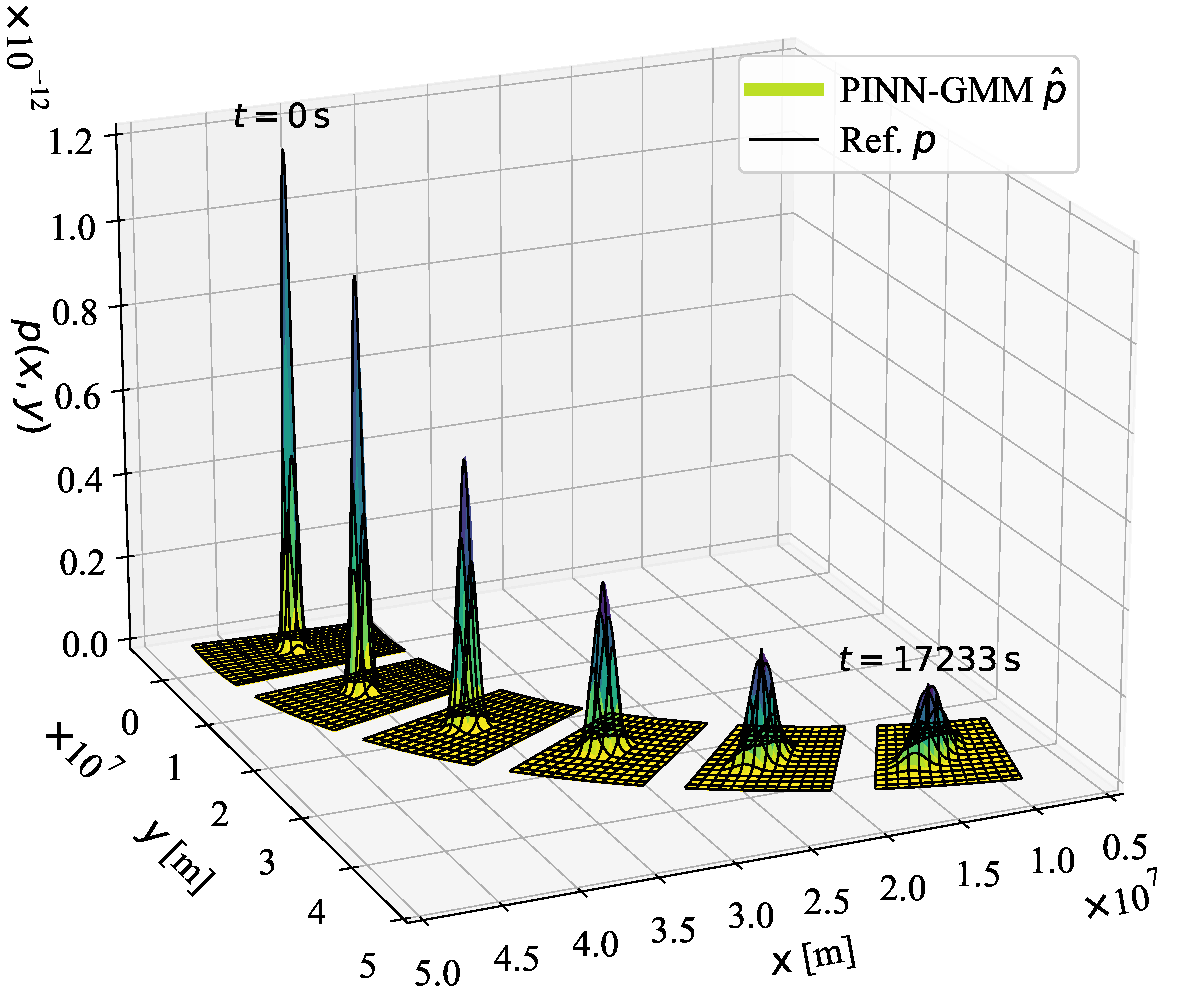}
\caption{\rv{
Cartesian $X$-$Y$ PDF evolution for the 4D $J_2$+diffusion case: PINN-GMM learned in rotating spherical coordinates, transformed, and compared with MC references.
}}
\label{fig:4d_j2_cart_pdf}
\end{figure}
\begin{figure}[ht!]
\centering
\begin{minipage}{0.45\textwidth}\centering
  \begin{subfigure}[t]{\linewidth}
    \includegraphics[width=\linewidth]{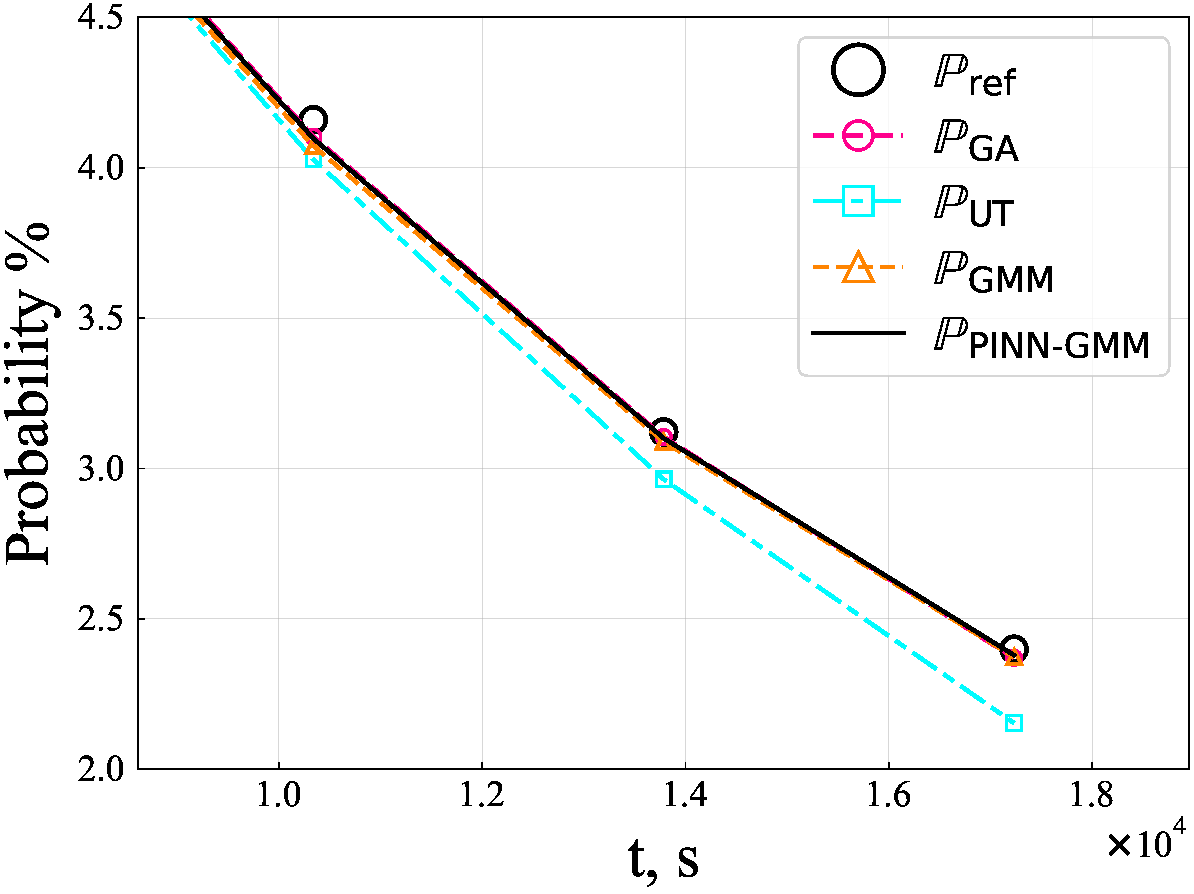}
    \caption{\rv{Point} estimates over a zoom-in time window}
    \label{fig:case2_j2_prob_est}
  \end{subfigure}
\end{minipage}
\begin{minipage}{0.45\textwidth}\centering
  \begin{subfigure}[t]{\linewidth}
    \includegraphics[width=\linewidth]{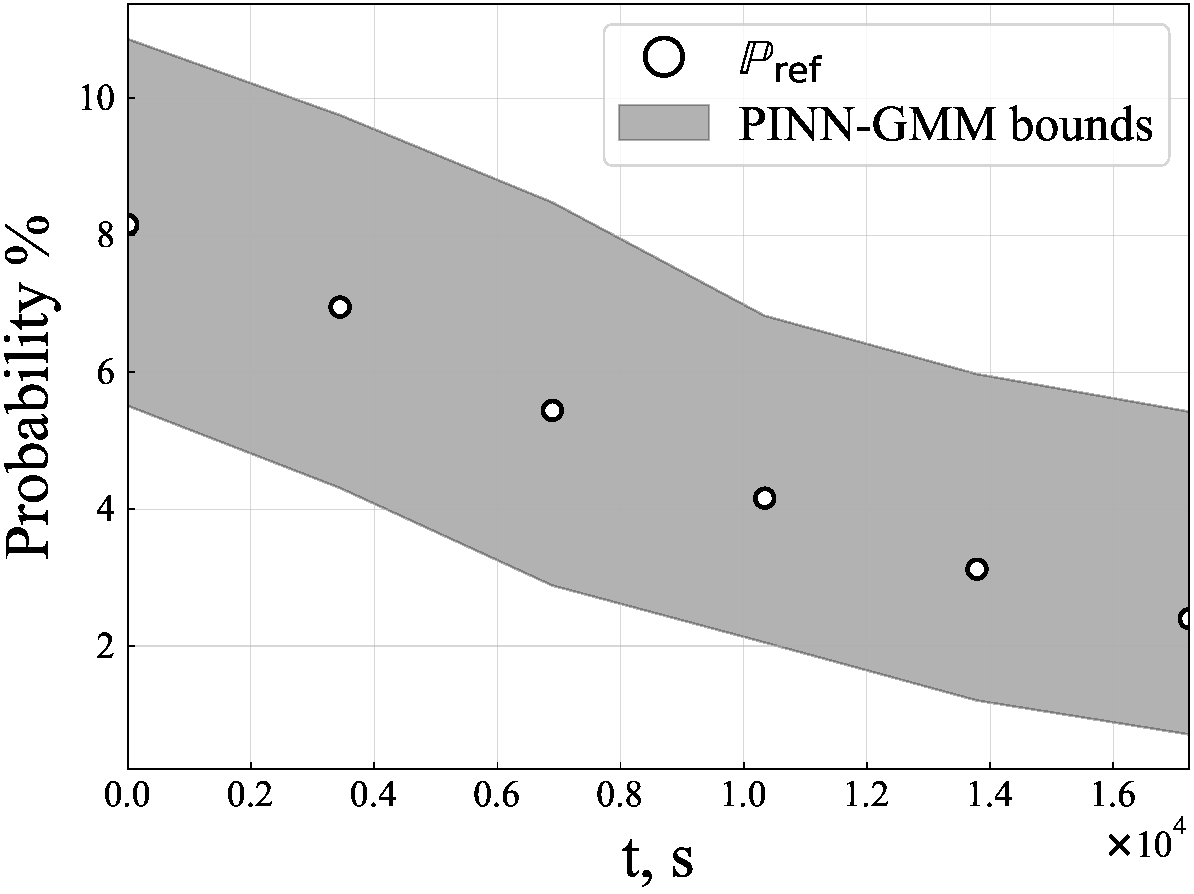}
    \caption{\rv{Reference and PINN-GMM bounds over $t$}}
    \label{fig:case2_j2_prob_bounds}
  \end{subfigure}
\end{minipage}
\caption{
\rv{4D $J_2$+diffusion event-probability results: a) method point estimates and b) reference probability with PINN-GMM bounds from Algorithm~\ref{alg:solvelp}.}}
\label{fig:case2_j2_prob}
\end{figure}

\rv{We use this \(J_2\)+diffusion case as the 4D demonstration of event-probability bounding via Algorithm~\ref{alg:solvelp}.}
\rv{Figure~\ref{fig:case2_j2_prob_est} shows the point estimates, and Fig.~\ref{fig:case2_j2_prob_bounds} overlays the PINN-GMM upper and lower bounds. This bound computation required 21.9 seconds on average for this case.
As before, only PINN-GMM supplies a certified deviation interval.
Compared with the 6D event-probability example in Fig.~\ref{fig:case1_equin_solver_Nx1_Ndeg2000_GMMguide1_Cheap0(Bounds)}, the resulting probability interval is tighter, illustrating that Algorithm~\ref{alg:solvelp} can be less conservative in the lower-dimensional 4D setting.}

\subsubsection{\(J_2\) Perturbation, Diffusion, and Low Thrust}
\begin{figure}[hbtp!]
\centering
\begin{minipage}{0.32\textwidth}\centering
  \begin{subfigure}[t]{\linewidth}
    \includegraphics[width=\linewidth]{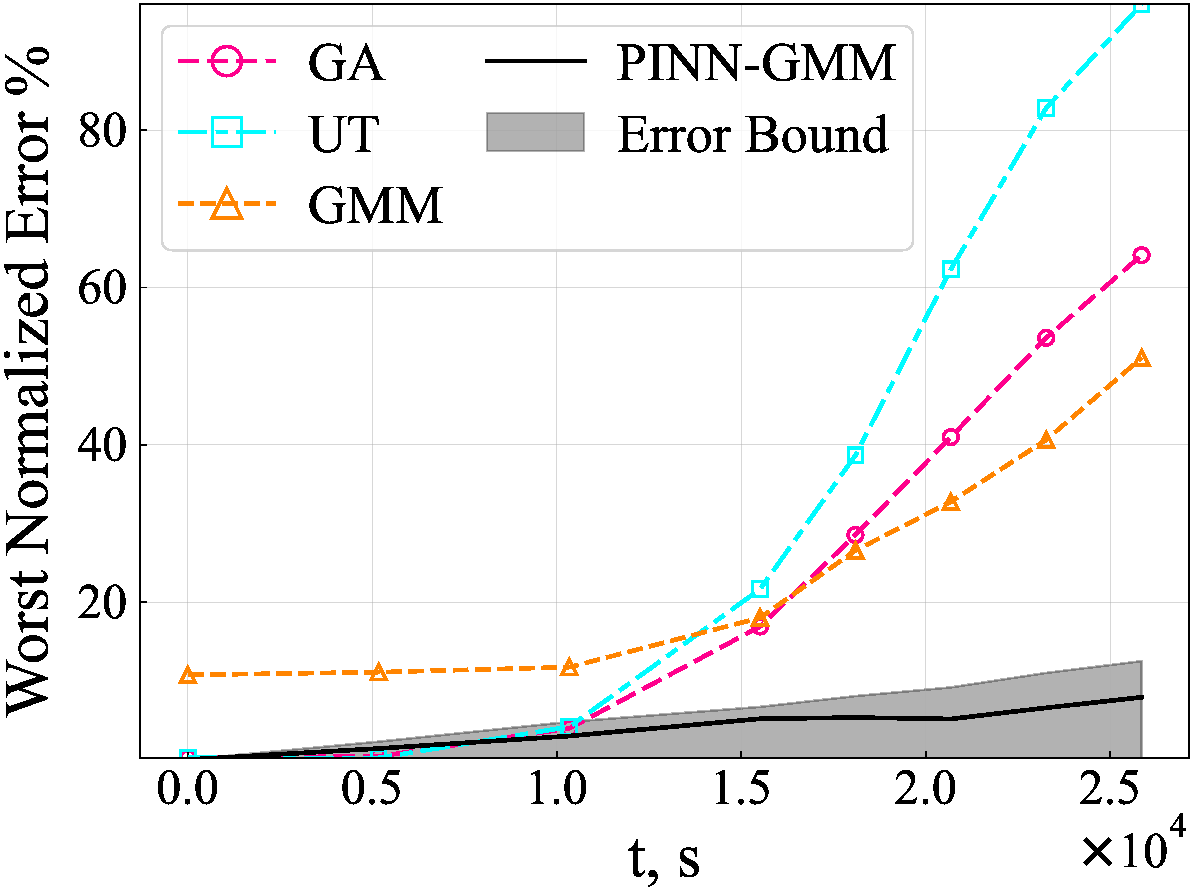}
    \caption{Worst normalized error \%}\label{fig:4d_j2_lowthrust_norm_error}
  \end{subfigure}
\end{minipage}
\begin{minipage}{0.32\textwidth}\centering
  \begin{subfigure}[t]{\linewidth}
    \includegraphics[width=\linewidth]{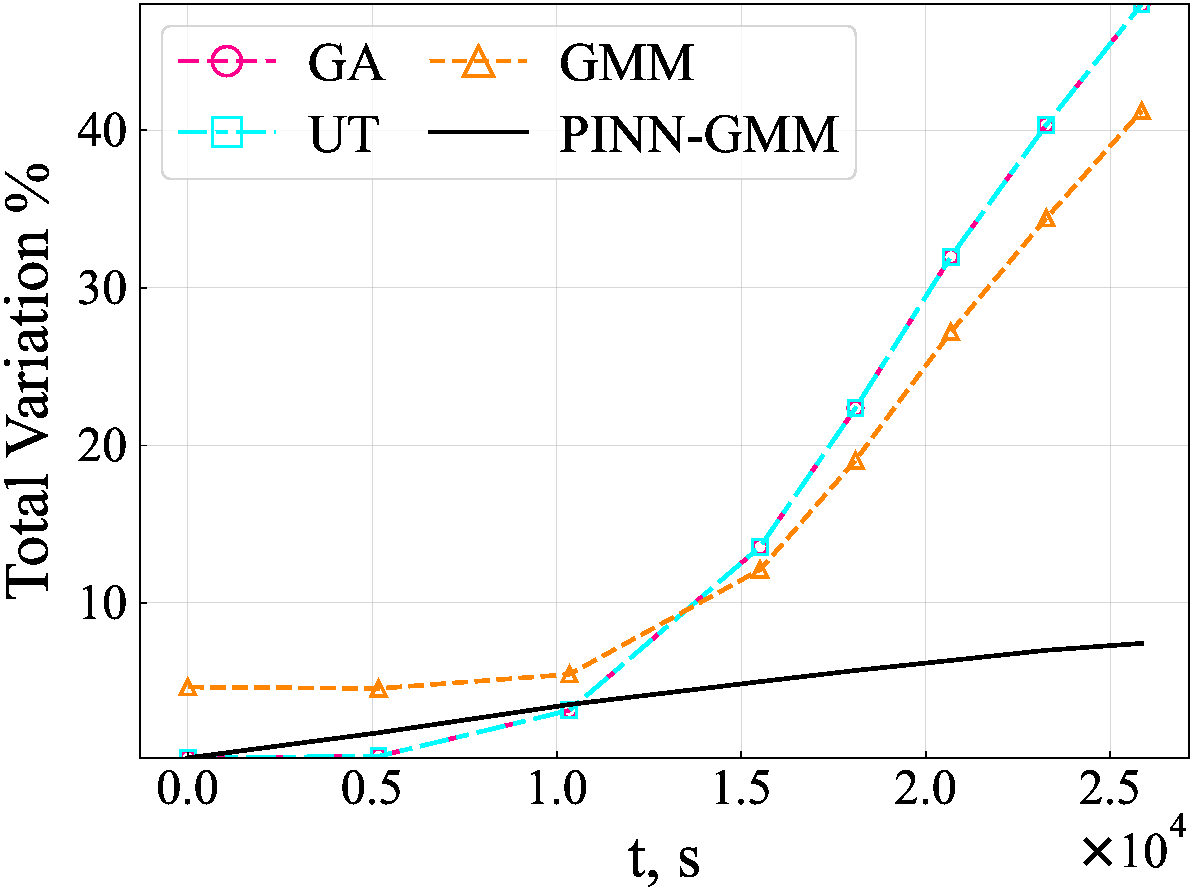}
    \caption{Total variation \%}\label{fig:4d_j2_lowthrust_metric-TV}
  \end{subfigure}
\end{minipage}
\begin{minipage}{0.32\textwidth}\centering
  \begin{subfigure}[t]{\linewidth}
    \includegraphics[width=\linewidth]{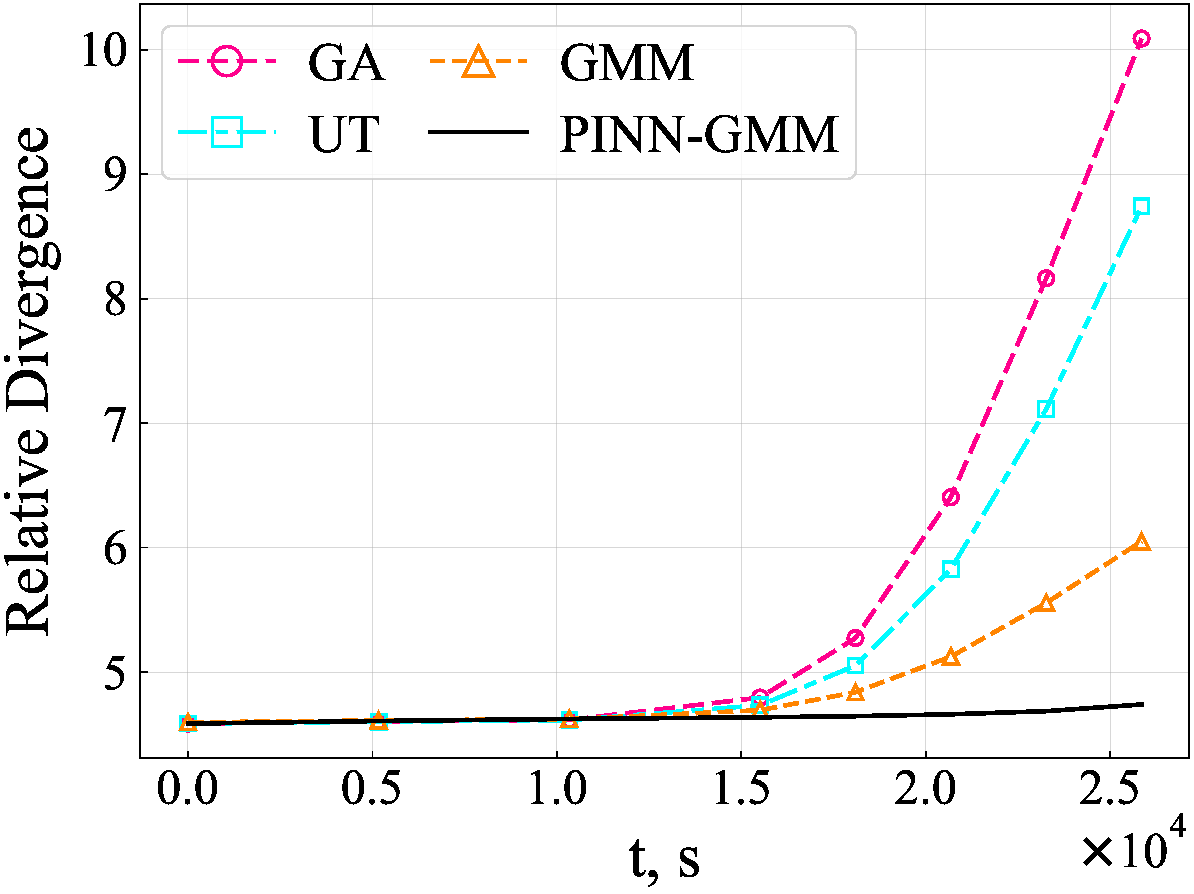}
    \caption{Relative divergence}\label{fig:4d_j2_lowthrust_rd}
  \end{subfigure}
\end{minipage}
\caption{\rv{
4D $J_2$+diffusion+low-thrust error evaluation: a) $\widetilde{\mathrm{WNE}}$ with the PINN-GMM error bound $B(t)$, b) $\widetilde{\mathrm{TV}}$, and c) $\widetilde{\mathrm{RD}}$ against the MC reference.
}}
\label{fig:4d_j2_lowthrust}
\end{figure}
\rv{Next, we add continuous low-thrust acceleration to the 4D \(J_2\)+diffusion dynamics over an extended time horizon \(T_{\mathrm{end}}=25849~\mathrm{s}\).} \rv{The thrust model is most simply described in physical coordinates: the acceleration has constant magnitude \(a_T=0.01\) m$/$s$^2$ and is aligned with the instantaneous in-plane velocity, i.e., \(\mathbf a_{\mathrm{LT}}=a_T(\dot r\,\mathbf e_r+r\dot\phi\,\mathbf e_\phi)/V\), where \(V=\sqrt{\dot r^2+(r\dot\phi)^2}\), and \(\mathbf e_r\) and \(\mathbf e_\phi\) are the radial and increasing-azimuth unit vectors, respectively.}
\rv{Similarly, we construct the reference density using a fully numerical Monte Carlo simulation and compare PINN-GMM with GA, UT, and GMM.}
\rv{For this case, constructing the reference density with \(10^6\) MC samples took \(66200\) seconds.
To mitigate initial-condition error, we constrain the PINN-GMM parameterization to match the initial Gaussian distribution exactly at $t=0$.
We note that this construction applies only when the initial PDF $p_0$ is exactly representable by a GMM with the chosen $K_{\mathrm{GMM}}$; see Appendix~\ref{appendix:4d_lowthrust} for details.
Training PINN-GMM \(\hat p\) over 60000 iterations took \(7229\) seconds; training its error PINN \(\hat e\) over 40000 iterations took \(23211\) seconds.}
\rv{Figure~\ref{fig:4d_j2_lowthrust} reports the evaluation measures.
Thanks to the exact initial-condition matching, PINN-GMM starts with the same initial error as GA and UT.
It then exhibits substantially slower error growth across all measures, while the learned bound $B(t)$ in Fig.~\ref{fig:4d_j2_lowthrust_norm_error} contains the observed worst normalized error over the propagation horizon.}

\subsection{Summary of Orbital Experiments}

\begin{table*}[htb!]
    \centering
    \small
\caption{
\rv{Summary of orbital cases. The columns under ``Error measures'' report the ($\min$, $\max$) range of each measure over the evaluation time horizon. $\widetilde{\mathrm{TV}}$ and $\widetilde{\mathrm{WNE}}$ are reported in percent, and $\widetilde{\mathrm{RD}}$ is dimensionless; $\downarrow$ indicates that smaller values are better. The column ``Offline comp. time'' reports the computation time (in seconds) required to construct the density representation: training time for PINN-based methods, propagation time for GA and UT, and initial GMM fitting plus propagation time for GMM. 
A dagger superscript ($^\dagger$) indicates methods that provide 
a valid error bound $B(t)$.
An asterisk superscript ($^*$) on a PINN-based entry indicates that the training run plateaued before convergence, and the reported results use the best checkpoint obtained. ``N/A'' means that \(\log \hat p\) is not available because the values of \(\hat p\) are below machine precision.}
}
    \renewcommand{\arraystretch}{1.0}
    \setlength{\tabcolsep}{1.5pt} 
    \begin{tabular}{p{2.5cm} l c c c >{\centering\arraybackslash}p{2.0cm}}
      \toprule
      \bfseries Experiments & \bfseries Methods &
      \multicolumn{3}{c}{\bfseries Error \rv{measures} ($\min$, $\max$)} &
      \bfseries \rv{Offline comp. time} 
      \\
      \cmidrule(lr){3-5}
      & & $\widetilde{\mathrm{TV}}$\% $\downarrow$ & $\widetilde{\mathrm{WNE}}$\% $\downarrow$& $\widetilde{\mathrm{RD}}$ $\downarrow$& \\
      \midrule

      \multirow{7}{1.6cm}{\centering \textbf{6D Kepler-Equinoctial orbit elements}}
      & GA            & 
      \rv{($2.38\times 10^{-5}$, 9.95)}
      & \rv{($6.73\times 10^{-5}$, 12.3)} & \rv{($-0.40$,-0.30)}
      & \rv{0.005} \\
      & UT            & \rv{($2.38\times 10^{-5}$, 9.77)}   & \rv{($6.73\times 10^{-5}$, 15.3)} & \rv{(-0.40, -0.33)} & \rv{0.066} \\
      & GMM           & \rv{(1.11, 8.96)}   & \rv{(3.91, 9.24)}   & \rv{(-0.40, -0.33)} & \rv{1520} \\
      & PINN-MLP (vanilla)         & \rv{(0.74, 4.12)}   & \rv{(0.19, 4.54)}   & \rv{(-0.42, -0.39)} & 8033 
      \\
      & PINN-GMM (vanilla)       & \rv{(17.0, 81.8)} & \rv{(11.6, 100)} &  \rv{(-0.24, 15.3)} & 2231$^*$ 
      \\
      & PINN-GMM (no-encoder)      & \rv{(65.8, 100)} & \rv{(413, 499)}&  \rv{(3.1, 76.0)} & 1278$^*$ 
      \\
      & \textbf{PINN-GMM (proposed)$^\dagger$}      & \textbf{\rv{(0.58, 2.87)}}   & \textbf{\rv{(0.53, 2.45)}}   & \textbf{\rv{(-0.39, -0.38)}} & \textbf{7317} 
      \\
      \midrule

      \multirow{7}{2.0cm}{\centering \textbf{6D Kepler-Rotating spherical}}
      & GA            & \rv{(0.22, 54.1)}  & \rv{(0.30, 71.0)}  &  \rv{(6.46, 13.6)}  & \rv{0.683} \\
      & UT            & \rv{(0.22, 54.1)}  & \rv{(0.30, 98.7)}  &  \rv{(6.46, 12.6)}  & \rv{5.016} \\
      & GMM           & \rv{(2.30, 49.7)}  & \rv{(4.03, 54.3)}  &  \rv{(6.46, 10.2)}  & \rv{1000} \\
      & PINN-MLP (vanilla)         & \rv{($4.6\times10^5$, $1.3\times10^6$)} & \rv{(72.9, 88.6)} & \rv{($9.2\times 10^3$, $2.6\times 10^4$)} & 87$^*$ \\
      & PINN-GMM (vanilla)       & \rv{(0.22, 93.7)} & \rv{(0.30, 100)}  & \rv{(6.46, 29.5)}  & 2233$^*$  \\
      & PINN-GMM (no-encoder)      &  \rv{(45.0, 53.4)} & \rv{(100, 100)} & N/A & 137$^*$  \\
      & \textbf{PINN-GMM (proposed)$^\dagger$}      & \textbf{\rv{(4.08, 9.02)}}    & \textbf{\rv{(2.35, 7.78)} }   &  \textbf{\rv{(6.47, 6.70)}}  & \textbf{13074}  \\
      \midrule

      \multirow{7}{2.0cm}{\centering \textbf{4D Kepler-$J_2$+diffusion}}
      & GA            & \rv{(0.16, 20.2)}    & \rv{(0.22, 25.7)}    & \rv{(4.58, 5.12)} & \rv{0.026} \\
      & UT            & \rv{(0.16, 20.2)}    & \rv{(0.22, 33.5)}   & \rv{(4.58, 4.89)} & \rv{0.089} \\
      & GMM           & \rv{(0.81, 17.3)}    & \rv{(2.34, 17.5)}    & \rv{(4.58, 4.91)} & \rv{1402} \\
      & PINN-MLP (vanilla)         & \rv{(1.49, 6.36)}    & \rv{(0.46, 8.55)}    & \rv{(4.59, 4.61)} & 10454  \\
      & PINN-GMM (vanilla)       & \rv{(0.54, 20.3)}    & \rv{(0.54, 25.5)}   & \rv{(4.58, 5.00)} & 3125 \\
      & PINN-GMM (no-encoder)      & N/A & N/A & N/A & 2$^*$ \\
      & \textbf{PINN-GMM (proposed)$^\dagger$}      & \textbf{\rv{(3.47, 4.18)}}    & \textbf{\rv{(2.17, 4.29)}}    & \textbf{\rv{(4.59, 4.61)}} & \textbf{7246}  \\
      \midrule
      \multirow{7}{2.0cm}{\centering \textbf{4D Kepler-$J_2$+diffusion+low-thrust}}
      & GA            & \rv{(0.16, 48.0)}    & \rv{(0.22, 64.0)}    & \rv{(4.58, 10.0)} & \rv{0.04} \\
      & UT            & \rv{(0.16, 48.0)}    & \rv{(0.21, 95.9)}   & \rv{(4.58, 8.74)} & \rv{0.15} \\
      & GMM           & \rv{(4.56, 41.2)}    & \rv{(10.7, 51.0)}    & \rv{(4.59, 6.04)} & \rv{2378} \\
      & PINN-MLP (vanilla)         & \rv{(1.74, 50.1)}    & \rv{(0.89, 100)}    & \rv{(4.60, 53.3)} & \rv{18593}  \\
      & PINN-GMM (vanilla)       & \rv{(0.16, 99.6)}    & \rv{(0.22, 100)}   & \rv{(4.58, 45.2)} & \rv{1.91$^*$}\\
      & PINN-GMM (no-encoder)      & \rv{(49.5, 50.7)} & \rv{(100, 100)} & N/A & \rv{3.55$^*$} \\
      & \textbf{PINN-GMM (proposed)$^\dagger$}      & \textbf{\rv{(0.16, 7.43)}}    & \textbf{\rv{(0.00, 7.91)}}    & \textbf{\rv{(4.58, 4.73)}} & \textbf{\rv{7229}}  \\
      \bottomrule
    \end{tabular}
    \label{tab:main_result}
\end{table*}

\rv{Table~\ref{tab:main_result} summarizes the orbital cases using the (minimum, maximum) range of each error measure over the evaluation horizon. This range-based summary captures how the uncertainty-propagation error varies over time. Across the four orbital cases, the proposed PINN-GMM keeps the maximum $\widetilde{\mathrm{TV}}$ below \(9.0\%\) and the maximum $\widetilde{\mathrm{WNE}}$ below \(7.9\%\). For $\widetilde{\mathrm{RD}}$, smaller values indicate better likelihood-based agreement with the reference samples; the proposed PINN-GMM also maintains consistently low and narrow RD ranges, whereas the baselines and naive PINN variants can exhibit substantially larger RD values, especially in the 6D rotating-spherical and 4D low-thrust cases.}

\rv{Table~\ref{tab:main_result} also reports three PINN ablations to isolate the roles of the density representation, sampling strategy, and encoder. The first, \emph{PINN-MLP (vanilla)}, uses a standard MLP for \(\hat p\) with uniform sampling; the second, \emph{PINN-GMM (vanilla)}, uses the same GMM density representation as the proposed method but still relies on uniform sampling; and the third, \emph{PINN-GMM (no-encoder)}, uses the combined sampling strategy but learns directly in \(x\)-space. The results show that the GMM representation alone is not sufficient: in higher-dimensional problems, uniform samples provide poor coverage under practical sampling budgets, leading to unstable training for both vanilla variants. The degradation of \emph{PINN-GMM (no-encoder)} further shows that the encoder is important for conditioning the training problem. The asterisks in Table~\ref{tab:main_result} mark runs whose losses plateau before convergence; this behavior occurs frequently for \emph{PINN-MLP (vanilla)}, \emph{PINN-GMM (vanilla)}, and \emph{PINN-GMM (no-encoder)}, indicating that these naive variants are harder to train reliably than the proposed PINN-GMM. In particular, the extremely large \(\widetilde{\mathrm{TV}}\) of \emph{PINN-MLP (vanilla)} in the 6D Keplerian rotating-spherical-coordinate case suggests that, without a PDF-valid representation, a neural-network approximation can produce probability estimates that are far from the true values.}

\rv{The offline computation time in Table~\ref{tab:main_result} should be interpreted separately from online evaluation cost. PINN-GMM is more expensive to construct than GA and UT because it requires offline training, while GMM also incurs initial mixture fitting and propagation cost. However, once trained, PINN-GMM returns the GMM parameters at any queried time through a single forward pass, which takes approximately 0.16 milliseconds on average in our implementation~\cite{PINN_Error:github}}

\section{Discussion}
The current framework relies on a well-posed Fokker-Planck equation and thus assumes continuously differentiable dynamics.
It does not apply directly to systems with impulsive effects.
For example, if an impulse thrust is applied at \rv{some time $t_{\mathrm{imp}}$, where $0 < t_{\mathrm{imp}} < T_{\mathrm{end}}$,} one cannot propagate uncertainty over $[0,T_{\mathrm{end}}]$ in a single Fokker--Planck formulation.
\rv{An operational limitation is that the PINN-GMM approximation error, and the associated learned error bound, may grow over longer horizons. Extending the method to one-day or multi-day orbital uncertainty propagation while maintaining informative error bounds is an important direction for future work.}

\rv{Computation speed remains a} key area for future investigation.
For the orbital experiments considered here, our approach currently takes longer (on CPU) than the baseline uncertainty-propagation methods, but remains faster than high-fidelity Monte Carlo.
We did not explore more advanced PINN architectures (e.g., transformer-based models) or training techniques (e.g., adaptive loss weighting, GPU parallelization); incorporating them could further reduce runtime.
To facilitate such developments, we release our code on
GitHub~\cite{PINN_Error:github} so that others can test alternative designs and push efficiency and accuracy further.

\section{Conclusion}
We present a physics-informed neural network (PINN) framework for uncertainty propagation that learns a space-time solution of the Fokker-Planck equation along with time-dependent worst-case error bounds \rv{that are theoretically grounded and empirically validated on the considered case studies.}
We propose a differentiable, time-conditioned Physics-Informed Gaussian-Mixture Model (PINN-GMM) that enforces PDF properties by construction. 
The PINN-GMM and its learned error bounds define an ambiguity set that, \rv{when the error PINN is sufficiently accurate, contains the true PDF and} yields rigorous event-probability intervals via tractable linear programs.
Across higher-dimensional orbital benchmarks, PINN-GMM provides accurate \rv{hours-scale uncertainty propagation} with \rv{informative error quantification} and improved robustness \rv{relative to} classical baselines and naive PINN variants.
\rv{Although the propagation horizons considered here are shorter than 
those of some orbital uncertainty studies, the ability to 
bound the propagated-density error is critical---uncertainty models that are inaccurate at unknown levels may lead to significant adverse consequences that are difficult to quantify.}

\clearpage            
\appendix
\section*{Appendix}
\section{Keplerian Dynamics in Orbit Elements}\label{appendix:oe_dyn}

\paragraph{Equinoctial orbit elements}
Let $(a,e,i,\Omega,\omega,M)$ denote the classical Keplerian elements (mean anomaly $M$). The equinoctial elements are
$[a,P_1,P_2,Q_1,Q_2,\lambda]^\top$, where
$P_1=e\sin(\omega+\Omega),
P_2=e\cos(\omega+\Omega),
Q_1=\tan(i/2)\sin\Omega,
Q_2=\tan(i/2)\cos\Omega,$ and $\lambda=M+\omega+\Omega$.
Let $\mathbf x:=[a,P_1,P_2,Q_1,Q_2,\lambda]$ be the state of a nominal Keplerian orbit governed by $\mathrm d\mathbf x(t)=\mathbf f(\mathbf x(t))\,\mathrm dt$, where $
\mathbf f(\mathbf x(t))=\big[0,0,0,0,0,\sqrt{\mu/a(t)^3}\big]^\top$.
This dynamical system induces the flow $\Phi(t;t_0)$, with
$
\Phi(t;t_0)=\mathbf x(t_0)+\Big[0,0,0,0,0,\sqrt{\mu/a^3}\,(t-t_0)\Big]^\top
$, and the Fokker-Planck solution for $t \geq t_0$ is the pushforward: $p(x,t)=p_0\!\big(\Phi^{-1}(t;t_0,x)\big)$ where $p_0$ denotes the PDF at initial time $t_0$.

\section{Constructing Reference `True' PDF by Fitting Gaussian Mixture}\label{appendix:fitgmm}

Let $\{x_i(t)\}_{i=1}^N$ denote $N$ Monte Carlo samples of the stochastic differential equation at time $t \in T$.
At each time \rv{$t$, when no analytical reference density is available, we fit a Gaussian mixture density to the propagated MC samples $\{x_i(t)\}_{i=1}^N$ using the expectation-maximization algorithm in the scikit-learn package~\cite{pedregosa2011scikit}. 
As discussed in Sec.~\ref{subsec:ref_true_pdf}, this fitted density serves only as a post-processing reference estimator, never as training data or exact ground truth.}
The best fit is obtained by searching over a list of candidate numbers of mixture components $\{1,2,4,8,16\}$.

\section{Baseline Uncertainty Propagation Methods}\label{appendix:baselines}

\paragraph{Gaussian Approximation}
Gaussian Approximation (GA) propagates the state mean \rv{through the} deterministic nonlinear dynamics \rv{and the covariance through dynamics linearized about the mean.
For the SDE in Eq.~\eqref{eq:sde_general}, we integrate
$
\dot{\bar x}=f(\bar x,t)$ and $
\dot P_x = F P_x + P_x F^T + g(\bar x,t)g(\bar x,t)^T,
$ where $F=\partial f/\partial x|_{x=\bar x(t)}$; see~\citet[Algorithm 9.4]{sarkka2019applied}.
The ODEs are solved using SciPy's \texttt{solve\_ivp}~\citep{2020SciPy-NMeth}, and the resulting density is $N(x;\bar x(t),P_x(t))$.}

\paragraph{Unscented Transform}
Unscented Transform (UT) propagates a set of deterministic sigma points over time, and compute the statistics of these sigma points to obtain the means and covariances~\cite{wan2000unscented,julier2004unscented}.
Our implementation \rv{follows} ~\citet[Algorithm 9.5]{sarkka2019applied}, \rv{using} the scaled unscented transform with a tunable scaling parameter $\alpha \geq 0$, usually chosen according to some rules of thumb~\citep{julier2002scaled}; we use $\alpha=10^{-3}$ by default.

\paragraph{Gaussian Mixture Model}
A simple uncertainty propagation method using Gaussian mixture is implemented by: (i) fit an initial Gaussian mixture to the initial PDF $p_0(x)$ by the same method in Appendix~\ref{appendix:fitgmm} with more candidate number of mixture components $\{10, 20, 30, 40\}$; (ii) propagate each mixture by the Gaussian approximation (GA) method to any time $t\in T$, and combine them to obtain the final Gaussian mixture density at each time.

\section{Orbit Experiments}

\paragraph{6D Unperturbed Keplerian Orbit in Equinoctial Orbit Element Space}\label{appendix:6d_equin}
The state and the dynamics are provided in Sec.~\ref{subsubsec:oec}.
The initial PDF is a Gaussian distribution with mean $\mu_0 =  [4.2245700 \times 10^{7} \text{ meters}, 7.7893469\times 10^{-7}, 3.1708085\times 10^{-4}, 2.7199831\times 10^{-6}, -1.0321360\times 10^{-6}, 1.5750291\times 10^{-5}]$ and covariance $\Sigma_0 = \mathrm{diag}([2.9833938\times 10^{12}, 1.0585866\times 10^{-4}, 1.2572521\times 10^{-3},2.4929550\times 10^{-6}, 4.7029530\times 10^{-6}, 5.2294909\times 10^{-4}]).$
The state domain is $X =  [-4.0507,   5.9977] \times [-4.7033,   4.8175] \times [-4.4900,   4.9228] \times [-5.0428,   4.7609] \times [-4.7474,   4.9450] \times [-10.9330, 118.0678].$
We use $K_{\mathrm{GMM}}=5$ for the PINN-GMM.
\rv{This choice reflects the relatively structured density evolution in equinoctial elements, where only $\lambda$ varies under the unperturbed Keplerian flow.}
The training time of PINN-GMM $\hat p$ over $k_{\max}=60000$ iterations is 7317 seconds\rv{; its initial-condition loss and FP-PDE residual loss are $6.2 \times 10^{-5}$ and $6.4 \times 10^{-4}$, respectively.}
The training time of $\hat e_1$ over $k_{\max}=40000$ iterations is 8347 seconds.
\paragraph{6D Unperturbed Keplerian Orbit in Rotating Spherical Space}\label{appendix:6d_sph}
In this experiment (adapted from~\cite{sun2016uncertainty}), the state and the dynamics are defined in~\cite[Eq.~16]{sun2016uncertainty} without $J_2$ and diffusion noise: $J_2 = 0, \mathbf{G}(\mathbf{x})=0$.
The scaling constants of the rotating spherical coordinates are
$R=2 \times 10^6$ meter, $T_{\mathrm{orb}}=86164.695$ seconds, $\Theta=0.015$ rad, $\Phi=0.0387$ rad, and $W=7.2920644\times 10^{-5}$ rad/s.
The initial PDF is a Gaussian distribution with mean $\mu_0 =  [\frac{4.2164 \times 10^7}{R}, \frac{\pi/2}{\Theta}, \frac{0}{\Phi}, \frac{0}{R/T_{\mathrm{orb}}}, \frac{0}{\Theta/T_{\mathrm{orb}}} \frac{7.2921 \times 10^{-5} - W}{\Phi / T_{\mathrm{orb}}}]$ and covariance $\Sigma_0 = \mathrm{diag}([\frac{10^{11}}{R^2}, \frac{10^{-5}}{\Theta^2}, \frac{10^{-4}}{\Phi^2}, \frac{10^3}{(R/T_{\mathrm{orb}})^2}, \frac{10^{-13}}{(\Theta/T_{\mathrm{orb}})^2}, \frac{10^{-12}}{(\Phi/T_{\mathrm{orb}})^2}]).$
The state domain is $X =  [15, 27] \times [102, 107], \times [-9, 12] \times [-25, 25] \times [-10, 10] \times [-50, 100].$
$10^7$ samples from semi-analytical Monte Carlo simulation are used to construct the reference ``true'' PDF.
Sine the dynamics is more complicated, we use $K_{\mathrm{GMM}}=11$ for the PINN-GMM.
\rv{The larger $K_{\mathrm{GMM}}$ reflects the more complex density evolution in rotating spherical coordinates, where all state components couple nonlinearly.}
The training time of PINN-GMM $\hat p$ over $k_{\max}=60000$ iterations is 13074 seconds\rv{; its initial-condition loss and FP-PDE residual loss are $6.0 \times 10^{-4}$ and $4.5 \times 10^{-3}$, respectively.} The training time of $\hat e_1$ over $k_{\max}=40000$ iterations is 8816 seconds.
\paragraph{6D Unperturbed Keplerian Orbit in Cartesian Space}\label{appendix:6d_cart}
The Cartesian state is $\mathbf{x}=[x,y,z,v_x,v_y,v_z]$, with the standard Keplerian dynamics: $
    f_1(\mathbf{x})=v_x, f_2(\mathbf{x})=v_y, f_3(\mathbf{x})=v_z,
    f_4(\mathbf{x})=-\mu\,\frac{x}{r^3}, f_5(\mathbf{x})=-\mu\,\frac{y}{r^3}, f_6(\mathbf{x})=-\mu\,\frac{z}{r^3}.$
To define the initial distribution in Cartesian space, we sample from the rotating spherical coordinate's initial distribution, map those samples to Cartesian coordinates, compute the resulting Cartesian mean and variances, and then set $p_0$ to a diagonal Gaussian with the corresponding mean $\mu_0 =  [4.2158868 \times 10^{7}, -2.8649365\times 10^{3}, -6.4722394\times 10^{2}, 6.9467247\times 10^{-2}, 3.0742595\times 10^{3}, 3.1288903\times 10^{-2}]$ and covariance $\Sigma_0 = \mathrm{diag}([9.9124732\times 10^{10}, 1.7442639\times 10^{11}, 1.7627195\times 10^{10}, 1.9275731\times 10^{3}, 2.3256350\times 10^{3}, 1.8096840\times 10^{2}])$ in SI units.
The state domain is $X = [-3.0\times 10^{7}, 5.0\times 10^{7}] \times [-2.0\times 10^{6}, 5.5\times 10^{7}], \times [-8\times 10^{5}, 9\times 10^{5}] \times  [-4000.0, 250.0] \times [-4000.0,  4000.0] \times [-100.0,  100.0]$.
We use the same PINN-GMM architecture for $\hat p$ and the same training hyperparameters as in the previous experiment.
Similar to the previous experiments, we normalize the Cartesian coordinates by $x'=x/R, y'=y/R, z'=z/R, v_x'=v_x/(R/T), v_y'=v_y/(R/T), v_z'=v_z/(R/T)$ during training for numerical stability. \rv{Since the training of $\hat p$ did not converge over 60000 iterations---with its initial-condition loss and FP-PDE residual loss plateauing at $2.9 \times 10^{-2}$ and $2.1 \times 10^{-2}$, respectively---we did not train its error PINN.}

\paragraph{4D Perturbed Keplerian Orbit with Diffusion in Rotating Spherical Space}\label{appendix:4d_sph}
In this experiment (adapted from~\cite[Eq. 16]{sun2016uncertainty}), the four-dimensional state are defined by setting
$
    \mathbf{x} := [ r', \phi', v_r', v_\phi'],\;\theta' = \pi/2, v_{\theta}' = 0,
$
The standard $J_2=1.0826\times10^{-3}$ coefficient is used.
The initial PDF is a Gaussian distribution with mean $\mu_0 =  [\frac{4.2164 \times 10^7}{R}, \frac{0}{\Phi}, \frac{0}{R/T_{\mathrm{orb}}}, \frac{7.2921 \times 10^{-5} - W}{\Phi / T_{\mathrm{orb}}}]$ and covariance $\Sigma_0 = \mathrm{diag}([\frac{10^{11}}{R^2}, \frac{10^{-4}}{\Phi^2}, \frac{10^3}{(R/T_{\mathrm{orb}})^2}, \frac{10^{-12}}{(\Phi/T_{\mathrm{orb}})^2}]).$
The state domain is $X =  [18, 24] \times [-2, 3] \times [-20, 20] \times [-22, 32].$
Due to the computation limit, only $10^6$ samples from a full numerical Monte Carlo simulation are used to construct the reference ``true'' PDF.
\rv{The same $K_{\mathrm{GMM}}=11$ is used as in the rotating-spherical 6D case, as the 4D dynamics involve similar nonlinear coupling under $J_2$ perturbation and diffusion.}
\rv{The initial-condition loss and FP-PDE residual loss for $\hat p$ are $2.5 \times 10^{-4}$ and $1.8 \times 10^{-3}$, respectively.}
\rv{\paragraph{4D Low-Thrust}\label{appendix:4d_lowthrust}
For the 4D low-thrust case, the state domain is enlarged to \(X=[17,30]\times[-8,8]\times[-15,40]\times[-55,65]\). 
To remove initial-condition fitting error, the PINN-GMM with $K_{\mathrm{GMM}}=11$ is parameterized in whitened coordinates with time-gated mean and Cholesky-factor deviations, so that at \(t=0\) all mixture components reduce to the same Gaussian density \(\mathcal N(x; \mu_0,\Sigma_0)\), independent of the mixture weights.
Thus, the initial-condition fitting error is eliminated up to numerical precision, and the final FP-PDE residual loss for PINN-GMM $\hat p$ is $6.6 \times 10^{-3}$.}

\section{Neural Network Architectures for All Orbit Experiments}\label{appendix:nn}
\begin{table}[htb!]
    \centering
    \caption{PINN-GMM $\hat{p}$ architecture of each weight, mean, and covariance parameter for all the orbit experiments.}
    \begin{tabular}{lccc}
        \toprule
        \textbf{Layer Connection} & \textbf{Type} & \textbf{\# Neurons} & \textbf{Activation Function} \\
        \midrule
        Input Layer $\rightarrow$ Hidden Layer 1  & Fully Connected & 64  & Softplus \\
        Hidden Layer 1 $\rightarrow$ Hidden Layer 2  & Fully Connected & 64  & Softplus \\
        Hidden Layer 2 $\rightarrow$ Hidden Layer 3  & Fully Connected & 64  & Softplus \\
        Hidden Layer 3 $\rightarrow$ Output Layer  & Fully Connected & output-dimension  & N/A \\
        \bottomrule
    \end{tabular}
    \label{tab:nn_pinngmm_summary}
\end{table}
\begin{table}[htb!]
    \centering
    \caption{Error network $\hat e$ architecture for all the orbit experiments.}
    \begin{tabular}{lccc}
        \toprule
        \textbf{Layer Connection} & \textbf{Type} & \textbf{\# Neurons} & \textbf{Activation Function} \\
        \midrule
        Input Layer $\rightarrow$ Hidden Layer 
            & Fully Connected 
            & 64 
            & SiLU \\
        Hidden Layer $\rightarrow$ Hidden Layer 
            & 8$\times$ Residual Blocks\footnotemark 
            & 64 
            & SiLU (with LayerNorm) \\
        Hidden Layer $\rightarrow$ Output Layer 
            & Fully Connected 
            & 1 
            & N/A \\
        \bottomrule
    \end{tabular}
    \label{tab:nn_enetxl_summary}
\end{table}
\footnotetext{Each residual block is a 2-layer module with SiLU activations, optional LayerNorm, and a skip connection. A skip connection that also uses the normalized inputs is added after the 4th block.}


\section*{Funding Sources}
This work was supported in part by the Air Force Research Laboratory/RVSW under Contract No. FA945324CX026.


\section*{Acknowledgments}
Generative AI was used to edit grammar and improve readability; no new results, analyses, or citations were generated by the tool.
\bibliography{sample}

\end{document}